\documentclass[arxiv,preprint]{imsart}

\RequirePackage[OT1]{fontenc}
\RequirePackage{amsthm, amsmath, amssymb}
\usepackage{natbib}
\RequirePackage[bookmarks=false]{hyperref}

\usepackage{graphicx}
\usepackage{enumerate}
\usepackage{multirow}
\usepackage{supertabular}

\newcommand{\by}{\mathbf{y}}
\newcommand{\bE}{\mathbf{E}}
\newcommand{\bV}{\mathbf{Var}}
\newcommand{\cov}{\mathbf{cov}}
\newcommand{\bP}{\mathbf{P}}
\newcommand{\bQ}{\mathbf{Q}}

\newcommand{\bfPB}{\bP_\mathtt{B}}
\newcommand{\bfEB}{\bE_\mathtt{B}}
\newcommand{\VB}{\mathbf{Var}_\mathtt{B}}

\newcommand{\bfE}{\mathbf{E}}

\startlocaldefs
\numberwithin{equation}{section}
\theoremstyle{plain}
\newtheorem{theorem}{Theorem}[section]
\newtheorem{lemma}[theorem]{Lemma}
\newtheorem{prop}[theorem]{Proposition}
\newtheorem{remark}[theorem]{Remark}
\newtheorem{assumption}[theorem]{assumption}
\newtheorem{condition}{Condition}

\endlocaldefs

\graphicspath{{image/} {image/thirdmoment/} {image/author/} {image/cellphone/} {image/graphs/} {image/coverage/}}

\begin{document}

\begin{frontmatter}
\title{Graph-Based Change-Point Detection} 
\runtitle{Graph-Based Change-Point Detection}

\begin{aug}
\author{\fnms{Hao} \snm{Chen}\thanksref{m1} \ead[label=e1]{hxchen@ucdavis.edu}}
\and
\author{\fnms{Nancy} \snm{Zhang} \thanksref{m2} \ead[label=e2]{nzh@wharton.upenn.edu}}

\runauthor{Hao Chen and Nancy Zhang}

\affiliation{\thanksmark{m1} Department of Statistics, University of California, Davis \\
\thanksmark{m2} Department of Statistics, The Wharton School, University of Pennsylvania}

\address{Department of Statistics \\
University of California, Davis \\
One Shields Avenue \\
Davis, CA 95616 \\
USA \\
\printead{e1}\\
\phantom{E-mail:\ }}

\address{Department of Statistics \\
The Wharton School \\
University of Pennsylvania \\
Philadelphia, PA 19104 \\
USA \\
\printead{e2}}
\end{aug}

\begin{abstract}
We consider the testing and estimation of change-points -- locations where the distribution abruptly changes -- in a data sequence. A new approach, based on scan statistics utilizing graphs representing the similarity between observations, is proposed.  The graph-based approach is non-parametric, and can be applied to any data set as long as an informative similarity measure on the sample space can be defined.  Accurate analytic approximations to the significance of graph-based scan statistics for both the single change-point and the changed interval alternatives are provided.  Simulations reveal that the new approach has better power than existing approaches when the dimension of the data is moderate to high.  The new approach is illustrated on two applications: The determination of authorship of a classic novel, and the detection of change in a network over time.


\end{abstract}



\end{frontmatter}

\section{Introduction}
\label{sec:introduction}

Change-point models are widely used in various fields for detecting lack of homogeneity in a sequence of observations.  In the typical formulation, the observations $\{y_i:~i=1,2,\dots,n\}$ are assumed to have distribution $F_0$ for $i \leq \tau$ and possibly a different distribution $F_1$ for $i>\tau$.  The parameter $\tau$ is referred to as the change-point.  We consider the case where the total length of the sequence $n$ is fixed.  There is a rich literature on theory and applications of this model when $y_i$ are real or integer valued scalars.  For example, in a well known study of the annual flow volume of the Nile River at the city of Aswan, Egypt, from 1871 to 1970, each $y_i$ is a continuous measurement of the annual discharge from the river \citep{cobb1978problem}, and the goal is to detect shifts in flow volume.  If the distribution of $y_i$ were assumed to be normal, score- or likelihood- based tests can be applied \citep{james1987tests}.  Bayesian and non-parametric approaches have also been developed (see \cite{carlstein1994change} for a survey).

Modern statistical applications are faced with data of increasing richness and dimension.  High throughput measurement schemes and digitization in many scientific fields have produced data sequences $\{\mathbf{y}_i:~i=1,2,\dots,n\}$, where each $\mathbf{y}_i$ is a high dimensional vector or even a non-Euclidean data object.   The dimension of each observation can be larger than the length of the sequence.  Testing the homogeneity of such high dimensional sequences is a challenging but important problem.  Following are some motivating examples:

\begin{description}
\item[Network evolution:] Data on networks have become increasingly common.  For example, email, phone, and online chat records can be used to construct a network of social interactions among individuals \citep{kossinets2006empirical, eagle2009inferring}.  High throughput biological experiments have led to the ubiquitous study of protein- or gene- interaction networks.  A large part of these studies is characterizing how the network evolves through time.  Here, the observation at each time point is a graphical encoding of the network.  In a longitudinal study, one might ask whether there is an abrupt shift in network connectivity at any point in time.  
\item[Image analysis:] Image data collected through time appears in diverse applications, from video surveillance to climatology to neuroscience.  The detection of abrupt events, such as security breaches, storms, or brain activity, can be formulated as a change-point problem.   Here, the observation at each time point is the digital encoding of an image.  
\item[Text or sequence analysis:]  Many classic works in both western and eastern literature have ongoing authorship debates.  For example, the debate surrounding both \textit{Tirant lo Blanc}, a Catalan romance, and \textit{Dream of the Red Chamber}, a Chinese masterpiece, is whether there is a change of authorship mid-way through the novel.  In the digital era, an objective approach to these debates is to statistically test for abrupt changes in writing style, which can be reflected by word usage.  Similar problems arise in genomic sequence analysis in biology, where it is often of interest to find regions of the genome with different DNA-word compositions (see, for example, \citet{tsirigos2005new}).  In both settings, each observation in the sequence is a vector of word counts over a large dictionary of words.

\end{description}

In all of these examples, the problem can be given the following statistical formulation:  We observe a sequence of observations $\{\by_i\}, i=1,\dots,n$, indexed by some meaningful ordering, such as time or location.  We are concerned with testing the null hypothesis
\begin{align}
  H_0: & \  \by_i \sim F_0,\ i=1,\dots,n, \label{H0}
\end{align}
against the single change-point alternative
\begin{align}
  H_1: & \ \exists ~ 1\leq\tau <n, \ \by_i \sim \left\{ \begin{array}{ll} F_1, & i> \tau \\ F_0, & \text{otherwise}, \end{array} \right. \label{H1one}
\end{align}
or the changed interval alternative
\begin{align}
  H_2: & \ \exists ~ 1 \leq \tau_1<\tau_2 \leq n, \ \by_i \sim \left\{ \begin{array}{ll} F_1, & i= \tau_1+1,\dots,\tau_2 \\ F_0, & \text{otherwise}, \end{array} \right. \label{H1two}
\end{align}
where $F_0$ and $F_1$ are two probability measures that differ on a set of non-zero measure.
Scenarios with multiple change-points can be decomposed into these two types of simple alternatives.

We study this change-point problem under the assumption that $\{\by_i\}$ are \emph{independent}.  
Independence is an ideal assumption that may be violated in some settings.  However, this assumption allows us to conduct theoretical analysis, which also produce results that are useful when the assumption is slightly violated. 
We later discuss modifications to our approach when the independence assumption is violated.

In the multivariate setting, existing approaches
are limited in many ways.
Most methods are based on parametric models that are highly context specific.  For example, \cite{zhang2010detecting} and \cite{siegmund2011detecting} studied the problem of detecting common shifts in mean in sequences of independent multivariate Gaussian variables with identity covariance. Under the same setting, \cite{srivastava1986likelihood} and \cite{james1992tests} discussed general likelihood ratio tests for a change in mean, which requires that the dimension of the observations be smaller than the number of observations.  As we will show in simulations, parametric change-point tests for multivariate data work under very specific assumptions, and are sensitive to violation of these assumptions.  The existing parametric tests also can not be applied in very high dimensions, unless strong assumptions are made to avoid the estimation of the large number of nuisance parameters that are a by-product of increasing dimension.

In the nonparametric context, \cite{desobry2005online} and \cite{harchaoui2009kernel} used kernel-based methods.  A common drawback for kernel-based methods is that they rely heavily on the choice of the kernel function and its parameters, and the problem becomes more severe when the data is in moderate to high dimensions.  Also, none of these methods offer a fast analytical formula for false positive control, thus making them difficult to apply for large data sets.  \cite{lung2011homogeneity} proposed a non-parametric approach based on marginal rank statistics, which is useful if there is a clear ranking mechanism, but also requires the restriction that the number of observations be larger than the dimension of the data.  



In this paper, we describe a nonparametric approach to change-point detection and estimation.  The approach can be applied to data in arbitrary dimension and even to non-Euclidean data, with a general, analytic formula for type I error control.  We illustrate the approach on two applications:  Testing for a change in author of a classic European novel, and testing the temporal homogeneity of a social network.  We show, via simulations, that as dimension increases this nonparametric method gains power over parametric methods in cases where the parametric methods can be applied.  The generality of the new approach and the availability of analytic formulas for type I error make it an easy off-the-shelf tool for homogeneity testing in multivariate settings.  The method is implemented in an R package ``gSeg'', which is available in CRAN.

This paper is organized as follows:  In Section \ref{sec:gener-fram-change} we describe the proposed method.  The underlying idea is graph-based two-sample tests adapted to the scan-statistic setting.  Two-sample tests based on various types of graphs representing the similarity between observations were first proposed in \cite{friedman1979multivariate} and \cite{rosenbaum2005exact}.    We review these previous works in Section \ref{sec:graph-based}.  Once the graph has been constructed, theoretical analysis of the scan statistic can be decoupled from the modeling of the high dimensional data.  We describe the test statistic in the detection of a single change-point in Section \ref{sec:single-change-point}, and that in the detection of a changed interval in Section \ref{sec:changed-interval}.  Section \ref{sec:an-analyt-appr} gives analytic formulas for approximating the significance of the tests, and evaluates their accuracy in numerical studies.  Section \ref{sec:power-analysis} evaluates the power of the test via simulations.  In Section \ref{sec:real-data-examples}, the new method is applied to the analysis of the text of \textit{Tirant lo Blanc}, and the analysis of the Friendship Network data set collected by the MIT Media Laboratory \citep{eagle2009inferring}.  In Section \ref{sec:extensions}, we discuss some extensions to the approach to deal with local dependency in the sequence and to construct a confidence interval to the change-point.  Finally, we conclude with a discussion in Section \ref{sec:disc-concl}.

\section{A Graph-Based Framework for Change-Point Detection}
\label{sec:gener-fram-change}


In both the single change-point \eqref{H1one} and the changed interval alternatives \eqref{H1two}, the observations are partitioned into two groups.  We allow each group to have a minimum number of observations: $1<n_0\leq \tau \leq n_1 <n$ for the single change-point scenario and $1<l_0\leq \tau_2-\tau_1 \leq l_1 <n$ for the changed interval scenario, where $n_0$, $n_1$, $l_0$, $l_1$ are prespecified.  Sometimes, these values can be better chosen using domain knowledge.  We may also have some further constrains on the locations of $\tau_1$ and $\tau_2$.


We do not impose any restrictions on the sample space or distribution of $\by_i$.  Our approach requires that the similarity between $\by_i$ can be represented  by a graph, with edges in the graph connecting observations that are ``close'' in some sense.  For the proposed method to have good power, data points drawn from $F_0$ need to be closer to each other than to data points drawn from $F_1$, in a global sense, and vice versa.  We describe this in more detail next, and briefly review graph-based two-sample tests.


\subsection{Graph-Based Two-Sample Tests}
\label{sec:graph-based}

By graph-based tests, we refer to tests that are based on graphs with the observations $\{\by_i\}$ as nodes.  The graph is usually derived from a distance or a generalized dissimilarity on the sample space, with edges connecting observations that are close in distance.  For example, \cite{friedman1979multivariate} proposed the first graph-based test for testing the null hypothesis that subjects from two groups are equal in distribution against an omnibus alternative.  Their method relies on the minimum spanning tree (MST), which is a tree connecting all observations minimizing the total distance across edges.  Their test statistic is the number of edges in the tree connecting observations from different groups, rejecting the null hypothesis when this count is low compared to its distribution under permutation.  The rationale is that, if the two groups come from different distributions, data points from the same group should be closer to each other, and thus edges in the tree should be more likely to connect subjects within a group.  

There are many other ways to construct the graph.  \cite{rosenbaum2005exact} proposed minimum distance pairing (MDP), which divides the $n$ subjects into $n/2$ (assuming $n$ is even) non-overlapping pairs in such a way as to minimize the total of $n/2$ distances between pairs.  For odd $N$, Rosenbaum suggested creating a pseudo data point that has distance 0 with all other subjects, and later discarding the pair containing this pseudo point.  This method has the desirable property of being truly distribution free. 

The nearest neighbor graph (NNG), which connects each data point to its nearest neighbor, can also be used to define a statistic in similar style to \cite{friedman1979multivariate} and \cite{rosenbaum2005exact}.

Figure \ref{fig:graph} illustrates the MST, MDP, and NNG  on 40 points in $\mathbb{R}^2$.  Ways to construct the graph are not limited to these three.  In some applications, the graph may be given at the start of the analysis without alluding to an underlying distance measure, see the Haplotype example in \cite{chen2013graph}.  The proposed method does not depend on how the graph was constructed.  The test statistic and its properties under the permutation null rely only on the graph and not on the underlying distance measure nor on the original data.  However, the quality of the graph in separating $F_0$ and $F_1$ is integral to the power of the test.



\begin{figure}[]
  \centering
  \includegraphics[width=.32\textwidth]{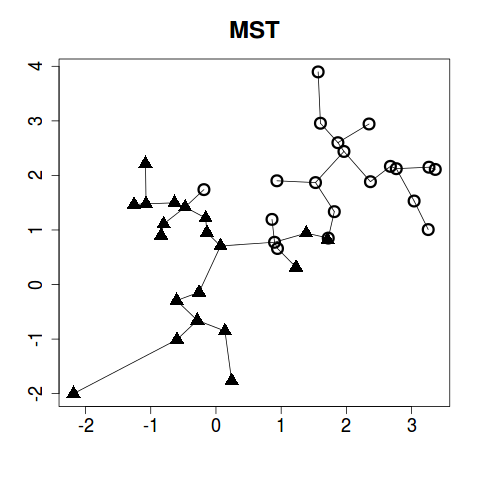}
  \includegraphics[width=.32\textwidth]{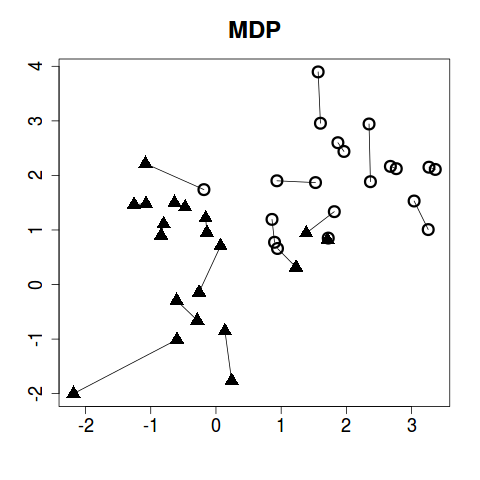}
  \includegraphics[width=.32\textwidth]{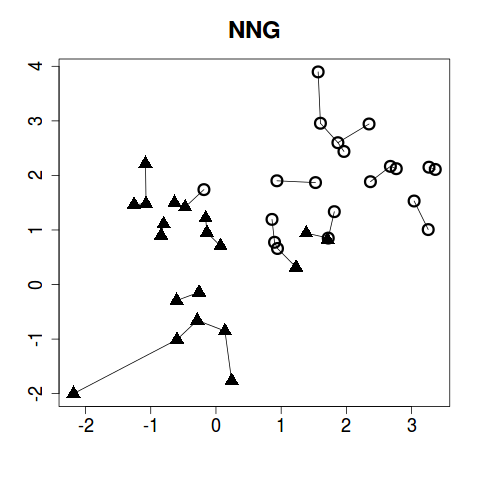}
  \caption{The MST, MDP and NNG graphs on an example two-dimensional data set. 20 points were drawn from $\mathcal{N}(\mathbf{0}, I_2)$ (shown in triangles) and 20 points were drawn from $\mathcal{N}((2,2)', I_2)$ (shown in circles).}
  \label{fig:graph}
\end{figure}

%



\subsection{Test Statistic for a Single Change-Point Alternative}
\label{sec:single-change-point}

Here, we derive the test statistic for testing the null $H_0$ (\ref{H0}) versus the single change-point alternative $H_1$ (\ref{H1one}).  Each possible value of $\tau$ divides the observations into two groups: Observations come before $\tau$ and observations that come after $\tau$.  Let $G$ be the similarity graph on $\{\by_i\}$, as described in Section \ref{sec:graph-based}.  We use $G$ to refer to both the graph and its set of edges when the vertex set is implicitly obvious.  For any event $x$ let $I_x$ be the indicator function that takes value 1 if $x$ is true, and 0 otherwise. Then, for any candidate value $t$ of $\tau$, the number of edges connecting points from different groups is
\begin{align*}
  R_G(t) & = \sum_{(i,j)\in G} I_{g_i(t) \neq g_j(t)}, \quad g_i(t) = I_{i>t}.
\end{align*}
Here, $g_i(t)$ is an indicator function for the event that $\by_i$ is observed after $t$.  So $R_G(t)$ is the number of edges in the graph $G$ that connect observations from the ``past'' $(\leq t)$ to the ``future'' $(>t)$.  Relatively small values of $R_G(t)$ are evidence against the null hypothesis.  

Figure \ref{fig:Rt} illustrates the computation of $R_G(t)$ on a small artificial data set of length $n=40$ with the first 20 points drawn from $\mathcal{N}(\mathbf{0},I_2)$ and the second 20 points drawn from $\mathcal{N}((2,2)^\prime, I_2)$.  The similarity graph $G$ is the MST constructed using Euclidean distance.

\begin{figure}[!htp]
  \centering
  \includegraphics[width=.32\textwidth]{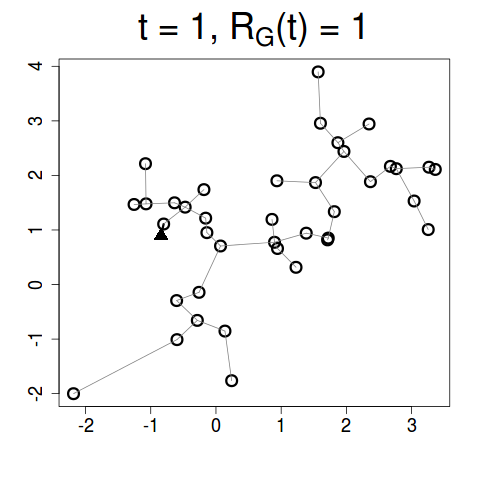}
  \includegraphics[width=.32\textwidth]{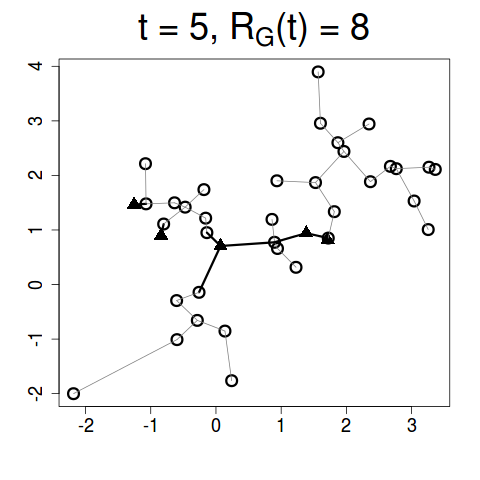}
  \includegraphics[width=.32\textwidth]{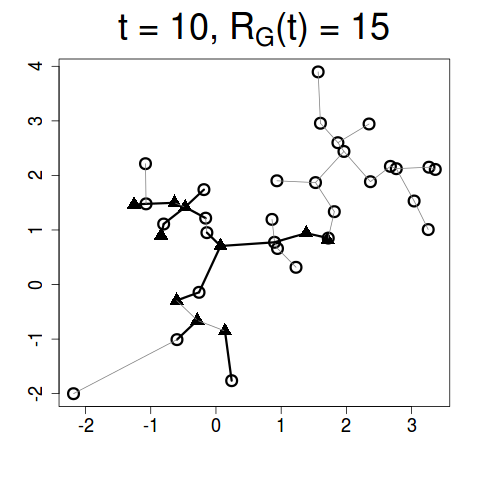} \\
  \includegraphics[width=.32\textwidth]{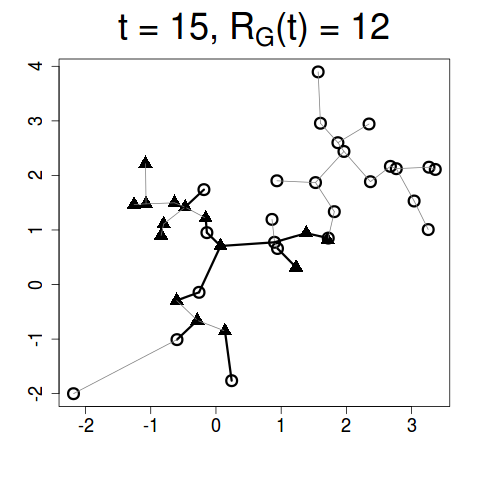}
  \includegraphics[width=.32\textwidth]{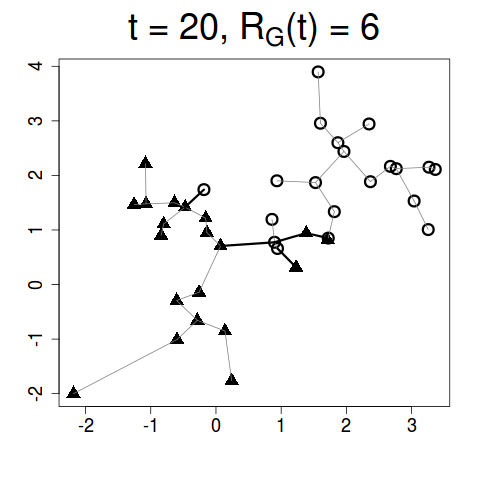}
  \includegraphics[width=.32\textwidth]{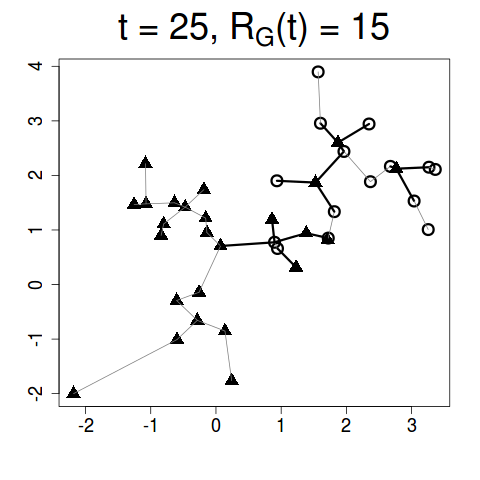} \\
  \includegraphics[width=.32\textwidth]{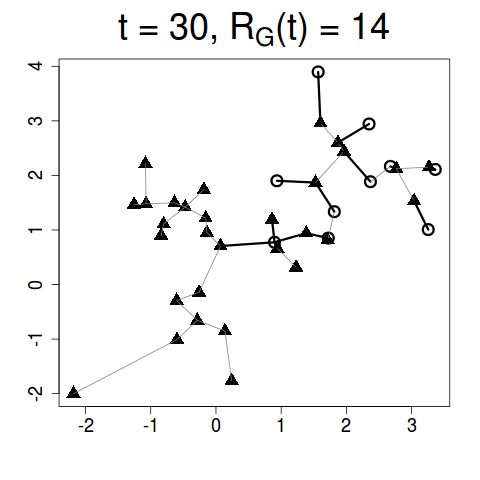}
  \includegraphics[width=.32\textwidth]{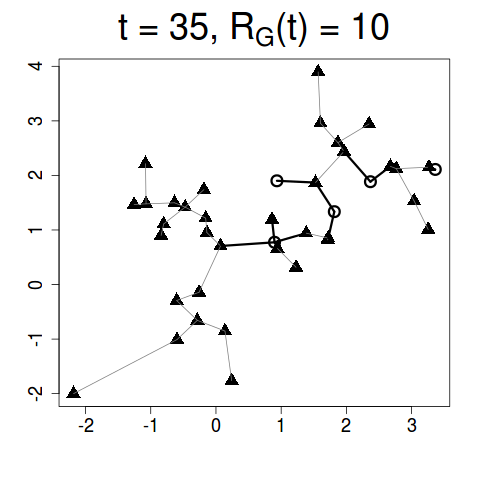}
  \includegraphics[width=.32\textwidth]{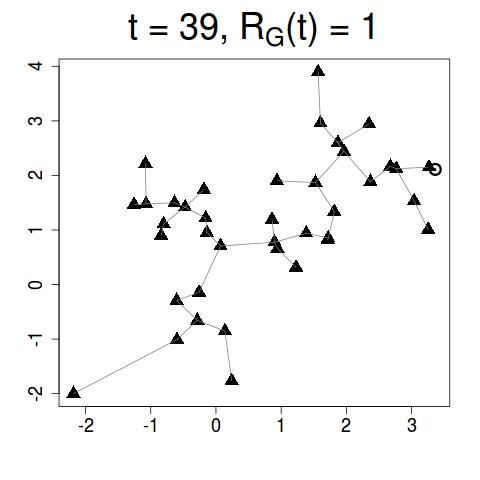}
  \caption{The computation of $R_G(t)$ for nine different values of $t$.  The data is a sequence of length $n=40$, with the first 20 points drawn from $\mathcal{N}(\mathbf{0}, I_2)$ and the second 20 points drawn from $\mathcal{N}((2,2)', I_2)$.  The similarity graph $G$ shown in the plots is the MST on Euclidean distance.  Each $t$ divides the observations into two groups, one group for observations before and at $t$ (shown as triangles) and the other group for observations after $t$ (shown as circles).  Edges that connect observations from the two different groups (i.e. edges connecting a triangle and a circle) are bold in the graph.  Notice that $G$ does not change as $t$ changes, but the group identities of some observations change, causing $R_G(t)$ to change. }
  \label{fig:Rt}
\end{figure}

Under the null hypothesis $H_0$ \eqref{H0} and the independence assumption, the joint distribution of $\{\by_i:~i=1,\dots,n\}$ is the same under the permutation distribution.  We define the null distribution of $R_G(t)$ to be the permutation distribution, which places $1/n!$ probability on each of the $n!$ permutations of $\{\by_i:~i=1,\dots,n\}$.  Let $\pi(i)$ be the time of observing $\by_i$ after permutation, then for the permuted sequence, $g_i(t)$ becomes $I_{\pi(i)>t}$.  Notice that the graph $G$ is determined by the values of $\by_i$'s, not their order of appearance, and thus remains constant under permutation.  When there is no further specification, we denote by $\bP$, $\bE$, $\bV$ probability, expectation, and variance, respectively, under the permutation null distribution.

Since the null distribution of $R_G(t)$ depends on $t$, we standardize $R_G(t)$ so that it is comparable across $t$.  Let
\begin{align} \label{Zt}
  Z_G(t) & = -\frac{R_G(t)-\bE[R_G(t)]}{\sqrt{\bV[R_G(t)]}}.
\end{align}
In the standardization, we also invert the sign, so that \emph{large} values of $Z_G(t)$ are evidence against the null.


Lemma \ref{lemma:onechptstat} below gives analytic formulas for $\bE[R_G(t)]$ and $\bV[R_G(t)]$.  Before we state the lemma, we introduce some new notation: Let $G_i$ be the subgraph of $G$ containing all edges that connect to node $\by_i$.  As before, we recycle notation and use $G_i$ to denote the set of edges in $G_i$.  $|G_i|$ denotes the number of edges in $G_i$, which is apparently also the degree of node $\by_i$ in $G$.

\begin{lemma}
\label{lemma:onechptstat}
  Under the permutation null, the expectation and variance of $R_G(t)$ are
  \begin{align*}
    \bE(R_G(t)) & = p_1(t) |G|, \\
    \bV(R_G(t)) & = p_2(t) |G| + \left( \frac{1}{2}p_1(t) - p_2(t)\right) \sum_i|G_i|^2 + \left( p_2(t) - p_1^2(t) \right) |G|^2,
  \end{align*}
where
\begin{align*}
  p_1(t) & = \frac{2t(n-t)}{n(n-1)},&
  p_2(t) & = \frac{4t(t-1)(n-t)(n-t-1)}{n(n-1)(n-2)(n-3)}.
\end{align*}
\end{lemma}

The expressions for the expectation and variance are obtained by combinatorial analysis and the details are in Appendix \ref{sec:proof-lemma}.

\begin{remark}
  The expectation and variance of $R_G(t)$ under the permutation null depend only on $t$, $n$, and two characteristics of the graph -- the number of edges ($|G|$) and the sum of squares of node degrees ($\sum_{i=1}^n|G_i|^2$).
\end{remark}

Figure \ref{fig:RZ} shows the $R_G(t)$ and $Z_G(t)$ processes for the same illustration data set in Figure \ref{fig:Rt}.  We see that $Z_G(t)$ peaks at the true change-point 20.  For contrast, Figure \ref{fig:RZ0} shows $R_G(t)$ and $Z_G(t)$ for a sequence of 40 points all drawn from $\mathcal{N}(\mathbf{0},I_2)$.  Note that for the latter data set, with no change-point, $Z_G(t)$ exhibits random fluctuation and attains a maximum value much smaller than that of Figure \ref{fig:RZ}.

\begin{figure}[!htp]
  \centering
  \includegraphics[width=.4\textwidth]{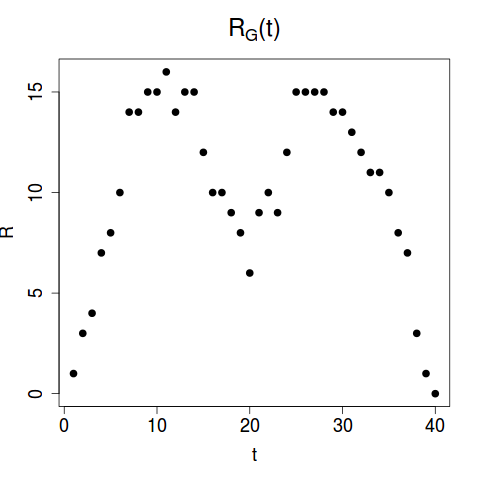}
  \includegraphics[width=.4\textwidth]{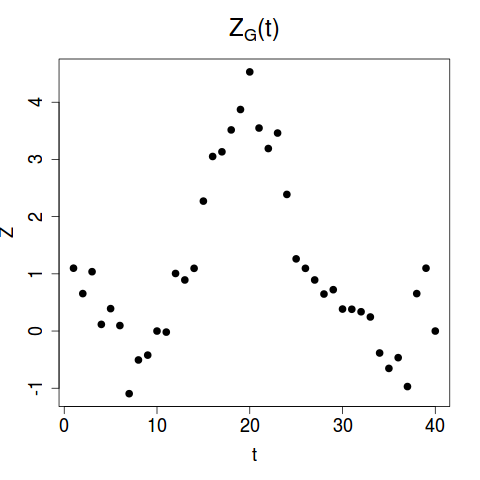}
  \caption{The profile of $R_G(t)$ and $Z_G(t)$ against $t$ for the same data set as in Figure \ref{fig:Rt}.  There is a change-point at $t=20$.}
  \label{fig:RZ}
\end{figure}

\begin{figure}[!htp]
  \centering
  \includegraphics[width=.4\textwidth]{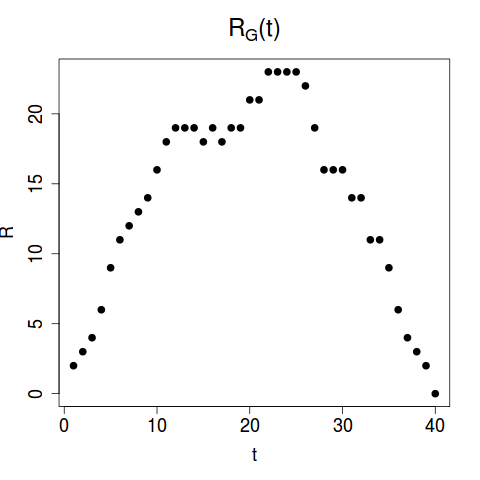}
  \includegraphics[width=.4\textwidth]{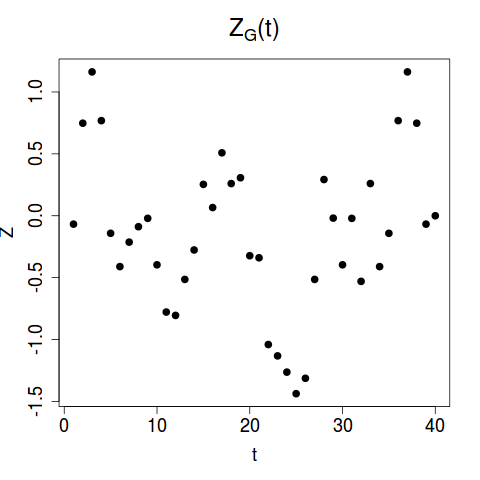}
  \caption{The profile of $R_G(t)$ and $Z_G(t)$ against $t$ on a sequence of points all randomly drawn from $\mathcal{N}(\mathbf{0},I_2)$.  There is no change-point in the sequence.}
  \label{fig:RZ0}
\end{figure}

To test $H_0$ versus $H_a$ we use the scan statistic
\begin{equation}
  \label{eq:onechpt}
  \max_{n_0 \leq t \leq n_1} Z_G(t),
\end{equation}
where $n_0$ and $n_1$ are pre-specified constraints for the range of $\tau$ as described earlier.  The null hypothesis is rejected if the maxima is greater than some threshold.  Section \ref{sec:an-analyt-appr} describes how to choose the threshold to control the family wise error rate.


\subsection{Test Statistic for a Changed Interval Alternative}
\label{sec:changed-interval}
Next, we derive the test statistic for testing $H_0$ \eqref{H0} versus the changed interval alternative $H_2$ (\ref{H1two}).
Similar to the single change-point case, any
specific alternative $(t_1,t_2)$ divides the data into two groups, one group containing all points observed during $(t_1,t_2]$, and the other group containing all points observed outside of this interval.  Then, the number of edges in $G$ connecting data points from different groups is
\begin{align*}
   R_G(t_1,t_2) & = \sum_{(i,j)\in G} I_{g_i(t_1,t_2) \neq g_j(t_1,t_2)}, \quad g_i(t_1,t_2) = I_{t_1<i\leq t_2}.
\end{align*}
We standardize $R_G(t_1,t_2)$ as before,
\begin{align*}
  Z_G(t_1,t_2) & = -\frac{R_G(t_1,t_2)-\bE(R_G(t_1,t_2))}{\sqrt{\bV(R_G(t_1,t_2))}}.
\end{align*}
Lemma \ref{lemma:twochptstat} below gives explicit expressions for the expectation and variance of $R_G(t_1,t_2)$ under the permutation null.  The scan statistic involves a maximization over $t_1$ and $t_2$,
\begin{equation}
  \label{eq:twochpt}
  \max_{\footnotesize \begin{array}{c} 1\leq t_1<t_2\leq n \\ l_0\leq t_2-t_1 \leq l_1 \end{array} } Z_G(t_1, t_2)
\end{equation}
where $l_0$ and $l_1$ are constraints on the window size.  For example, we can set $l_1 = n-l_0$ so that only alternatives where the number of observations in either group is larger than $l_0$ are considered. 

We can further constrain $t_1$ and $t_2$ to prefixed sets based on domain knowledge.  If we do so, the $p$-value approximations in Section \ref{sec:asympt-appr} will have minor but obvious modifications, which can be followed straightforwardly by steps given in Section \ref{sec:asympt-appr}.

\begin{lemma}\label{lemma:twochptstat}
    Under the permutation null, the expectation and variance of $R_G(t_1,t_2)$ are
  \begin{align*}
    \bE(R_G(t_1,t_2)) & = p_1(t_2-t_1) | G|, \\
    \bV(R_G(t_1,t_2)) & = p_2(t_2-t_1) | G| + \left( \frac{1}{2}p_1(t_2-t_1) - p_2(t_2-t_1)\right) \sum_i|G_i|^2 \\
    & \quad + \left( p_2(t_2-t_1) - p_1^2(t_2-t_1) \right) | G|^2,
  \end{align*}
where $p_1(\cdot)$ and $p_2(\cdot)$ are defined in Lemma (\ref{lemma:onechptstat}).
\end{lemma}

The proof for this lemma is very similar to the proof of Lemma \ref{lemma:onechptstat} and is omitted here.





\section{Analytic Approximations to Significance Levels}
\label{sec:an-analyt-appr}

How large do the values of the scan statistics (\ref{eq:onechpt}) and (\ref{eq:twochpt}) need to be to constitute sufficient evidence against the null hypothesis of homogeneity?  In other words, we are concerned with the tail distribution of the scan statistics under $H_0$, that is,
\begin{equation}\label{eq:onechptp}
  \bP\left(\max_{n_0 \leq t \leq n_1} Z_G(t)>b \right)
\end{equation}
for the single change-point alternative, and
\begin{equation}\label{eq:twochptp}
\bP\left( \max_{\scriptsize \begin{array}{c} 1\leq t_1<t_2\leq n \\ l_0\leq t_2-t_1 \leq l_1 \end{array}}Z_G(t_1, t_2) > b  \right)
\end{equation}
for the changed interval alternative. In the rest of the paper, we omit the implicitly obvious constraint $1\leq t_1<t_2\leq n$ for simplicity. 

The null distributions of $\max Z_G(t)$ and $\max Z_G(t_1,t_2)$ are defined as the permutation distribution.  For small $n$, we can directly sample from the permutation distribution to approximate \eqref{eq:onechptp} and \eqref{eq:twochptp}.  However, when $n$ is large, permutation is computationally prohibitive, especially for \eqref{eq:twochptp} where each scan is of order $\mathcal{O}(n^2)$ if $l_1-l_0 \sim \mathcal{O}(n)$.  Therefore, we derive analytic expressions for both tail probabilities to make the method instantly applicable.  Treating $\{Z_G(t)\}$ and $\{Z_G(t_1,t_2)\}$ as families of tests, the two probabilities are their family-wise error rates.  The tests are dependent since they are all based on the same sequence. The marginal distributions of $Z_G(t)$ and $Z_G(t_1,t_2)$, under permutation, are also quite complicated.  Therefore, it is impossible to obtain exact expressions for the two probabilities for finite $n$.  In the rest of this chapter, we give analytic approximations to the two probabilities.  We first show that, under mild conditions on $G$, $\{Z_G([nu]\footnote{$[x]$ is the largest interger that is no larger than $x$.}): 0<u<1\}$ converges to a Gaussian process and $\{Z_G([nu],[nv]): 0<u<v<1\}$ converges to a Gaussian random field as $n\rightarrow\infty$ (Section \ref{sec:limit-distr}).  We then derive analytic approximations to the two probabilities under Gaussian field approximation (Section \ref{sec:asympt-appr}).  To achieve better accuracy for the case of small $n$ and for the case where the conditions for Gaussian convergence are questionable, we refine our approximations by correcting the skewness in the marginal distributions (Section \ref{sec:skewness-correction}).  All of these approximations are checked by numerical studies under a set of representative scenarios (Section \ref{sec:numerical-studies}).


\subsection{Asymptotic Properties of the Processes}
\label{sec:limit-distr}



In this section, we derive the limiting distributions of $\{Z_G([nu]):0<u<1\}$ and $\{Z_G([nu],[nv]):0<u<v<1\}$ under permutation.  We first introduce some notation.  For edge $e = (e_-, e_+)$, where $e_-<e_+$ are the indices of the nodes connected by the edge $e$, let
\begin{equation}
  \label{eq:Ae}
  A_e = G_{e_-} \cup G_{e_+},
\end{equation}
be the set of edges that connect to either node $e_-$ or node $e_+$, and
\begin{equation}
  \label{eq:Be}
  B_e = \cup \{ A_{e^\prime}: e^\prime\in A_e\},
\end{equation}
be the set of edges that connect to nodes in $G_{e_-}$ and $G_{e_+}$.

We define two asymptotic conditions on the graph.
\begin{condition}\label{cond:graph1}
   $|G|\sim \mathcal{O}(n^\alpha)$, $0<\alpha < 1.125$.
\end{condition}
\begin{condition}\label{cond:graph2}
  $\sum_{e\in G}|A_e||B_e| \sim o(n^{1.5(\alpha \wedge 1)})$.
\end{condition}

\begin{theorem}\label{thm:limit}
  Under conditions \ref{cond:graph1} and \ref{cond:graph2}, as $n\rightarrow\infty$,
  \begin{enumerate}
  \item $\{Z_G([nu]):0<u<1\}$ converges to a Gaussian process, which we denote as $\{Z_G^\star(u):0<u<1 \}$,
  \item  $\{Z_G([nu],[nv]):0<u<v<1\}$ converges to a two-dimensional Gaussian random field, which we denote as $\{Z_G^\star(u,v): 0<u<v<1\}$,
  \end{enumerate}
under the permutation distribution.
\end{theorem}
The proof for this theorem utilizes the Stein's method \cite{chen2005stein}.  The whole proof is in Appendix \ref{sec:proofs-limit-distr}.

\begin{remark}
Condition \ref{cond:graph2} restricts both the size and number of hubs, which are nodes with a large degree.  The largest hub in the graph must have degree smaller than $n^{0.75(\alpha\wedge 1)}$ to satisfy the condition.  On the other hand, if we increase the number of edges in the graph (increase $\alpha$), the densest graph we could achieve under condition \ref{cond:graph2} has the number of edges of order less than $n^{1.125}$.  This is because when $|G|\sim\mathcal{O}(n^\alpha), \alpha>1$, $\sum_{e\in G}|A_e||B_e| \geq n^\alpha n^{\alpha-1} n^{2(\alpha-1)} = n^{4\alpha-3}$.
\end{remark}

\begin{lemma}\label{lemma:rho}
The covariance function of the Gaussian process $Z_G^\star(u), 0<u<1$, defined as $\rho^\star_G(u,v)  \overset{\Delta}{=} \cov(Z_G^\star(u), Z_G^\star(v))$ has the following expression:
   \begin{align}
    \label{eq:rho}
    \rho^\star_G(u,v) & = \frac{2(u\wedge v)^2(1-(u\vee v))^2|G|}{\sigma^\star_G(u) \sigma^\star_G(v)} \nonumber  \\
& \quad + \frac{(u\wedge v)(1-(u\vee v))(1-2u)(1-2v)\sum_i|G_i|^2 }{\sigma^\star_G(u) \sigma^\star_G(v)},
  \end{align}
where $$\sigma^\star_G(u) = \sqrt{2u^2(1-u)^2|G|+u(1-u)(1-2u)^2\sum_i|G_i|^2}.$$
\end{lemma}

The lemma is proved through combinatorial analysis and the details are in Appendix \ref{sec:proof-lemma-rho}.

  $\rho^\star_G(u,v)$ is partially differentiable in $u (\neq v)$ to all orders. So, fixing $v$, the $k$-th order left- and right- derivatives in $u$ at $u=v$ are well defined for all $k$.  We denote the $k$-th left- and right- derivative by $f_{v,-}^{(k)}(0) (\equiv \lim_{u\nearrow v}\frac{ \partial \rho^\star_G(u,v)}{\partial u}) $ and $f_{v,+}^{(k)}(0)$, respectively.  One important property, which can be checked by tedious algebra, is that $f_{v,-}^\prime(0) = -f_{v,+}^\prime(0)$.




\subsection{Asymptotic Approximations to $p$-Values}
\label{sec:asympt-appr}


We now examine the asymptotic behavior of the two probabilities \eqref{eq:onechptp} and \eqref{eq:twochptp}.  Our approximations will involve the function $\nu(x)$ defined as
\begin{equation}\label{eq:nu}
\nu(x) = 2x^{-2} \exp\left\{ -2 \sum_{m=1}^\infty m^{-1} \Phi\left(-\frac{1}{2}xm^{1/2} \right)\right\}, \quad x>0. \end{equation}
This function is closely related to the Laplace transform of the overshoot over the boundary of a random walk.  A simple approximation given in \cite{siegmund2007statistics} is sufficient for numerical purpose:
\begin{equation}
  \label{eq:nu_approx}
  \nu(x) \approx \frac{(2/x)(\Phi(x/2)-0.5)}{(x/2)\Phi(x/2)+\phi(x/2)}.
\end{equation}
The following proposition is the foundation for obtaining analytic approximations to the probabilities.

\begin{prop}\label{thm:asym}
Assume that $n_0\rightarrow\infty$, $n_1\rightarrow\infty$, $b\rightarrow \infty$, and $n\rightarrow\infty$ in a way such that for some $0<x_0<x_1<1$ and $b_0>0$
$$n_i/n\rightarrow x_i ~(i=0,1) \text{ and } b/\sqrt{n}\rightarrow b_0.$$
Then as $n\rightarrow\infty$,
\begin{equation}
  \label{eq:onechptp1star}
  P\left(\max_{n_0\leq t\leq n_1} Z_G^\star(t/n)>b\right) \sim b \phi(b) \int_{x_0}^{x_1} h^\star_{r_0,r_1}(x)\nu\left(b_0\sqrt{2h^\star_{r_0,r_1}(x)}\right)dx,
\end{equation}
\begin{align}
   P& \left(\max_{ n_0\leq t_2-t_1\leq n_1 } Z_G^\star(t_1/n,t_2/n)>b\right) \label{eq:twochptp1star}  \\
&\quad \quad \quad \quad \sim b^3 \phi(b) \int_{x_0}^{x_1} \left(h^\star_{r_0,r_1}(x)\nu(b_0\sqrt{2h^\star_{r_0,r_1}(x)})\right)^2(1-x)dx \nonumber
\end{align}
where
$$h^\star_{r_0,r_1}(x) = \frac{1}{2x(1-x)} + \frac{2}{4x(1-x)+(1-2x)^2(r_1-4r_0)}, $$
with $r_0\overset{\Delta}{=}\lim_{n\rightarrow\infty} |G|/n$, and $r_1\overset{\Delta}{=}\lim_{n\rightarrow\infty} \sum_i|G_i|^2/|G|$.
\end{prop}
The proof of this proposition utilizes Woodroofe's method  \citep{woodroofe1976frequentist, woodroofe1978large} and Siegmund's method \citep{siegmund1988approximate,siegmund1992tail}.  The whole proof is in Appendix \ref{sec:proof-prop-refthm}.

\begin{remark}
  Since $n\sum_i|G_i|^2 \geq (\sum_i|G_i|)^2 = 4 |G|^2$, $r_1-4r_0$ is always non-negative, and $$h_{r_0,r_1}^\star(x) \in \left[\frac{1}{2x(1-x)}, \frac{1}{x(1-x)} \right].$$
\end{remark}

Based on Proposition \ref{thm:asym}, when $\sum_{e\in G}|A_e||B_e| \sim o(n^{3/2})$, $|G|\sim \mathcal{O}(n)$, we approximate \eqref{eq:onechptp} and \eqref{eq:twochptp} by
\begin{align}
  \label{eq:onechptp1}
  P & \left(\max_{n_0\leq t\leq n_1} Z_G(t)>b\right)  \\
& \quad \quad \sim b \phi(b) \int_{n_0/n}^{n_1/n} h^\star_{\hat{r}_0, \hat{r}_1}(x)\nu\left(b_0\sqrt{2h^\star_{\hat{r}_0, \hat{r}_1}(x)}\right)dx, \nonumber \\
   P& \left(\max_{ n_0\leq t_2-t_1\leq n_1 } Z_G(t_1,t_2)>b\right) \label{eq:twochptp1}  \\
&\quad \quad \sim b^3 \phi(b) \int_{n_0/n}^{n_1/n} \left(h^\star_{\hat{r}_0, \hat{r}_1}(x)\nu(b_0\sqrt{2h^\star_{\hat{r}_0, \hat{r}_1}(x)})\right)^2(1-x)dx, \nonumber
\end{align}
where $\hat{r}_0 = |G|/n$, $\hat{r}_1 = \sum_i|G_i|^2/|G|$.

\begin{remark}
  In practice, when using \eqref{eq:onechptp1} and \eqref{eq:twochptp1} to approximate the tail probabilities, we use $h_G(n,x)$ in place of $h^\star_{\hat{r}_0,\hat{r}_1}(x)$, where $h_G(n,x)$ is the finite-sample equivalent of $h^\star_{\hat{r}_0,\hat{r}_1}(x)$ for the stochastic process $Z_G([nu])$.  That is
  \begin{align*}
    h^\star_{\hat{r}_0,\hat{r}_1}(x) & = \lim_{u\nearrow x}\frac{\partial \rho^\star_G(u,x)}{\partial u}, \\
    h_G(n,x) & = \frac{1}{n} \lim_{s\nearrow nx}\frac{\partial \rho_G(s,nx)}{\partial s},
  \end{align*}
where $\rho_G(s,t) \overset{\Delta}{=} \cov(Z_G(s), Z_G(t))$.  The explicit expression for $h_G(n,x)$ is
  \begin{equation}
\label{eq:rhoneg}
h_G(n,x) =  \frac{(n-1)[h_1(n,x) | G| + h_2(n,x) \sum_{i=1}^n |G_i|^2 - h_3(n,x) | G|^2] }{2u(1-u)[h_4(n,x) | G| + h_5(n,x) \sum_{i=1}^n |G_i|^2 - h_6(n,x) | G|^2 ]},
\end{equation}
where
\begin{align*}
  h_1(n,x) & = 4 n {\left(n - 1\right)} \left( - 2nx^2 + 2nx - 1\right) \\ 
h_2(n,x) & = n\left[n(n+1)(1-2x)^2 - 2(n-1)\right] \\ 
h_3(n,x) & = 4n\left[n(1-2x)^2 -1 \right] \\
h_4(n,x) & =4 n (n-1)(nx-1)(n-nx-1) \\
h_5(n,x) & = n  \left(n - 1\right) \left[ n^2(1-2x)^2 - n + 2\right]  \\
h_6(n,x) & = 4 n \left[ n^2(1-2x)^2 -2n(1-3x+3x^2) +1 \right] .
\end{align*}
It is easy to show that $\lim_{n\rightarrow\infty} h_G(n,x) = h^\star_{\hat{r}_0,\hat{r}_1}(x)$.
\end{remark}


\subsection{Skewness Correction}
\label{sec:skewness-correction}


 Convergence of $Z_G(t)$ to normal is slow if $t/n$ is close to 0 or 1.  Also, we may doubt the validity of conditions \ref{cond:graph1} and \ref{cond:graph2} if the graph contains large hubs.  For instance, as we show via simulation in Section \ref{sec:numerical-studies} and as detailed in \cite{radovanovic2010hubs}, MST and NNG constructed on high dimensional data can have large hubs under standard distance measures, such as $L_2$ and $L_1$.  
Then the statistic $Z_G(t)$ is left-skewed (see Figure \ref{fig:skew}, left panel), and the $p$-value approximations \eqref{eq:onechptp1} and \eqref{eq:onechptp2} overestimate the tail probabilities.   The other extreme is the MDP, where each node has degree 1 and the graph is completely ``flat''.  The statistic $Z_G(t)$ and the two ends is right-skewed (see Figure \ref{fig:skew}, right panel), and the $p$-value approximations \eqref{eq:onechptp1} and \eqref{eq:onechptp2} underestimate the true tail probabilities.

\begin{figure}[!htp]
  \centering
  \includegraphics[width=.4\textwidth]{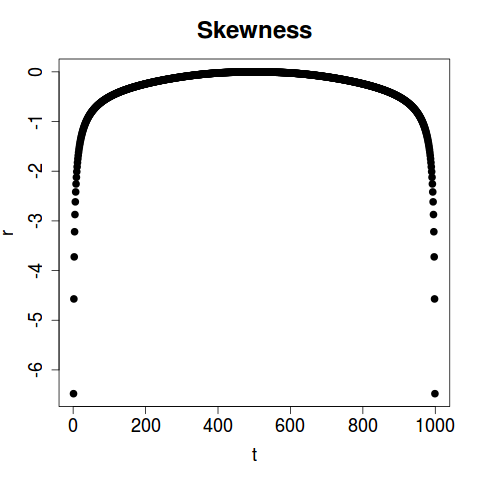}
  \includegraphics[width=.4\textwidth]{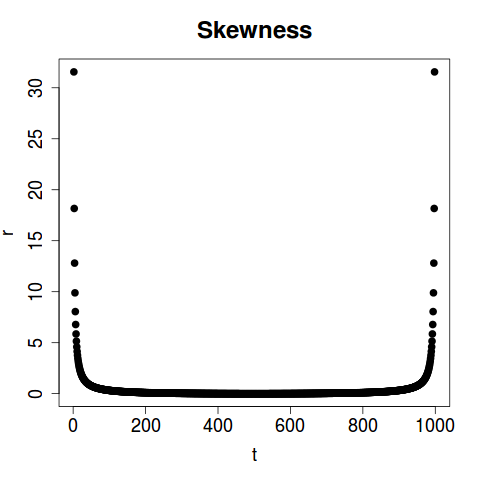}
\caption{Plots of skewness $\gamma_G(t) (=\bE(Z_G(t)))$ against $t$ with $G$ being MST (left panel) and MDP (right panel) constructed on Euclidean distance on a sequence of 1,000 points randomly generated from $\mathcal{N}(\mathbf{0}, I_{100})$. }
  \label{fig:skew}
\end{figure}


Skewness correction in tail probability approximation of change-point tests was first carried out in \cite{tu1999maximum} and later modified in \cite{tang2001mapping}.  Both of these papers applied a universal third moment correction.  In our problem, the extent of the skewness of $Z_G(t)$ depends on the value of $t$.  This can be seen clearly in Figure \ref{fig:skew}, as $Z_G(t)$ is more skewed towards the two ends. Since universal corrections are too crude, we adopt a different approach where the skewness correction adapts to the skewness of $Z_G(t)$ at each $t$.  In particular, we give a better approximation to the marginal probability, $\bP(Z_G(t)\in b+dx/b)$ in the single change-point case and $\bP(Z_G(t_1,t_2)\in b+dx/b)$ in the changed interval case, for which normal approximation was used in producing the approximations \eqref{eq:onechptp1} and \eqref{eq:onechptp2}.


Consider first the approximation of the marginal probability $\bP(Z\in b+dx/b)$, suppressing in our notation the dependence on $t$.  Since $Z$ has been properly standardized, $\bE(Z)=0, \bE(Z^2)=1.$  Let  $\gamma=\bE(Z^3)$ be the skewness term, which can be calculated explicitly by a combinatorial analysis described in Section \ref{sec:expl-expr-skewn} below.  We make use of the cumulant generating function $\psi(\theta) = \log \bE_P(e^{\theta Z})$.  By change of measure $dQ_\theta = e^{\theta Z - \psi(\theta)}dP$, we can approximate $\bP(Z\in b+dx/b)$ by
 $$\frac{1}{\sqrt{2\pi (1 + \gamma \theta_b) }} \exp(-\theta_b b - x \theta_b /b + \theta_b^2 (1+ \gamma \theta_b/3)/2),$$
where $\theta_b$ is chosen such that $\dot{\psi}(\theta_b) = b$.  By a third Taylor approximation, we get
$$\theta_b \approx (-1+\sqrt{1+2\gamma b})/\gamma.$$  More details are given in appendix \ref{sec:skewness}.

The $p$-value approximations, after correcting for the skewness of the marginal distribution of the two processes, become
  \begin{align}
    \label{eq:onechptp2}
    \bP & \left(\max_{n_0 \leq t \leq n_1} Z_G(t)>b \right)
\approx  b\phi(b)\int_{n_0/n}^{n_1/n}  S_G(nx)   h_G(n,x) \nu(\sqrt{2b_0^2 h_G(n,x)}) dx,
  \end{align}
where
\begin{equation}
  \label{eq:fS}
  S_G(t) = \frac{\exp\left(\frac{1}{2}(b-\hat{\theta}_{b,G}(t))^2 + \frac{1}{6}\gamma_G(t) \hat{\theta}_{b,G}(t)^3\right)}{\sqrt{1 + \gamma_G(t) \hat{\theta}_{b,G}(t) }},
\end{equation}
with $\gamma_G(t) = \bE[Z_G^3(t)]$ and $ \hat{\theta}_{b,G}(t) = (-1+\sqrt{1+2\gamma_G(t) b})/\gamma_G(t)$.

 \begin{align}
    \label{eq:twochptp2}
    & \bP \left(\max_{n_0\leq t_2-t_1\leq n_1} Z_G(t_1, t_2) > b  \right) \\
   & \approx \frac{\phi(b)}{b} \sum_{n_0\leq t_2-t_1\leq n_1} S_G(t_1,t_2)\left(b_0^2h_G(n,(t_2-t_1)/n) \nu(b_0\sqrt{2h_G(n,(t_2-t_1)/n)})\right)^2,  \nonumber
  \end{align}
where
\begin{equation}
  \label{eq:fS}
  S_G(t_1,t_2) = \frac{\exp\left(\frac{1}{2}(b-\hat{\theta}_{b,G}(t_1,t_2))^2 + \frac{1}{6}\gamma_G(t_1,t_2) \hat{\theta}_{b,G}(t_1,t_2)^3\right)}{\sqrt{1 + \gamma_G(t_1,t_2) \hat{\theta}_{b,G}(t_1,t_2) }}, \end{equation}
with $\gamma_G(t_1,t_2) = \bE[Z_G^3(t_1,t_2)]$ and $$ \hat{\theta}_{b,G}(t_1,t_2) = (-1+\sqrt{1+2\gamma_G(t_1,t_2) b})/\gamma_G(t_1,t_2).$$

\begin{remark}\label{remark:extra}
When the marginal distribution is highly left-skewed, it is possible that $\gamma(t)$ is too small for $1+2\gamma(t) b$ to be positive.  This does not mean that the solution to $\dot{\psi}_t(\theta) = b$ does not exist, but that higher moments are needed to get a good approximation.  In this paper, we apply an easy heuristic fix to this problem:  Since $1+2\gamma(t) b<0$ usually happens when $t/n$ is close to 0 or 1, within this problematic region $\theta_b(t)$ can be extrapolated using its values outside the region.  The details of the extrapolation method are given in Appendix \ref{sec:extr-skewn-corr}.  
\end{remark}

\subsection{Explicit Expressions for Skewness}
\label{sec:expl-expr-skewn}
We now derive an explicit expression for the skewness terms $\gamma_G(t)$ and $\gamma_G(t_1,t_2)$ that are used in (\ref{eq:onechptp2}) and (\ref{eq:twochptp2}).
We have
\begin{align*}
  \bE(Z_G^3(t)) & = \frac{\bE^3(R_G(t)) + 3\bE(R_G(t))\bV(R_G(t)) - \bE(R^3(t))}{(\bV(R_G(t)))^{3/2}}, \\
  \bE(Z_G^3(t_1,t_2)) & = \frac{\bE^3(R_G(t_1,t_2)) + 3\bE(R_G(t_1,t_2))\bV(R_G(t_1,t_2)) - \bE(R^3(t_1,t_2))}{(\bV(R_G(t_1,t_2)))^{3/2}}.
\end{align*}
The explicit expressions of $\bE(R_G(t))$, $\bV(R_G(t))$, $\bE(R_G(t_1,t_2))$, and $\bV(R_G(t_1,t_2)$ are given in Lemma \ref{lemma:onechptstat} and Lemma \ref{lemma:twochptstat}.  The explicit expressions of $\bE^3(R_G(t))$ and $\bE^3(R_G(t_1,t_2))$
are given in the following lemma.

\begin{lemma}\label{lemma:gammat}
\allowdisplaybreaks
\begin{align*}
  \bE(R_G^3(t)) & = p_1(t) | G| + \frac{3}{2}p_1(t) \sum_i |G_i|(|G_i|-1) \\
  & \quad + 3p_2(t)\left(| G|(| G|-1) + \frac{1}{2} \sum_i |G_i|(|G_i|-1)(| G|-|G_i|) \right) \\
  & \quad - 3p_2(t) \left( \sum_i |G_i|(|G_i|-1) + \sum_{(i,j)\in G} (|G_i|-1)(|G_j|-1) \right) \\
  & \quad + p_3(t) \sum_i |G_i|(|G_i|-1)(|G_i|-2)  \\
  & \quad + p_4(t) \left( | G|(| G|-1)(| G|-2) + 6 \sum_{(i,j)\in G} (|G_i|-1)(|G_j|-1)\right) \\
  & \quad - 2 p_4(t) \sum_{(i,j)\in G}|\{k:(i,k),(j,k)\in G\}| \\
  & \quad - p_4(t)\left(\sum_i |G_i|(|G_i|-1)(3| G|-2|G_i|-2)  \right).
\end{align*}
The functions $p_1(t)$ and $p_2(t)$ are given in Lemma \ref{lemma:onechptstat}, and
\begin{align*}
  p_3(t) & := \frac{t(n-t)((n-t-1)(n-t-2) + (t-1)(t-2))}{n(n-1)(n-2)(n-3)}, \\
  p_4(t) & := \frac{8t(t-1)(t-2)(n-t)(n-t-1)(n-t-2)}{n(n-1)(n-2)(n-3)(n-4)(n-5)}.
\end{align*}
Also
\begin{equation}
  \label{eq:ER3}
  \bE^3(R_G(t_1,t_2)) = \bE^3(R_G(t_2-t_1)).
\end{equation}
\end{lemma}


\begin{proof}
   For the un-centered process $R_G(t)$,
  $$\bE(R_G^3(t)) = \sum_{(i,j),(k,l),(u,v)\in G} \bP(g_i(t)\neq g_j(t), g_k(t)\neq g_l(t), g_u(t) \neq g_v(t)).$$
There are in total eight different configurations for three edges randomly chosen (with replacement) from the graph (see Figure \ref{fig:3edges} for illustrations).  We derive $\bP(g_i(t)\neq g_j(t), g_k(t)\neq g_l(t), g_u(t) \neq g_v(t)) \overset{\Delta}{=} P_3$ separately for each configuration.  

\begin{figure}[!htp]
  \centering
  \includegraphics[width=\textwidth]{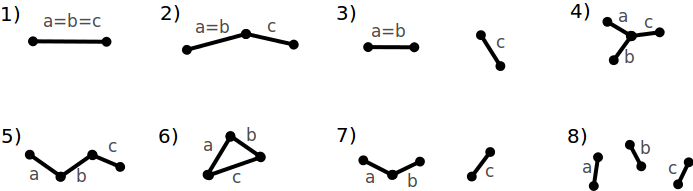}
  \caption{Eight configurations of three edges ($a,b,c$) randomly chosen, with replacement, from the graph.}
  \label{fig:3edges}
\end{figure}

\begin{enumerate}[ 1)]
\item The three edges are actually the same edge.
$$P_3 = \bP(g_i(t)\neq g_j(t)) = \frac{2t(n-t)}{n(n-1)}.$$
\item Two edges are the same and share one node with the third edge.
$$P_3 = \bP(g_i(t)\neq g_j(t), g_i(t)\neq g_k(t)) = \frac{t(n-t)}{n(n-1)}.$$
\item Two edges are the same and do not share any node with the third edge.
$$P_3 = \bP(g_i(t)\neq g_j(t), g_k(t)\neq g_l(t)) = \frac{4t(t-1)(n-t)(n-t-1)}{n(n-1)(n-2)(n-3)}.$$
\item The three edges share one node, and neither of them share the other node (star-shaped).
  \begin{align*}
    P_3 & = \bP(g_i(t)\neq g_j(t), g_i(t)\neq g_k(t), g_i(t)\neq g_l(t)) \\
    & = \frac{t(n-t)((n-t-1)(n-t-2) + (t-1)(t-2))}{n(n-1)(n-2)(n-3)}.
  \end{align*}
\item One edge share one node with another edge and share the other node with the third edge.  No node sharing between the second and the third edge (linear chain).
  \begin{align*}
    P_3 & = \bP(g_i(t)\neq g_j(t), g_i(t)\neq g_k(t), g_j(t)\neq g_l(t)) \\
    & = \frac{2t(t-1)(n-t)(n-t-1)}{n(n-1)(n-2)(n-3)}.
  \end{align*}
\item The three edges form a triangle.
  $$P_3 = \bP(g_i(t)\neq g_j(t), g_j(t)\neq g_k(t), g_k(t)\neq g_i(t)) = 0.$$
\item Two edges share one node, and share no node with the third edge.
  \begin{align*}
    P_3 & = \bP(g_i(t)\neq g_j(t), g_i(t)\neq g_k(t), g_u(t)\neq g_v(t)) \\
    & = \frac{2t(t-1)(n-t)(n-t-1)}{n(n-1)(n-2)(n-3)}.
  \end{align*}
\item No pair of the three edges share any node.
  \begin{align*}
    P_3 & = \bP(g_i(t)\neq g_j(t), g_k(t)\neq g_l(t), g_u(t) \neq g_v(t)) \\
    & = \frac{8t(t-1)(t-2)(n-t)(n-t-1)(n-t-2)}{n(n-1)(n-2)(n-3)(n-4)(n-5)}.
  \end{align*}
\end{enumerate}

Among all $|G|^3$ possible ways of randomly selecting the three edges, the number of occurrences for each of the configuration are:
\begin{enumerate}[ 1)]
\item $|G|$
\item $3\sum_i |G_i|(|G_i|-1)$
\item $3|G|(|G|-1) - 3 \sum_i |G_i|(|G_i|-1)$
\item $\sum_i |G_i|(|G_i|-1)(|G_i|-2)$
\item $6 \sum_{(i,j)\in G} (|G_i|-1)(|G_j|-1) - 6\sum_{(i,j)\in G}|\{k:(i,k),(j,k)\in G\}|$
\item $2\sum_{(i,j)\in G}|\{k:(i,k),(j,k)\in G\}|$
\item $3 \sum_i |G_i|(|G_i|-1)(|G|-|G_i|) + 6\sum_{(i,j)\in G}|\{k:(i,k),(j,k)\in G\}| - 12 \sum_{(i,j)\in G} (|G_i|-1)(|G_j|-1)$
\item $|G|(|G|-1)(|G|-2) + 6 \sum_{(i,j)\in G} (|G_i|-1)(|G_j|-1) - 2\sum_{(i,j)\in G}|\{k:(i,k),(j,k)\in G\}| - \sum_i |G_i|(|G_i|-1)(3|G|-2|G_i|-2)$
\end{enumerate}
The lemma follows by summing up all of the probabilities as enumerated above.

It is not hard to observe that the number of occurrences only depends on the sizes of the two groups, so $ \bE^3(R_G(t_1,t_2)) = \bE^3(R_G(t_2-t_1))$.

\end{proof}

The terms in $\bE(R_G^3(t))$ can be rearranged and written in other forms.  The expansion shown in Lemma \ref{lemma:gammat} makes it easier to understand the origin of each term in the context of the proof.  If we examine the expression, we would find $\gamma_G(t)$ are fully determined by $t$, $n$, $|G|$, $\sum_i |G_i|^2$, $\sum_i |G_i|^3$, $\sum_{(i,j)\in G}(|G_i|-1)(|G_j|-1)$ and the number of triangles in $G$.

For MDP, only configurations 1), 3), and 8) are possible, and the number of occurrences of each case is
\begin{description}
\item[1)] $|G| = n $
\item[3)] $3|G|(|G|-1) = 3n(n-1)$
\item[8)] $|G|(|G|-1)(|G|-2) = n(n-1)(n-2)$
\end{description}
and its $\bE(R_G^3(t))$ has a much simpler expression:
\begin{align*}
  \bE(R_G^3(t)) & = p_1(t)n + p_2(t) 3n(n-1) + p_4(t)  n(n-1)(n-2) \\
  & = (p_1(t)-3p_2(t)+2p_4(t))n + 3(p_2(t)-p_4(t))n^2 + p_4(t) n^3.
\end{align*}



\subsection{Numerical Studies}
\label{sec:numerical-studies}

In this section, we check the analytic approximations to $p$-values, both assuming Gaussianity and after skewness correction, through numerical studies.  We examine both the accuracy of the critical value and the coverage probability.

\subsubsection{Critical Value}
\label{sec:critical-value-1}

We compare the critical values obtained from (\ref{eq:onechptp1}), (\ref{eq:onechptp2}), (\ref{eq:twochptp1}), and (\ref{eq:twochptp2}) to those obtained from doing 10,000 permutations, under various simulation settings.   In each simulation, \emph{iid} sequences of length 1000 were generated from a given distribution $F_0$ in $\mathbb{R}^d$.  MST, MDP, and NNG were constructed on the data based on Euclidean distance.  For each graph, analytic and permutation critical values were computed for both 0.05 and 0.01 $p$-value thresholds.



We first check the single change-point alternative.  Tables \ref{tab:mst_05} - \ref{tab:mdp} show the results for the single change-point alternative with the underlying graph being MST or MDP.  Results for when the underlying graph is NNG, shown in Appendix \ref{sec:single-change-point-1},  are similar to those for when the graph is MST.  In the column headers, ``A1'' denotes critical values obtained assuming Gaussianity (\ref{eq:onechptp1}), ``A2'' denotes critical values obtained after correcting for skewness (\ref{eq:onechptp2}), and ``Per'' denotes critical values obtained by 10,000 permutations (can be viewed as the true $p$-value).

  Six different choices for $F_0$ are shown, for two different distributions (standard normal and exponential with mean 1), each in three different dimensions (d=1,10, or 100).  For $d=10$ or 100, each element of the data vector is generated independently from the given distribution.  The analytic approximations depend also on constraints on the region in which the change-point is searched.  These are reflected in the choice of $n_0$ and $n_1$ ($l_0$ and $l_1$ for the changed interval alternative).  To make things simple, we set $n_1=n-n_0$.  In general, the analytic approximations become less precise when the minimum segment length decreases.  This is mainly because the Gaussian approximation (and skewness correction) to the distribution of $Z(t)$ degrades for small samples.  

Both the analytic and permutation $p$-values depend on certain characteristics of the graph's structure.  The structures of  MST (for $d\geq 2$) and NNG  depend on the underlying data set, and thus the critical values vary by simulation run.  In such cases, we show results for 5 randomly simulated sequences.  Two characteristics of the graph are also shown for each simulated sequence: The sum of squared node degrees ($\sum_i |G_i|^2$) and the maximum node degree ($d_{\text{max}}$).  These quantities give some intuition on the size and density of hubs in the graph.  Since the MST for any one-dimensional data set is a chain, in this case the critical values do not change with simulation run for each setting of the parameters.

The structure of the MDP graph is always the same for all data sets.  Therefore, the critical values for MDP-based scan depend only on $n, n_0$, $n_1$  ($l_0$ and $l_1$ for the changed interval alternative).  The critical values for MDP-based scan do not depend on the dimension or the underlying distribution of the data.  As emphasized in \citet{rosenbaum2005exact}, statistics based on the MDP is truly a distribution free method, which can sometimes be desirable. 

We can see from the tables that the analytic approximations after skewness correction perform much better than the analytic approximations under Gaussian assumption, especially when dimension increases.  The accuracy of the skew-corrected approximation does not degrade significantly with dimension.  For the statistics based on MST and NNG, the skew-corrected approximations remain accurate for window sizes as small as 25 at both 0.05 and 0.01 significance levels.  For the statistics based on MDP, the skew-corrected approximations work well when the minimum window size is as small as 25 at 0.05 significance level, and 50 at 0.01 significance level.

There is not much difference between results for simulations based on normal and those based on exponential distributions.  The main factor influencing approximation accuracy, other than the minimum window size, is the dimension ($d$).  As dimension increases, the graph becomes more ``star-shaped'' as reflected by the increase in both $\sum |G_i|^2$ and $d_{\text{max}}$.  As shown in Section \ref{sec:skewness-correction}, skewness and other higher order moments of $Z_G(t)$ are a function of polynomials of the node degrees.  Thus the increase in the number and density of hubs makes skewness correction important in high dimensions.  This also indicates that a different distance measure other than Euclidean distance in high dimension to better distinguish different distributions.

For the changed interval alternative, the results are similar, with details in Appendix \ref{sec:chang-interv-altern}.

\begin{table}[!htp]
  \caption{Critical values for the single change-point scan statistic based on MST at 0.05 significance level.  $n=1000$. ``A1'' denotes critical values obtained assuming Gaussianity (\ref{eq:onechptp1}), ``A2'' denotes critical values obtained after correcting for skewness (\ref{eq:onechptp2}), and ``Per'' denotes critical values obtained by 10,000 permutations.}
  \label{tab:mst_05}
  \centering
%

\begin{tabular}{c|ccc|ccc|ccc|cc}
\hline \hline
& \multicolumn{9}{|c|}{Critical Values} & \multicolumn{2}{|c}{Graph} \\ \cline{2-12}
& \multicolumn{3}{|c|}{$n_0=100$} & \multicolumn{3}{|c|}{$n_0=50$} & \multicolumn{3}{|c|}{$n_0=25$} & & \\ \cline{2-10}
  & A1 & A2 & Per & A1 & A2 & Per & A1 & A2 & Per & $\sum |G_i|^2$ & $d_{\text{max}}$ \\ \hline \hline
  $d=1$ & 2.98 & 3.05 & 3.04 & 3.08 & 3.22 & 3.23 & 3.14 & 3.39 & 3.49 & 4994 & 2 \\  \hline
  & 2.92 & 2.90 & 2.90 & 3.00 & 2.95 & 2.95 & 3.05 & 2.98 & 2.96 &  5430 &  8 \\
N(0,1) & 2.92 & 2.89 & 2.89 & 3.00 & 2.95 & 2.92 & 3.05 & 2.97 & 2.95 &  5438 &  7 \\
$d=10$ & 2.92 & 2.90 & 2.87 & 3.00 & 2.95 & 2.94 & 3.05 & 2.98 & 2.96 &  5394 &  7 \\
  & 2.92 & 2.89 & 2.86 & 3.00 & 2.94 & 2.90 & 3.05 & 2.97 & 2.92 &  5534 &  8 \\
  & 2.92 & 2.89 & 2.89 & 3.00 & 2.95 & 2.92 & 3.05 & 2.97 & 2.95 &  5460 &  7 \\ \hline
  & 2.93 & 2.91 & 2.89 & 3.01 & 2.97 & 2.96 & 3.06 & 3.00 & 2.97 &  5064 &  7 \\
Exp(1) & 2.93 & 2.91 & 2.88 & 3.01 & 2.97 & 2.92 & 3.06 & 3.00 & 2.95 &  5082 &  7 \\
$d=10$ & 2.93 & 2.91 & 2.91 & 3.01 & 2.98 & 2.97 & 3.06 & 3.01 & 3.00 &  5028 &  5 \\
  & 2.93 & 2.91 & 2.87 & 3.01 & 2.98 & 2.93 & 3.06 & 3.01 & 2.97 &  5028 &  6 \\
  & 2.93 & 2.91 & 2.88 & 3.01 & 2.96 & 2.92 & 3.06 & 2.98 & 2.94 &  5180 &  9 \\ \hline
  & 2.86 & 2.69 & 2.68 & 2.94 & 2.70 & 2.68 & 3.00 & 2.70 & 2.68 & 12454 & 38 \\
N(0,1) & 2.86 & 2.72 & 2.72 & 2.95 & 2.74 & 2.72 & 3.00 & 2.74 & 2.72 & 10904 & 38 \\
$d=100$ & 2.86 & 2.70 & 2.66 & 2.94 & 2.71 & 2.66 & 3.00 & 2.71 & 2.66 & 11294 & 42 \\
  & 2.87 & 2.72 & 2.68 & 2.95 & 2.74 & 2.68 & 3.00 & 2.74 & 2.68 & 10690 & 40 \\
  & 2.86 & 2.69 & 2.65 & 2.94 & 2.70 & 2.65 & 3.00 & 2.70 & 2.65 & 11722 & 40 \\ \hline
  & 2.85 & 2.64 & 2.60 & 2.93 & 2.65 & 2.60 & 2.99 & 2.65 & 2.60 & 14706 & 56 \\
Exp(1) & 2.87 & 2.77 & 2.76 & 2.95 & 2.80 & 2.77 & 3.01 & 2.81 & 2.77 &  9608 & 25 \\
$d=100$ & 2.84 & 2.62 & 2.53 & 2.93 & 2.62 & 2.53 & 2.99 & 2.62 & 2.53 & 15536 & 77 \\
  & 2.86 & 2.74 & 2.69 & 2.95 & 2.76 & 2.69 & 3.00 & 2.76 & 2.69 & 10890 & 30 \\
  & 2.86 & 2.72 & 2.66 & 2.94 & 2.73 & 2.66 & 3.00 & 2.73 & 2.66 & 12018 & 39 \\
\hline \hline
\end{tabular}

\end{table}

\begin{table}[!htp]
  \caption{Critical values for the single change-point scan statistic based on MST at 0.01 significance level.  $n=1000$.}
  \label{tab:mst_01}
  \centering
%

\begin{tabular}{c|ccc|ccc|ccc|cc}
\hline \hline
& \multicolumn{9}{|c|}{Critical Values} & \multicolumn{2}{|c}{Graph} \\ \cline{2-12}
& \multicolumn{3}{|c|}{$n_0=100$} & \multicolumn{3}{|c|}{$n_0=50$} & \multicolumn{3}{|c|}{$n_0=25$} & & \\ \cline{2-10}
  & A1 & A2 & Per & A1 & A2 & Per & A1 & A2 & Per & $\sum |G_i|^2$ & $d_{\text{max}}$ \\ \hline \hline
    $d=1$ & 3.52 & 3.62 & 3.67 & 3.60 & 3.81 & 3.85 & 3.65 & 4.05 & 4.31 & 4994 & 2 \\  \hline
  & 3.47 & 3.43 & 3.46 & 3.53 & 3.46 & 3.48 & 3.57 & 3.48 & 3.48 &  5430 &  8 \\
N(0,1) & 3.47 & 3.43 & 3.44 & 3.53 & 3.46 & 3.46 & 3.57 & 3.47 & 3.46 &  5438 &  7 \\
$d=10$ & 3.47 & 3.43 & 3.44 & 3.53 & 3.46 & 3.47 & 3.58 & 3.48 & 3.48 &  5394 &  7 \\
  & 3.47 & 3.42 & 3.38 & 3.53 & 3.46 & 3.40 & 3.57 & 3.47 & 3.41 &  5534 &  8 \\
  & 3.47 & 3.43 & 3.44 & 3.53 & 3.46 & 3.46 & 3.57 & 3.47 & 3.46 &  5460 &  7 \\ \hline
  & 3.48 & 3.45 & 3.40 & 3.54 & 3.49 & 3.44 & 3.58 & 3.50 & 3.45 &  5064 &  7 \\
Exp(1) & 3.48 & 3.44 & 3.40 & 3.54 & 3.48 & 3.42 & 3.58 & 3.50 & 3.44 &  5082 &  7 \\
$d=10$ & 3.48 & 3.45 & 3.47 & 3.54 & 3.49 & 3.49 & 3.58 & 3.51 & 3.52 &  5028 &  5 \\
  & 3.48 & 3.45 & 3.41 & 3.54 & 3.49 & 3.44 & 3.58 & 3.51 & 3.46 &  5028 &  6 \\
  & 3.48 & 3.44 & 3.49 & 3.54 & 3.47 & 3.53 & 3.58 & 3.48 & 3.54 &  5180 &  9 \\ \hline
  & 3.42 & 3.17 & 3.19 & 3.48 & 3.17 & 3.19 & 3.53 & 3.17 & 3.19 & 12454 & 38 \\
N(0,1) & 3.42 & 3.21 & 3.24 & 3.49 & 3.21 & 3.24 & 3.53 & 3.21 & 3.24 & 10904 & 38 \\
$d=100$ & 3.42 & 3.19 & 3.17 & 3.49 & 3.19 & 3.17 & 3.53 & 3.19 & 3.17 & 11294 & 42 \\
  & 3.42 & 3.22 & 3.18 & 3.49 & 3.22 & 3.18 & 3.53 & 3.22 & 3.18 & 10690 & 40 \\
  & 3.42 & 3.18 & 3.21 & 3.49 & 3.18 & 3.21 & 3.53 & 3.18 & 3.21 & 11722 & 40 \\ \hline
  & 3.41 & 3.14 & 3.12 & 3.48 & 3.14 & 3.12 & 3.52 & 3.14 & 3.12 & 14706 & 56 \\
Exp(1) & 3.43 & 3.28 & 3.26 & 3.49 & 3.28 & 3.26 & 3.54 & 3.28 & 3.26 &  9608 & 25 \\
$d=100$ & 3.41 & 3.15 & 3.10 & 3.48 & 3.15 & 3.10 & 3.52 & 3.15 & 3.10 & 15536 & 77 \\
  & 3.42 & 3.24 & 3.21 & 3.49 & 3.24 & 3.21 & 3.53 & 3.24 & 3.21 & 10890 & 30 \\
  & 3.42 & 3.22 & 3.13 & 3.48 & 3.22 & 3.13 & 3.53 & 3.22 & 3.13 & 12018 & 39 \\
\hline \hline
\end{tabular}

\end{table}

\begin{table}[!htp]
  \caption{Critical values for the single change-point scan statistic based on MDP.  $n=1000$.}
  \label{tab:mdp}
  \centering
significance level = 0.05
  \begin{tabular}{c|cc|cc|cc|cc}
\hline \hline
 & & & \multicolumn{2}{|c|}{$d=1$} & \multicolumn{2}{|c|}{$d=10$} & \multicolumn{2}{|c}{$d=100$}   \\ \hline
$n_0$ & A1 & A2 & N(0,1) & Exp(1) & N(0,1) & Exp(1) & N(0,1) & Exp(1) \\ \hline \hline
200 & 2.82 & 2.84 & 2.83 & 2.81 & 2.85 & 2.85 & 2.85 & 2.83 \\ \hline
100 & 2.98 & 3.07 & 3.06 & 3.04 & 3.08 & 3.08 & 3.07 & 3.05 \\ \hline
50 & 3.08 & 3.27 & 3.30 & 3.29 & 3.35 & 3.36 & 3.35 & 3.31 \\ \hline
25 & 3.14 & 3.48 & 3.54 & 3.58 & 3.57 & 3.66 & 3.60 & 3.60  \\ \hline \hline
  \end{tabular}

\bigskip

significance level = 0.01
  \begin{tabular}{c|cc|cc|cc|cc}
\hline \hline
 & & & \multicolumn{2}{|c|}{$d=1$} & \multicolumn{2}{|c|}{$d=10$} & \multicolumn{2}{|c}{$d=100$}   \\ \hline
$n_0$ & A1 & A2 & N(0,1) & Exp(1) & N(0,1) & Exp(1) & N(0,1) & Exp(1) \\ \hline \hline
200 & 3.38 & 3.43 & 3.39 & 3.38 & 3.44 & 3.46 & 3.45 & 3.44 \\ \hline
100 & 3.52 & 3.66 & 3.66 & 3.64 & 3.67 & 3.75 & 3.67 & 3.59 \\ \hline
50 & 3.60 & 3.90 & 3.99 & 3.99 & 3.94 & 4.05 & 3.95 & 3.99 \\ \hline
25 & 3.65 & 4.21 & 4.61 & 4.65 & 4.78 & 4.72 & 4.59 & 4.81 \\ \hline \hline
  \end{tabular}

\end{table}

\subsubsection{Coverage Probability}
\label{sec:coverage-probability-1}

For both widely used significance levels, 0.05 and 0.01, we also check the coverage probability of the $p$-value approximations.  From the previous section on checking critical values, we see that the underlying distribution of the data does not affect the result, so we generate data only from the multivariate Gaussian distribution.  We now expand our study to a denser graphs.  In each simulation run, a sequence of length 1,000 were generated from $\mathcal{N}(\mathbf{0},I_d)$. 1,3,5-MST/MDP/NNG were constructed on the data based on Euclidean distance.  1-MST is the same as MST, which we also call the 1st MST.  The 2nd MST is defined as a spanning tree that is orthogonal to the 1st MST (not using any edge in the 1st MST) minimizing the total distance over the edges, and the 2-MST is defined as the union of 1st and 2nd MST.  Recursively, the $k$th MST is the spanning tree that is orthogonal to all $i$ MSTs ($i<k$) minimizing the total distance over edges, and the $k$-MST is defined as the union of all of the $i$th MSTs, $i=1,\dots,k$.  Similar definitions apply to $k$-MDP and $k$-NNG.  

 In each simulation run, we calculated the critical value based on the $p$-value approximation for a given significance level (0.05 or 0.01), and used this critical value as the actual threshold.  Then, we did 10,000 permutations and calculated the percentage of the permutations with the scan statistic larger than the threshold.  This percentage is viewed as the coverage probability.  We checked the coverage probability for data in low dimension ($d=10$) and high dimension ($d=100$), with 100 simulation runs for each.  Figures \ref{fig:cov_50_05} and \ref{fig:cov_50_01} show boxplots of the coverage probability for the single change-point alternative with the smallest window size ($n_0$) being 50.  The results for $n_0$ being 25 or 100, other settings unchanged, are shown in Appendix \ref{sec:coverage-probability}.  The coverage probabilities for the p-value apporixmation assuming Gaussianity \eqref{eq:onechptp1} are shown in blue and those after skewness correction \eqref{eq:onechptp2} are shown in red.  We see that coverage probabilities based on the skewness-corrected $p$-value approximation are closer to the designed significance level, with the improvement being very significant for MDP in all scenarios and for MST/NNG when the data dimension is high.  

\begin{figure}[!htp]
  \centering
  \includegraphics[width=\textwidth]{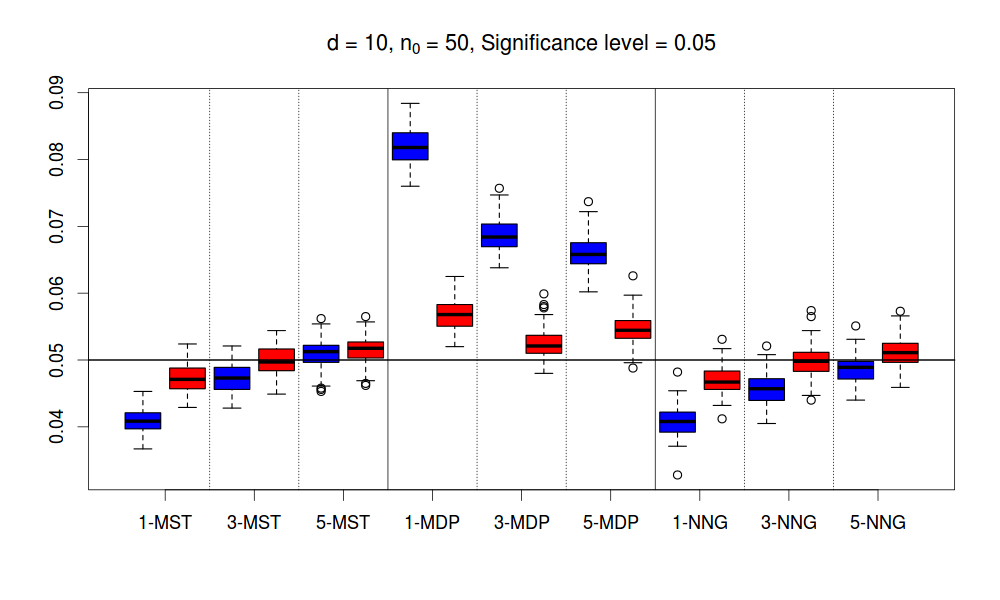}
  \includegraphics[width=\textwidth]{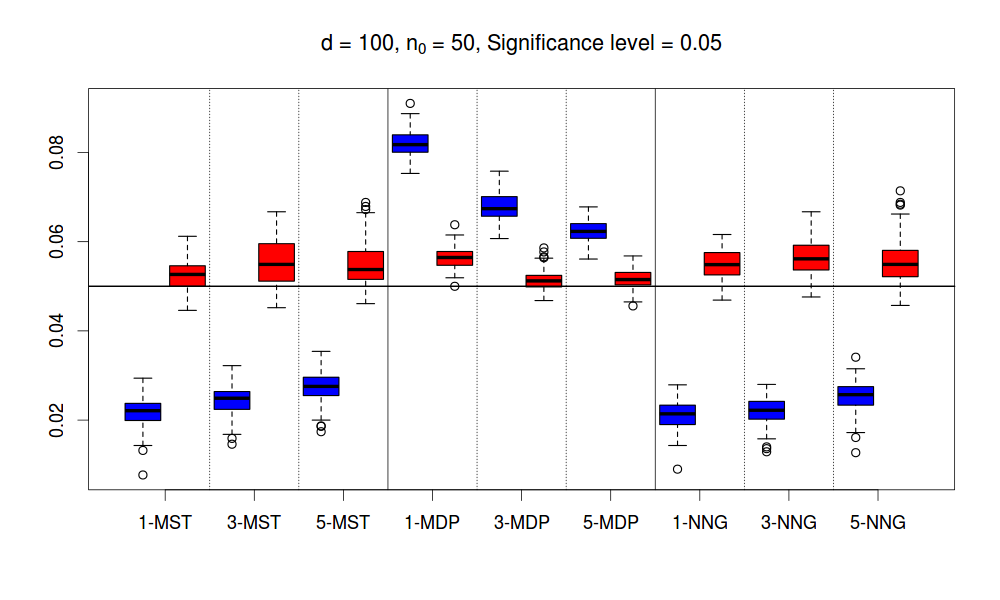}
  \caption{Boxplots for coverage probability with significance level 0.05 under the single change-point alternative.  The smallest windows size is 50.  The dimension of each observation in the sequence is 10 in the upper panel and 100 in the lower panel.  For each type of graph, the result from the $p$-value approximation assuming Gaussianity is shown in blue and that after skewness correction is shown in red.}
  \label{fig:cov_50_05}
\end{figure}

\begin{figure}[!htp]
  \centering
  \includegraphics[width=\textwidth]{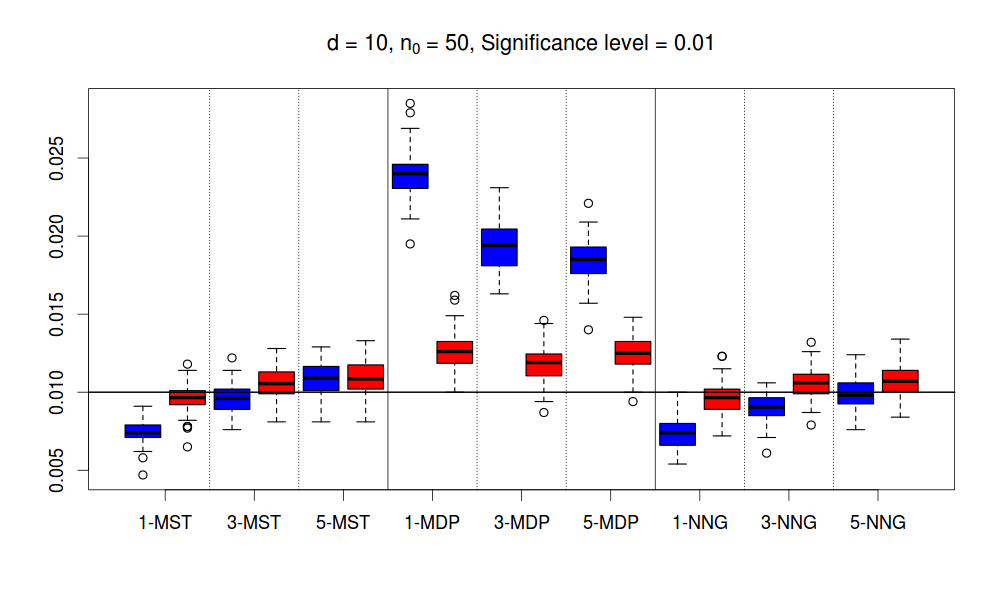}
  \includegraphics[width=\textwidth]{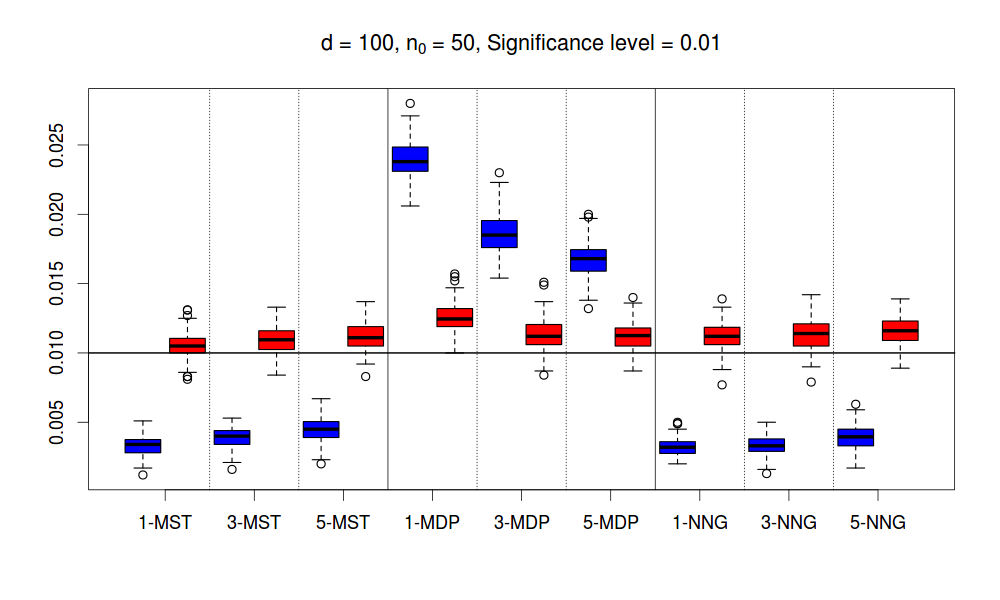}
  \caption{Boxplots for coverage probability with significance level 0.01. Other parameters remain unchanged from Figure \ref{fig:cov_50_05}.}
  \label{fig:cov_50_01}
\end{figure}

Base on results in both the critical values and coverage probabilities, the skew-corrected approximations are quite safe to use.



\section{Power Comparisons}
\label{sec:power-analysis}

To examine the power of our proposed method, we consider cases that parametric methods are applicable.  In particular, we consider cases where normal theory can apply.  In the first simulation set-up, we generated a sequence of 200 observations from the following model:
$$ \by_t \sim \left\{
                \begin{array}{ll}
                  N(\mathbf{0}, I_d), & t = 1,\dots, 100; \\
                  N(\boldsymbol\mu, \Sigma), & t = 101,\dots,200.
                \end{array}
              \right.
$$
As before, $d$ is the dimension of each observation.  There is a change-point at 100.  The mean $\mu$ of the second half of the data is shifted from 0 by amount $\Delta$ in Euclidean distance.  We considered cases where the covariance matrix remains constant ($\Sigma=I_d$), as well as cases where the covariance matrix also changes.  When the covariance matrix changes, we set $\Sigma$ to a diagonal matrix with $\Sigma[1,1]=d^{1/3}$ and $\Sigma[i,i]=1$ for $i=2,\dots,d$.
We chose $\Delta$ for each value of $d$ so that most methods have moderate power.

Hotelling's T$^2$ is a parametric test designed specifically for detecting a change in multivariate normal mean when there is no change in variance.
When there is a change in both mean and variance, the generalized likelihood ratio test (GLR) can be used.  We compare the graph-based scan statistics to scan statistics based on these two existing methods.  For any candidate change-point $t$, the Hotelling's $T^2$ is
$$
T^2(t)= \frac{t(n-t)}{n}(\bar{\by}_t-\bar{\by}_t^*)^T\widetilde{\Sigma}^{-1}(\bar{\by}_t-\bar{\by}_t^*),
$$
where
\begin{align*}
\bar{\by}_t & = \sum_{i=1}^t \by_i/t,  \quad
\bar{\by}_t^*  = \sum_{i=t+1}^n \by_i/(n-t), \\
\widetilde{\Sigma} & = (n-2)^{-1}\left[\sum_{i=1}^t(\by_i - \bar{\by}_t)(\by_i - \bar{\by}_t)^T + \sum_{i=t+1}^n(\by_i - \bar{\by}_t^*) (\by_i - \bar{\by}_t^*)^T\right].
\end{align*}
The GLR is
$$GLR(t)= n\log|\hat{\Sigma}_n|- t\log|\hat{\Sigma}_t| - (n-t)\log|\hat{\Sigma}_t^*|,  $$
where
\begin{align*}
  \hat{\Sigma}_t & = \frac{ \sum_{i=1}^t(\by_i - \bar{\by}_t)(\by_i - \bar{\by}_t)^T }{t},\quad
  \hat{\Sigma}_t^* = \frac{\sum_{i=t+1}^n(\by_i - \bar{\by}_t^*) (\by_i - \bar{\by}_t^*)^T}{n-t}.
\end{align*}
$T^2(t)$ and $GLR(t)$ both have some constraints on the dimension of the data.  For $T^2$ the number of observations $n$ needs to be larger than the dimension of the data $d$ so that $\widetilde{\Sigma}$ can be inverted.  For GLR, both $t$ and $n-t$ need to be larger than the dimension of the data so that the determinants of $\hat{\Sigma}_t, \hat{\Sigma}_t^*$ are not zero.  Thus, when $d \leq 20$, we set $n_0=d+10$ and $n_1 = n-n_0$.  When $d>20$, we set $n_0=50$ and $n_1=150$. (An exception for GLR is that when $d=50$, $n_0$ and $n_1$ are set to 60 and 140, respectively, so that the test statistic can be calculated.)  

Scan statistics based on the three ways of constructing the graph -- MST, MDP and NNG -- using Euclidean distance are compared to scan statistics based on maximization of $T^2(t)$ and $GLR(t)$.  We also examined the power of denser graphs: 3-MST, 3-MDP and 3-NNG.  The significance level is determined through 10,000 permutation runs (for Hotelling $T^2$ and GLR) or skew-corrected approximations (for graph-based methods).  Table \ref{tab:comparison} shows the number of trials, out of 100, that the null hypothesis is rejected at 0.05 level for each of these methods.  To examine the accuracy of the estimated change-point, the number of trials where the estimated change-point is within 20 from the true change-point is given in parentheses.  In the table, red numbers are cases when the graph-based method outperforms both tests based on normal theory.  In general, they appear when the dimension is relatively high.

  First, compare the graph-based methods to Hotelling's $T^2$:  When the variance does not change, $T^2$ outperforms all other methods in low to moderate dimension ($d<150$).  This is expected, as $T^2$ was designed specifically for this scenario.  Remarkably, graph-based methods surpass $T^2$ at its own game when dimension is high ($d\geq 150$).  If we increase the dimension further, our proposed method is still working while the standard Hotelling's $T^2$ is no longer applicable. 
In the case where the variance also changes.  By assuming an incorrect alternative, the power of $T^2$ is quickly surpassed by graph-based methods, for $d$ as low as 5.

Comparing graph-based methods to GLR-based scan statistic, we see a similar pattern:  When dimension is low ($d=1,5,10$), GLR-based scans dominate in power when both the mean and variance changes.  Graph-based methods exceed GLR in power when $d$ increases, already performing much better by $d=20$, which is considered quite low in today's applications. The low power of GLR at even moderate dimension is due to its requirement that the covariance matrix be estimated for both segments.

We also considered a case where the normality assumption is violated by generating data from the log-normal distribution ($\Sigma=I_d$).  Then, graph-based methods outperform $T^2$ by $d=10$, and GLR even when $d=1$ (3-MST and 3-NNG).

Comparing among the graph-based scan statistics, we see that MST and NNG have comparable power, and dominate MDP in all scenarios.  An explanation is that, of these three types of graphs, the MDP retains the least information from the data, having half as many edges as the other two graphs.  The fact that denser graphs lead to higher power is also evident as we compare the performance of 3-MST/MDP/NNG to the (1-) versions.  Also, 3-MST/MDP/NNG have similar power,  indicating that power is not sensitive to the method of graph construction, so long as the graph and distance function effectively separates $F_0$ from $F_1$.

Another interesting fact on the graph-based tests is that their power mainly depends on the size of the change and not decrease much as the dimension increases.  This can be seen clearly in the first table in Table \ref{tab:comparison}.  As we increase the change ($\Delta$) from 1.2 ($d=100$) to 2 ($d=175$), there is an increasing trend in power for each of the graph-based tests.  On the other hand, there is a slightly decreasing trend for the Hotelling's $T^2$ test.  Also, as we jump from $d=175$ to $d=500$, we only increase the change a little (2 to 2.5) to have all the graph-based tests remain similar power.  These results show that the graph-based tests are powerful in high dimension despite the hubbing phenomenon. 

For all scenarios that the null is rejected, we also tally whether the estimated change-points are within $[80,120]$ to check their accuracy (numbers in parentheses, Table \ref{tab:comparison}).  We see that, in terms of the accuracy, the graph-based methods are comparable to, if not better than, that based on normal theory.

\begin{table}[!htp]
\caption{Number of simulated sequences (out of 100) with significance less than 5\%, and the numbers in parentheses are those having the estimated change-point within [80, 120].}
  \centering
Normal data, $\Sigma = I$

  \begin{tabular}{c|cccccccccc}
\hline \hline
    $d$ & 1 & 10 & 50  & 100 & 125 & 150 & 175 & 500 \\
$\Delta$ & 0.5 & 0.8  & 1  & 1.2 & 1.4 & 1.6 & 2 & 2.5 \\
\hline \hline
T2  & \emph{85} & \emph{97} & \emph{80} & \emph{69} & \emph{69} & \emph{66} & \emph{53} & - \\
& (68) & (83) & (64) & (58) & (58) & (54) & (42) & - \\ \hline
GLR  & \emph{74} & \emph{26} & \emph{12} & -  & -  & -  & - & - \\
& (60) & (14) & (0) & - & - & - & - & - \\ \hline
1-, 3-MST & \emph{15 30} & \emph{20 52} & \emph{14 42} & \emph{17 38} & \emph{27 48} & \emph{38 65} & \emph{\textcolor{red}{{60}} \textcolor{red}{{86}}} & \emph{\textcolor{red}{58 87}} \\
 & (4 16) & (13 37) & (11 37) & (13 34) & (18 44) & (33 \textcolor{red}{\textbf{59}}) & (\textcolor{red}{\textbf{54} \textbf{85}}) & (\textcolor{red}{\textbf{51} \textbf{85}}) \\ \hline
 1-, 3-MDP & \emph{13 19} & \emph{16 34}  & \emph{14 29} & \emph{15 24} & \emph{30 42} & \emph{26 48} & \emph{40 \textcolor{red}{{77}}} & \emph{\textcolor{red}{49 72}} \\
& (0 6) & (6 23) & (7 18) & (8 15) & (19 24) & (19 37) & (30 \textcolor{red}{\textbf{63}}) & (\textcolor{red}{\textbf{29} \textbf{56}}) \\ \hline
1-, 3-NNG & \emph{11 28} & \emph{20 51} & \emph{18 40} &  \emph{17 32} & \emph{27 51} & \emph{32 \textcolor{red}{67}} & \emph{53 \textcolor{red}{{87}}} & \emph{\textcolor{red}{57 88}} \\
& (3 17) & (14 39) & (14 32) & (12 28) & (19 47) & (27 \textcolor{red}{\textbf{61}}) & (\textcolor{red}{\textbf{49} \textbf{85}}) & (\textcolor{red}{\textbf{50} \textbf{85}}) \\
\hline \hline
  \end{tabular}


\bigskip

Normal data, $\Sigma$ is diagonal with $\Sigma[1,1]=d^{1/3}, \Sigma[i,i]=1, i=2,\dots, d$.
  \begin{tabular}{c|cccc}
\hline\hline
    $d$ & 1 & 5 & 10 & 20  \\
$\Delta$ & 0.5 & 0.4 & 0.1 & 0.2  \\
\hline \hline
T2 & \emph{78} & \emph{16} & \emph{7} & \emph{7} \\
& (60) & (11) & (1) & (1) \\ \hline
GLR & \emph{65} & \emph{80} & \emph{69} & \emph{23} \\
& (45) & (70) & (59) & (10) \\ \hline
1-, 3-MST & \emph{14 33} & \emph{29 52} & \emph{35 61} & \emph{\textcolor{red}{{62}} \textcolor{red}{{85}}} \\
& (5 17) & (15 30) & (24 52) & (\textcolor{red}{\textbf{44}} \textcolor{red}{\textbf{77}}) \\ \hline
1-, 3-MDP & \emph{14 17} & \emph{12 29} & \emph{17 42} & \emph{\textcolor{red}{{44}} \textcolor{red}{{76}}} \\
& (1 6) & (4 17) & (7 28) & (\textcolor{red}{\textbf{24}} \textcolor{red}{\textbf{61}}) \\ \hline
1-, 3-NNG & \emph{8 31} & \emph{28 45} & \emph{30 64} & \emph{\textcolor{red}{{58}} \textcolor{red}{{86}}} \\
& (3 11) & (12 28) & (18 51) & (\textcolor{red}{\textbf{42}} \textcolor{red}{\textbf{74}}) \\
\hline \hline
  \end{tabular}

\bigskip

Log-normal data, $\Sigma=I$.

  \begin{tabular}{c|cccccccccc}
\hline \hline
    $d$ & 1 & 5 & 10 & 20 & 50 & 75 & 100 \\
$\Delta$ & 0.7 & 0.9 & 1 & 1 & 1.2 & 1.4 & 1.4 \\
\hline \hline
T2 & \emph{78} & \emph{84} & \emph{78} & \emph{54} & \emph{57} & \emph{43} & \emph{28} \\
& (57) & (69) & (61) & (38) & (43) & (34) & (20) \\ \hline
GLR & \emph{30} & \emph{15} & \emph{16} & \emph{14} &  11 & - & - \\
& (21) & (9) & (8) & (3) & (0) & - & - \\ \hline
1-, 3-MST & \emph{24 59} & \emph{41 67} & \emph{43 \textcolor{red}{{89}}} & \emph{33 \textcolor{red}{{64}}} & \emph{45 \textcolor{red}{{66}}} & \emph{\textcolor{red}{{54}} \textcolor{red}{{84}}} & \emph{\textcolor{red}{{52}} \textcolor{red}{{75}}} \\
&(11 43) & (29 54) & (38 \textcolor{red}{\textbf{79}}) & (26 \textcolor{red}{\textbf{57}}) & (32 \textcolor{red}{\textbf{60}}) & (\textcolor{red}{\textbf{48} \textbf{81}}) & (\textcolor{red}{\textbf{44} \textbf{72}}) \\ \hline
1-, 3-MDP & \emph{18 28} & \emph{23 53} & \emph{24 52} & \emph{13 32} & \emph{30 51} & \emph{23 \textcolor{red}{65}} & \emph{24 \textcolor{red}{58}} \\
& (5 17) & (7 35) & (15 40) & (3 19) & (19 37) & (19 \textcolor{red}{\textbf{50}}) & (15 \textcolor{red}{\textbf{41}}) \\ \hline
1-, 3-NNG & \emph{18 51} & \emph{38 67} & \emph{32 77} & \emph{27 \textcolor{red}{{60}}} & \emph{46 \textcolor{red}{{70}}} & \emph{\textcolor{red}{{49}} \textcolor{red}{{85}}} & \emph{\textcolor{red}{{46}} \textcolor{red}{{75}}} \\
& (10 33) & (27 53) & (27 \textcolor{red}{\textbf{66}}) & (20 \textcolor{red}{\textbf{52}}) & (32 \textcolor{red}{\textbf{61}}) & (\textcolor{red}{\textbf{43}} \textcolor{red}{\textbf{82}}) & (\textcolor{red}{\textbf{38}} \textcolor{red}{\textbf{71}}) \\
\hline \hline
  \end{tabular}

  \label{tab:comparison}
\end{table}



\section{Real Data Examples}
\label{sec:real-data-examples}

We illustrate the new approach on two different applications.   The first is a statistical analysis of the text of \emph{Tirant lo Blanc}.  The second is a longitudinal study of a network through time.


\subsection{Authorship Debate}
\label{sec:authorship-debate}

 \emph{Tirant lo Blanc},  a chivalry novel published in 1490, is considered to be one of the best known medieval works of literature in Catalan, and is well recognized to be a major influence to \emph{Don Quxote}.  For such an important work in western literature, there is a long lasting debate regarding its authorship originated from conflicting information provided in its first published version.  The dedicatory letter at the beginning of the book states,

 \begin{quote}
   ... So that no one else can be blamed if any faults are found in this work, I, Joanot Martorell, knight, take sole responsibility for it, as I have carried out the task singlehandedly...
 \end{quote}
However, the colophon at the end of the book states something different,
\begin{quote}
  ... by the magnificent and virtuous knight, Sir Joanot Martorell, who because of his death, could only finish writing three parts of it.  The fourth part, which is the end of the book, was written by the illustrious knight Sir Marti Joan de Galba.  If faults are found in that part, let them be attributed to his ignorance...
\end{quote}
This inconsistency sparked a debate, still ongoing, about the authorship of \emph{Tirant lo Blanc} since its publication.  Opinions have mainly fallen into two camps, one favoring single authorship by Joanot Martorell and the other favoring a change of author somewhere between chapters 350 and 400 with 487 chapters in total.  One objective way to settle this debate is through the statistical analysis of word usage, which reflects the unique writing style of different people.

\cite{giron2005bayesian} analyzed two sets of word usage statistics extracted from the book.  The first, which we call the word length data set, categorizes the words in each chapter by its length, with a single category for all words with length greater than nine letters.  Thus, this data set represents each chapter by a vector of length 10.  The second, which we call the context-free word frequency data set, counts the occurrence of the 25 most frequent context-free words in each chapter.  \cite{giron2005bayesian} analyzed the two data sets using a Bayesian multinomial change-point model and a Bayesian clustering method, and concluded in favor of the change of author hypothesis with the estimated change-point between chapters 371 and 382.

Here, we apply the graph based change-point method to the two data sets, treating each chapter as a time-point.  There are in total 487 chapters, and we use the 425 chapters that have more than 200 words.  For both data sets, we normalized the count vector for each chapter by dividing the total number of words in the chapter.  Thus, our data is a sequence of 425 normalized proportions, of dimension 10 for the word length data and dimension 25 for the context-free word frequency data.  The $L_2$ norm is used to construct the MST, MDP and NNG graphs representing similarity between chapters.  $Z_G(t)$ and the estimated change-points, computed for each type of graph, are shown in Figure \ref{fig:author}.  Test results using the three different graphs and the two data sets support the change of author hypothesis, with the estimated change-point around chapter 360, which is consistent with the view that there is a change of author somewhere between chapters 350 and 400.  
The $p$-values are shown in the table in Figure \ref{fig:author}. 

\begin{figure}[!htp] 
  \centering
  \includegraphics[width=\textwidth]{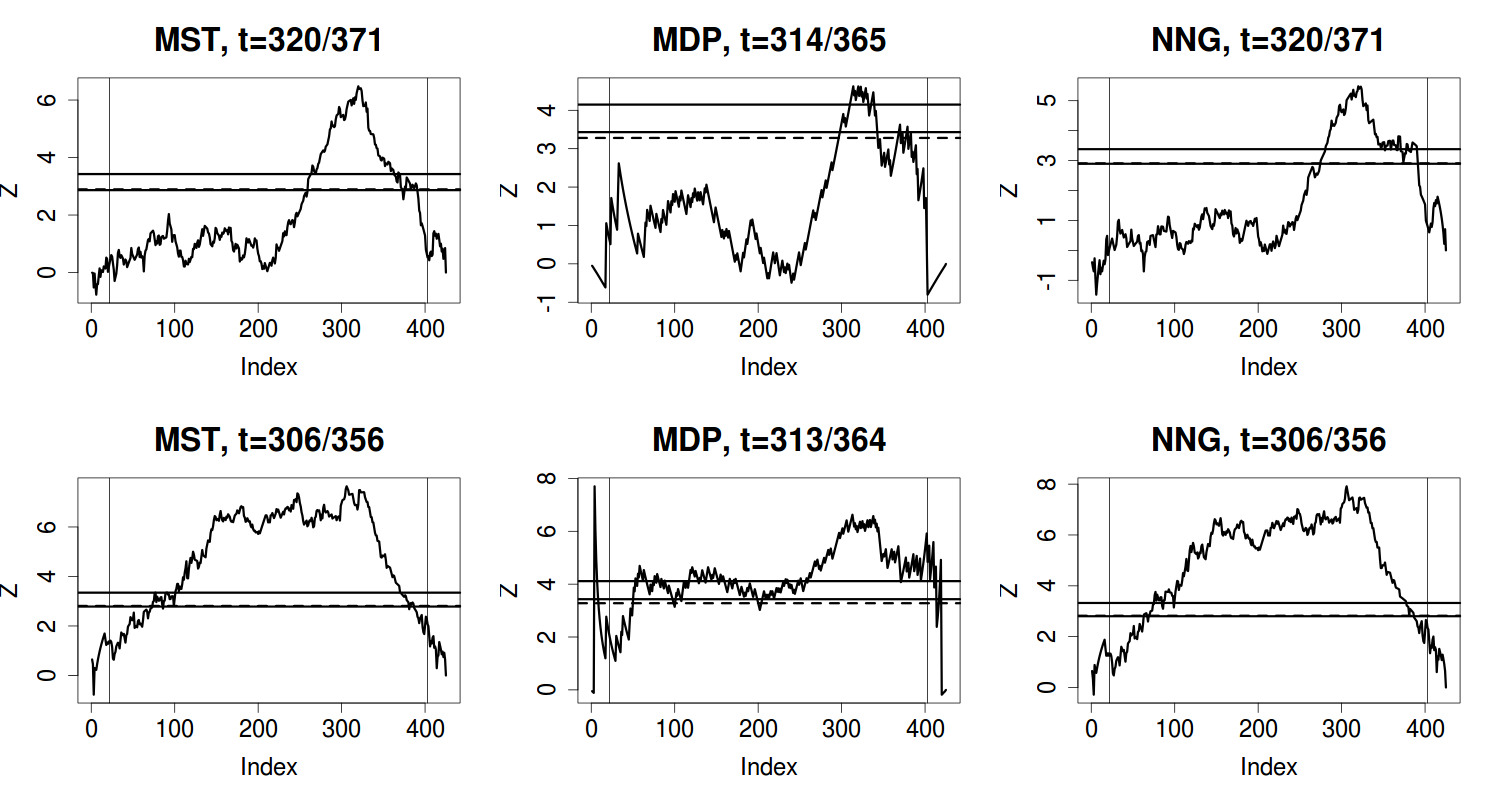}

\vspace{1em}

  \begin{tabular}{c|ccc}
    \hline
    data & MST & MDP & NNG \\ \hline
    word length & 0.0000 (1.5e-9) & 0.0042 (0.0018) & 0.0000 (7.5e-7) \\ \hline
    context-free word frequency & 0.0000 (2.7e-13)& 0.0000 (6.1e-6) & 0.0000 (3.0e-14) \\
    \hline
  \end{tabular}
  \caption{Results of graph-based scans of chapter-wise word usage frequencies of \emph{Tirant lo Blanc}, based on the word length data (first row) and context-free word frequency data (second row).  The three columns show scans based on three different graphs: MST, MDP, and NNG from left to right.  
In each plot, $Z_G(t)$ is plotted along $t$ (chapter).  The estimated change-point is shown in the caption above the plot in the form $A/B$, where $A$ is the index of the change-point within the 425 chapters used for analysis, and $B$ is the chapter number in the novel.  The two vertical lines show $n_0$ and $n_1$; we excluded the first 5\% and the last 5\% of the points.  The horizontal lines show critical values at 0.05 and 0.01 significance levels, with the solid lines showing critical values computed from 10,000 permutations and the dashed lines showing those computed from the analytic approximation with skewness correction.  The table lists the $p$-values for the tests through 10,000 permutations with the skew-corrected approximations in parentheses.}
  \label{fig:author}
\end{figure}


To check the robustness of our analysis, we also applied the scan on data for the first 350 chapters to see if it rejects the null there.  Opinions seem to be quite uniform that the first 350 chapters were all written by Joanot Martorell.  The results are shown in Figure \ref{fig:author2}.  The word length data does not reject the null for the 350 chapters at 0.05 significance level.  However, the context-free word frequency data supports a change-point, although different graphs favor different locations for the change-point.  The $p$-values of the tests are shown in the table in Figure \ref{fig:author2}. 
One explanation is that the context-free word frequency is still affected by the context and thus less robust than the word length in reflecting writing styles. It is also possible that the first author's writing style evolves as he proceeded in the dimension of context-free word frequency. 

\begin{figure}[!htp]
  \centering
  \includegraphics[width=\textwidth]{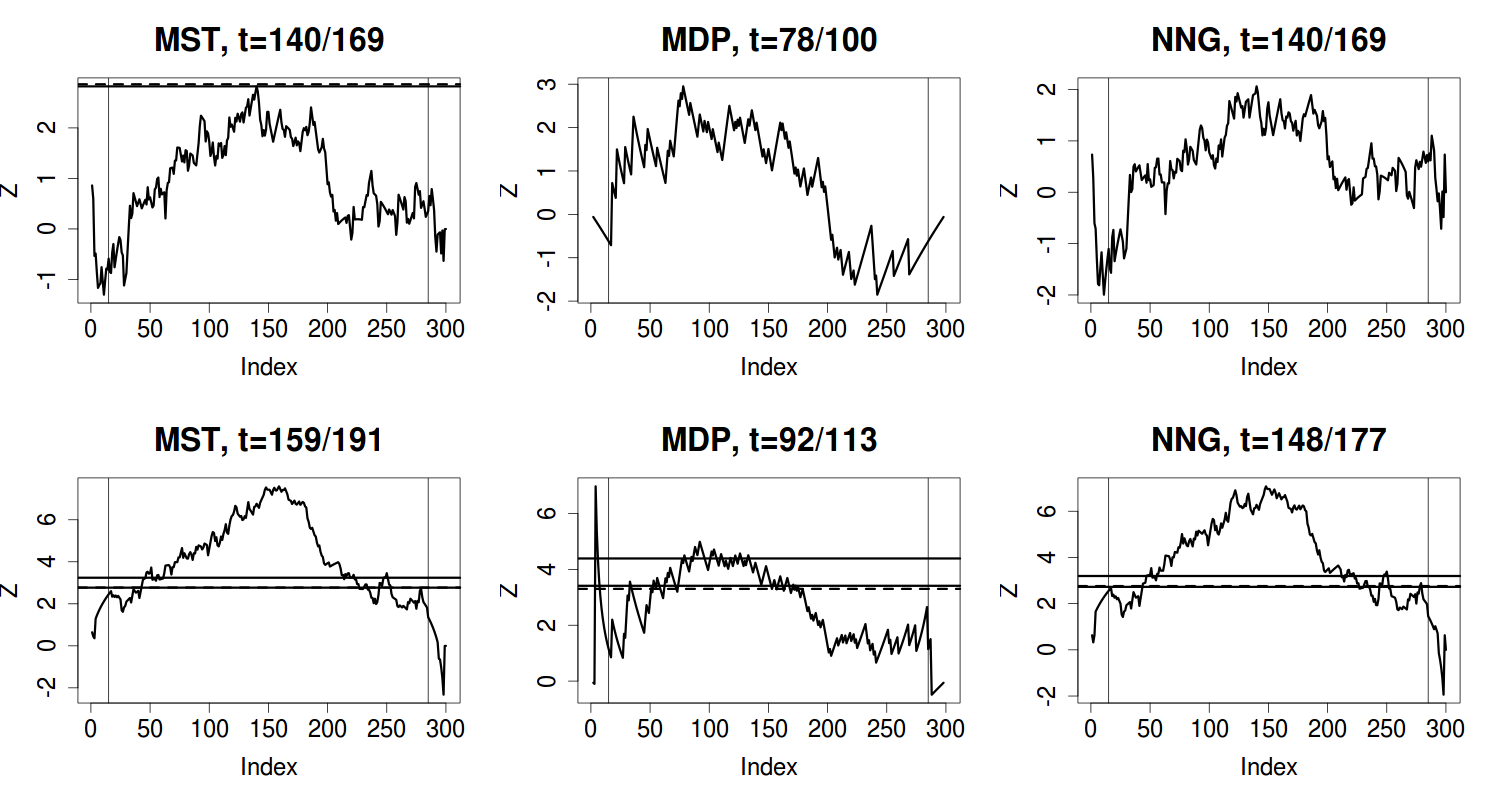}

\vspace{1em}

 \begin{tabular}{c|ccc}
    \hline \hline
    data & MST & MDP & NNG \\ \hline \hline
    word length & 0.0485 (0.0562) & 0.1079 (0.1040) & 0.3053 (0.3527) \\ \hline
    context-free word frequency & 0.0000 (2.9e-13) & 0.0018 (0.0009) & 0.0000 (1.3e-11) \\
    \hline \hline
  \end{tabular}

  \caption{Results from the first 350 chapters.  The setting of the figure is the same as in Figure \ref{fig:author}. The table lists the $p$-values for the tests through 10,000 permutations with the skew-corrected approximations in parentheses.}
  \label{fig:author2}
\end{figure}



\subsection{Friendship Network}
\label{sec:friendship-network}

The MIT Media Laboratory conducted a study following 90 subjects, consisting of students and staff at the university, using mobile phones with pre-installed software recording call logs from July 2004 to June 2005 \citep{eagle2009inferring}.  In this analysis, we extract the information on the caller, callee and time for every call that was made during the study period.   The question of interest is whether phone call patterns changed during this time, which may reflect a change in relationship among these subjects.  
We bin the calls by day and, for each day, construct a network with the 90 subjects as nodes and a link between two subjects if they had at lease one call on that day.  We encode the network of each day by an adjacency matrix, with 1 for element $[i,j]$ if there is an edge between subject $i$ and subject $j$, and 0 otherwise.  Thus, the processed data are adjacency matrices, one for each day from 2004/7/20 to 2005/6/14.

We show results for graphs constructed using two different dissimilarity measures.  Let $A_i$ be the 90 by 90 adjacency matrix on day $i$.  We denote $v_i$ to be the vector form of $A_i$. The dissimilarities are:
\begin{enumerate}[(1)]
\item the number of different edges: $\|v_i - v_j\|_1 = \|v_i-v_j\|_2^2$,
\item the number of different edges, normalized by the geometric mean of the total for each day: $\frac{\|v_i - v_j\|_1}{\sqrt{\|v_i\|_1 \|v_j\|_1}}$.
\end{enumerate}

\begin{figure}[!htp]
  \centering
  \includegraphics[width=\textwidth]{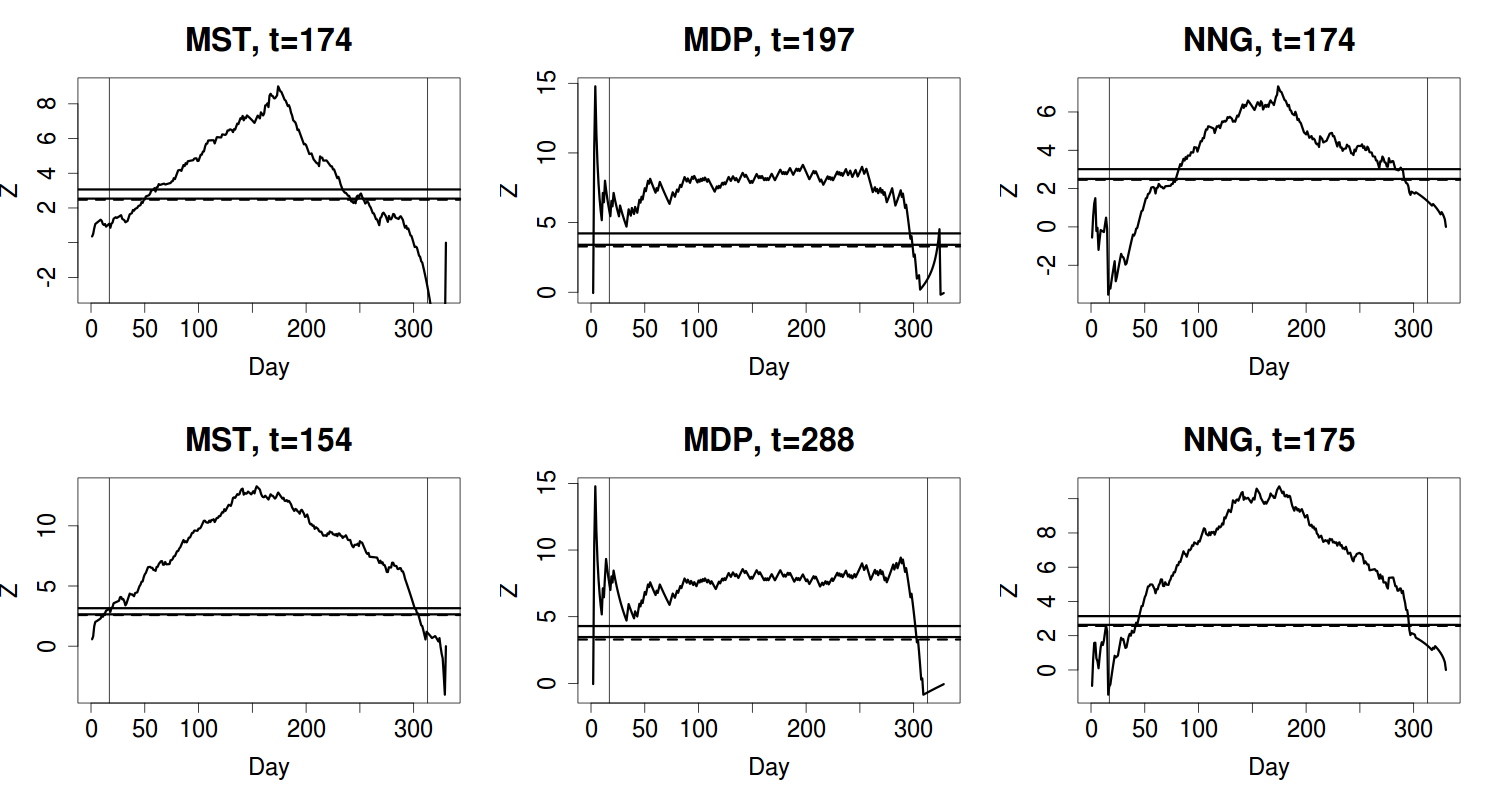}
  \caption{Results of graph-based scans of the MIT phone call network.   Top row shows results from using number of different edges as the dissimilarity measure and bottom row shows results from using the normalized number of different edges.   The three columns show three different ways of constructing the graph: MST, MDP, and NNG from left to right. The content in each plot is the same as in Figure \ref{fig:author}.}  
  \label{fig:phone}
\end{figure}

Results based on different dissimilarities and different ways of constructing the graph are shown in Figure \ref{fig:phone}.  We see that statistics based on MST and NNG give similar results under both dissimilarities.  Based on the scans using MST and NNG, a change-point occurred at around December 19, 2004 ($t=154$) or January 9/10, 2005 ($t=174/175$), which are almost the two ends of the Winter break.  So these results suggest a change of phone call pattern as they move from the fall quarter to the spring quarter.  The statistic based on MDP is quite horizontal for a long range of time.  One reason is that for this network data, the change is relatively gradual rather than abrupt and the constructing of MDP then tend to connect observation $t$ to observations $t-1$ or $t+1$, which makes the resulting graph not informative in determine the location of a ``big'' change.  
The $p$-values for the scan based on MST and NNG under both dissimilarity measures are all $<0.0001$, by both 10,000 permutations and skew-corrected approximations.  


\section{Extensions}
\label{sec:extensions}
In this section, we discuss some extensions to the approach to deal with local dependency in the sequence (Section \ref{sec:block-perm}) and to construct a confidence interval for the change-point (Section \ref{sec:conf-interv-estim}).

\subsection{Block Permutation for Local Dependency}
\label{sec:block-perm}

In both applications, independence is a useful but idealized assumption for the data. One way to deal with local dependency is to define the null distribution as the distribution under block permutation rather than permutation.  In block permutation, the sequence is divided into blocks of size $b$ and the blocks are permuted.\footnote{There are $b$ ways to divide the sequence into blocks of size $b$, with the first block of size $1,2,\dots,b$.  For each block permuted sequence, we first randomly chosen one way from the $b$ ways, and then randomly permute the blocks.}  
The standardized count is then defined as:
\begin{equation}
  \label{eq:Zbp}
  Z_{G,bp}(t) = - \frac{R_G(t)-\bE_{bp}(R_G(t))}{\sqrt{\bV_{bp}(R_G(t))}},
\end{equation}
where $\bE_{bp}(R_G(t))$ and $\bV_{bp}(R_G(t))$ are the expectation and variance for $R_G(t)$ under block permutation.\footnote{$\bE_{bp}(\cdot)$ and $\bV_{bp}(\cdot)$ can be calculated by doing, for example, 10,000 block permutations.}  The test statistic is now defined as:
\begin{equation}
  \label{eq:maxZbp}
  \max_{n_0\leq t\leq n_1} Z_{G,bp}(t),
\end{equation}
and the $p$-value for the above statistic can be obtained by block permutation.  While analytical formulas for the null moments and the family-wise error rate under the block permutation model are too complicated to be practical, for medium to small data sets these quantities can be obtained by brute force computation.  

We used the block permutation model to analyze both the \emph{Tirant lo Blanc} authorship and the friendship network data.  The $p$-values for the authorship data under different block sizes (2, 5, 10) are summarized in Table \ref{tab:authorbp}, and plots of the $Z_{G,bp}$ values are shown in Figure \ref{fig:authorbp5} (block size 5) and Appendix \ref{sec:scan-over-entire} (block size 2 and 10).  Results for the authorship data with the first 350 chapters and the phone call network data are in Appendix \ref{sec:block-perm-results}.

\begin{table}[!htp]
  \centering
   \caption{p-values from 10,000 block permutations for the authorship data set.}
 \begin{tabular}{|c||ccc|ccc|}
    \hline
    & \multicolumn{3}{|c|}{Word length} & \multicolumn{3}{c|}{Context-free word frequency} \\ \cline{2-7}
block size   & MST & MDP & NNG & MST & MDP & NNG \\ \hline \hline
1$^*$ & 0 & 0.0042 & 0 & 0 & 0 & 0 \\ \hline
2 & 0 & 0.0029 & 0 & 0 & 0 & 0 \\ \hline
5 & 0 & 0.0041 & 0 & 0 & 0.0001 & 0 \\ \hline
10 & 0 & 0.0057 & 0 & 0 & 0.0006 & 0 \\ \hline
  \end{tabular}
\begin{flushleft}
\quad \quad \quad \quad $^*$: a block size of 1 is equivalent to permutation under independence assumption.
\end{flushleft}
  \label{tab:authorbp}
\end{table}

\begin{figure}[!htp]
  \centering
  \includegraphics[width=\textwidth]{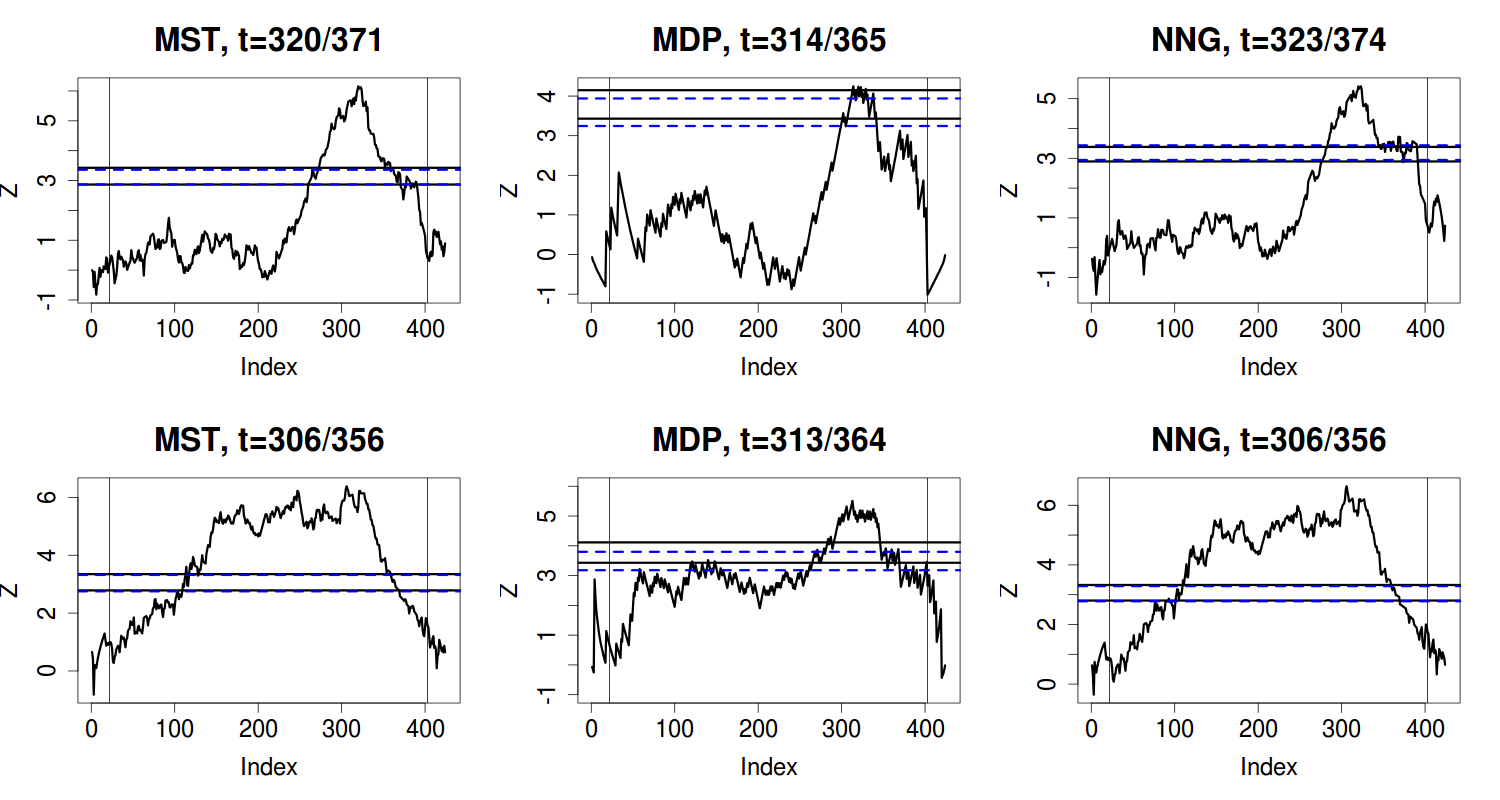}
   \caption{Results of graph-based scans of chapter-wise word usage frequencies of \emph{Tiran lo Blanc}, based on the word length data (first row) and context-free word frequency data (second row), under block permutation with block size 5.  The critical values at 0.05 and 0.01 levels, obtained from 10,000 block permutation runs are shown in these blue dash lines.  The solid black lines are critical values from permutation.}
  \label{fig:authorbp5}
\end{figure}

In all cases, block permutation gives the same conclusion as permutation.  Block permutation tends to increase the $p$-value when the block size is large.  Simulation studies show that this is the case even when the sequence is made up of independent observations.  It is due to the fact that 
block permutation with large blocks produces a less homogeneously mixed sequence.    Despite the slight decrease in significance, the fact that $p$-values remain in the same regime even under block permutation boosts our confidence in our conclusions for both applications.


\subsection{Confidence Interval for Estimated Change-Point}
\label{sec:conf-interv-estim}
Upon rejection of the null one is often concerned with the accuracy of the estimate of the change-point location.  To some extent, we explored this in Section \ref{sec:power-analysis} by tallying whether the estimated change-point is within a fixed window centered at the true change-point.  Here, we describe a procedure for constructing a confidence region for the change-point, which is motivated by the approach studied by \cite{worsley1986confidence}, where it was called a Cox-Spj{\o}tvoll type confidence region citing the original paper \cite{cox1982partitioning}.


The Cox-Spj{\o}tvoll type confidence region is based on the duality relationship where the $\alpha$ level confidence region for a change-point contains all values $k$ that partition the sequence into two subsequences (before and after $k$), where within the subsequences the hypothesis of homogeneity can not be rejected at level $\alpha$.  That is, let $p^L_k$ and $p^R_k$ be the $p$-values for testing the null hypothesis of homogeneity in respectively the left and right subsequences when partitioned at $k$.  
A $1-\alpha$ confidence region can be expressed as
$$D_\alpha = \{k:~p_k^L,~ p_k^R \geq 1-\sqrt{1-\alpha} \}.$$

On the word length data, the 0.01 confidence region $D_{0.01}$ for the chapter where the author changes from  Joanot Martorell to Marti Joan de Galba is  $\{296\} \cup [298,355]$.  While this region is informative, the break between chapters 296 and 298 is hard to interpret.  Note that for a value $k<\hat{\tau}$ to belong to a $1-\alpha$ level Cox-Spj{\o}tvoll region, both the subsequence to the left of $k$ and the subsequence to the right of $k$ must test negative for a change-point.  The subsequence to the right of $k$, which contains $\hat{\tau}$, usually tests positive if there are enough points between $k$ and $\hat{\tau}$.  In this way, the confidence region has the desirable tendency of including points close to $\hat{\tau}$.  The subsequence to the left of $k$, which does not include $\hat{\tau}$, may test positive for a change-point for two reasons:  Inhomogeneity in the left subsequence, e.g., existence of another change-point before $\tau$, or a false positive due to random chance.  Neither of these reasons seems to have much to do with our precision of estimating $\tau$ with $\hat{\tau}$.

For example, 297 is excluded in the confidence region for the change of author in \emph{Tirant lo Blanc} because of what happened in Chapters 1 to 297, not because of what happened in Chapter 298 onwards (the right subsequence actually test negative).  There is no historical evidence pointing to a third author, and thus we are willing to believe that there is either one author (Joan Martorell) or two authors (Joan Martorell and Marti Joan de Galba).  Thus, when we compute our confidence region for the change-point, we are doing so under the premise that there is a single change in author.  Hence, not including 297 due to possible inhomogeneity prior to chapter 297, when our best estimate of the change-point is 320, seems a bit silly.  We would much rather include 297, claim to have a conservative interval, and forego the exact coverage property of the  Cox-Spj{\o}tvoll region.

Motivated by these considerations, we modify the Cox-Spj{\o}tvoll type confidence region in the following way:  If $k$ comes before the estimated change-point ($\hat{\tau}$), we test whether the right-subsequence contains a change-point; and if $k$ comes after $\hat{\tau}$, we test whether the left-subsequence contains a change-point.  In other words, let
$$C_{\alpha, L} = \{k<\hat{\tau}:~ p_k^R\geq 1-\sqrt{1-\alpha} \},$$
$$C_{\alpha, R} = \{k<\hat{\tau}:~ p_k^L\geq 1-\sqrt{1-\alpha} \},$$
our confidence region is $C_\alpha = C_{\alpha, L} \cup C_{\alpha, R} \cup \{\hat{\tau}\}$.  Since $C_\alpha \supseteq D_\alpha$, $C_\alpha$ is a conservative $\alpha$ level confidence region.  $C_\alpha$ is more likely than $D_\alpha$ to form an interval, and despite its conservativeness it is more accurate in reflecting the precision of $\hat{\tau}$ in estimating $\tau$ when we believe $\tau$ to be the sole change-point.  

This modified procedure is illustrated on the word length data shown in Figure \ref{fig:CI}.  The 0.01 confidence region for the location of change in author is $[281,355]$, 
which correspond to original chapter numbers 330 to 409.  Comparing to $D_{0.01}$, we deduce that not only 297 but 281-295 were excluded from $D_{0.01}$ due to possible inhomogeneity in the left subsequence.  As for any real data, homogeneity is an ideal and not a completely correct assumption for the \emph{Tirant lo Blanc} word length sequence.  In reporting the $C_{\alpha}$ region, we are choosing a region that is more conservative, but in turn, more robust against slight deviations from the model.

\begin{figure}[!htp]
  \centering
  \includegraphics[width=.65\textwidth]{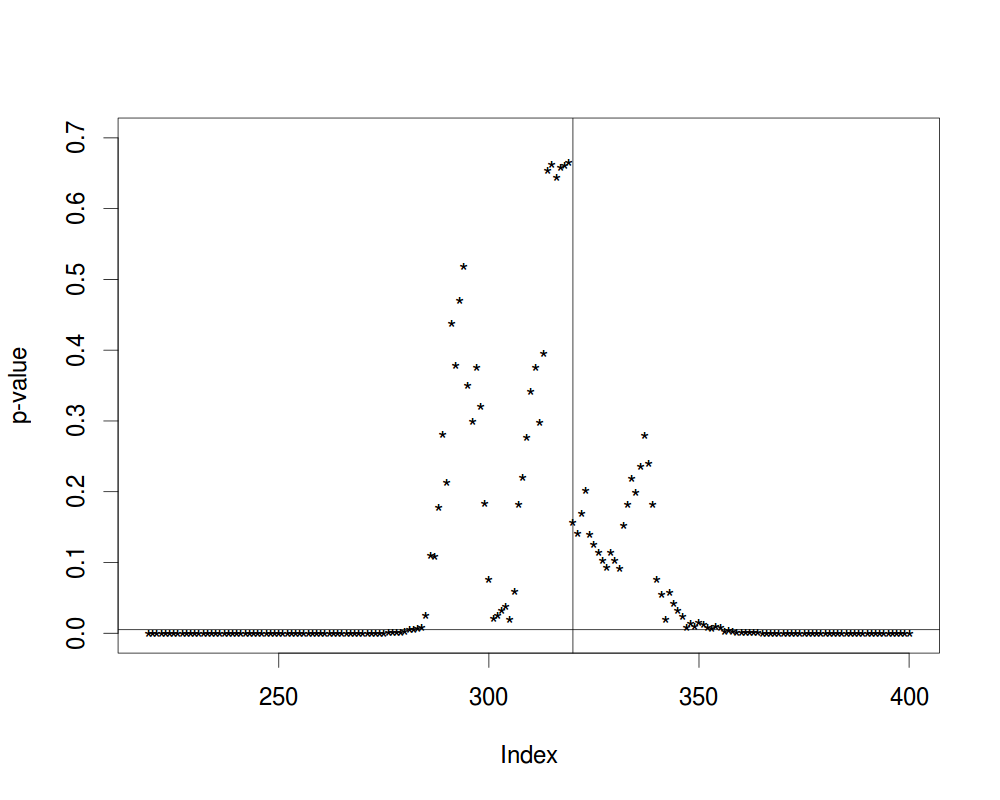}
  \caption{Illustration for computing the $C_{\alpha}$ confidence region for location of change in author in \emph{Tirant lo Blanc} on word length data with G being a MST.  The vertical line is the estimated change-point.  The $x$-axis is the indices on the chapters being used and $y$-axis is the $p$-value for the right- or left-subsequence (depending on whether the point is respectively before or after $\hat{\tau}$).  The horizontal line is at the value $1-\sqrt{1-0.01}$.}
  \label{fig:CI}
\end{figure}


\section{Conclusions and Discussion}
\label{sec:disc-concl}
The proposed method for change-point detection can be applied to a wide range of data, requiring only the existence of a dissimilarity measure on the sample space.  In applications, the choice of a good dissimilarity measure is critical, and domain knowledge should be used to design a measure that is sensitive to the signal of interest.  The graph-based approach in this paper decouples this modeling choice of dissimilarity measure from the formal test for a change-point.  Given the graph, the scan statistics are straightforward to compute, with general off-the-shelf analytic formulas for family-wise error control.

We have shown that the $p$-value approximations are quite accurate.  Our simulations were for a data sequence of length $n=1000$.  The accuracy of the approximations depend on the minimum allowed group size $n_0$ ($l_0$ for the changed interval alternative) and not so much on $n$.    Accuracy also depends on the structure of the graph.  When the graph is dominated by hubs, 
skewness correction is critical for the approximations to be accurate. For \emph{extremely} star-shaped graphs, we imagine that adjusting for kurtosis and higher order moments would also be helpful.  The strategy would be similar to skewness correction, but more technically complicated.  We don't compute these higher order terms in this paper, but if needed they can be computed in a similar fashion as the skewness term with the aid of a symbolic computation software.

If hubs dominate the topology of the graph, perturbation of any hub can change the topology drastically, and $R_G(t)$, which does not take into account the interaction between edges, loses all information regarding the high order structure.  Under such circumstances, the particular graph would not be useful for differentiating $F_1$ from $F_0$, and one would need to explore other dissimilarity measures and graph construction methods on the data.
\cite{radovanovic2010hubs} studied the hubbing phenomenon in high dimensional data under several similarity measures, which can serve as a starting point for choosing informative similarity measures for particular problems.

Compared to parametric approaches, the graph-based approach requires far fewer assumptions, but also makes less use of the data. Although this leads to loss of power in low dimensions if the data indeed follow the parametric model, it leads to robustness and wider applicability.  An important observation is that the graph-based approach has desirable power, compared to existing parametric tests, in moderate and high dimensions.  For high dimensional data, it is often hard to predict the direction and nature of the change.  Without such prior knowledge, parametric models would require the estimation of many parameters, most of which would be unrelated to the change.  For example, the Hotelling $T^2$ statistic requires the estimation of the large covariance matrix.  If, by prior knowledge or data pre-processing, we can circumvent the covariance estimation, then Hotelling $T^2$ would be preferable when the data satisfies its assumptions -- normality with no change of variance.  Otherwise, graph-based approaches gain increasing advantage over Hotelling's $T^2$ as $d$ increases, even in the problem for which Hotelling's $T^2$ was explicitly designed.

We explored three different ways of constructing the underlying graph given a dissimilarity measure.  From the numerical results and the analysis of the MIT cell phone network, we see that scans based on MST and NNG perform similarly, while scans based on MDP have lower power.  We suspect this is due to the fact that MDP is the least dense graph and utilizes the least amount of information from the original data set.  This is confirmed as the power increases when we use denser graphs (3-MST/MDP/NNG vs. 1-MST/MDP/NNG).  More study is needed to determine what is the optimal choice of graph. 
One may also consider assigning weights to the edges.  As in all problems, building more assumptions into the statistic leads to improved power if the assumptions are true, but sacrifices robustness.  

The analytic moment and significance formulas assume independent observations.  When there is local dependence, block permutations may be useful in producing more accurate $p$-values.  We illustrated this in Section \ref{sec:block-perm}.  Block permutation is computationally intensive, and in practice one always wrestles with the question of how to choose the block size.  When local dependence is weak, as for our data examples, the thresholds given by block permutation are quite close to the analytic thresholds that assume dependence.

A Cox-Spj{\o}tvoll type confidence region, as proposed by \cite{worsley1986confidence}, can be computed under this graph-based framework to assess the uncertainty in the estimation of the change-point.  As described in Section \ref{sec:conf-interv-estim}, we find Worsley's approach to be sometimes misleading in practice, and propose a modification that is conservative but more robust.  Our discussion focused on the inference for the chapter where authorship changed in \emph{Tirant lo Blanc}, because this seems to be a problem where the space of models is limited, and the interpretation of the change-point parameter is clear.



If more than one change-point or changed interval were of interest, the graph-based scan can be applied recursively in a procedure that is called binary or circular binary segmentation \citep{vostrikova1981detecting,olshen2004circular}.  



\bibliographystyle{imsart-nameyear}
\bibliography{graph_chpt}

\appendix

\section{Proofs for Lemmas and Propositions}
\label{sec:proofs-lemm-prop}

\subsection{Proof of Lemma \ref{lemma:onechptstat}} \label{woodroofe}
\label{sec:proof-lemma}

When observations $\by_i$ and $\by_j$ are linked in the graph, the edge is denoted as $(i,j)$. Then,
 \begin{align*}
    \bE(R_G(t)) & = \sum_{(i,j)\in G}\bP(g_i(t)\neq g_j(t)) = p_1(t)|G|,
  \end{align*}
because there are $2t(n-t)$ ways to place $i$ and $j$ on the two sides of $t$ among all $n(n-1)$ ways.

For the second moment,
  \begin{align*}
    \bE(R_G^2(t)) & = \sum_{(i,j),(k,l)\in G}\bP(g_i(t)\neq g_j(t), g_k(t)\neq g_l(t)).
  \end{align*}
By examining different ways of placing $i$, $j$, $k$, $l$, we have
\begin{align*}
  \bP(g_i(t) &\neq g_j(t),  g_k(t)\neq g_l(t))  \\
& = \left\{
    \begin{array}{ll}
      \frac{2t(n-t)}{n(n-1)} = p_1(t) & \text{ if } \left\{
        \footnotesize  \begin{array}{l}
            i = k,~j = l \\
            i = l,~j = k
          \end{array}
          \right. \\
      \frac{t(n-t)}{n(n-1)} = \frac{1}{2}p_1(t) & \text{ if } \left\{
       \footnotesize \begin{array}{l}
          i = k,~j\neq l \\
          i = l,~j\neq k \\
          j = k,~i\neq l \\
          j = l,~i\neq k
        \end{array}
        \right. \\
      \frac{4t(t-1)(n-t)(n-t-1)}{n(n-1)(n-2)(n-3)} = p_2(t)& \text{ if } i,j,k,l \text{ all different. }
    \end{array}
\right.
\end{align*}
So
\begin{align*}
  \bE(R_G^2(t)) & =  \sum_{(i,j)\in G} p_1(t) +  \sum_{\footnotesize \begin{array}{c} (i,j),(i,k)\in G \\ j\neq k \end{array}} \frac{1}{2}p_1(t) + \sum_{\footnotesize \begin{array}{c}(i,j),(k,l)\in G \\ i,j,k,l \text{ all different} \end{array}} p_2(t)  \\
  & = p_2(t) |G| + \left( \frac{1}{2}p_1(t) - p_2(t)\right) \sum_i|G_i|^2 + p_2(t)  |G|^2.
\end{align*}
$\bV(R_G(t))$ follows from $\bE(R_G^2(t)) - \bE^2(R_G(t))$.

\subsection{Proof of Theorem \ref{thm:limit}}
\label{sec:proofs-limit-distr}

We here prove that $\{Z_G([nu]):0<u<1\}$ converges to a Gaussian process.  The proof for the convergence of $\{Z_G([nu],[nv]):0<u<v<1\}$ to two-dimensional Gaussian random field can be done in the same manner but with a more careful treatment of the indices. 

To prove $\{Z_G([nu]):0<u<1\}$ converges to a Gaussian process, we only need to show that $(Z_G([nu_1]), Z_G([nu_2]), \dots, Z_G([nu_K]))$ becomes multivariate Gaussian as $n\rightarrow \infty$ for any $0<u_1<u_2<\dots <u_K<1$ and fixed K.  For simplicity, let $t_k=[nu_k], k=1,\dots,K$.

To prove that $(Z_G(t_1), Z_G(t_2), \dots, Z_G(t_K))$ is multivariate Gaussian, we take one step back.  In permutation distribution, we permute the order of the observations. Let $\pi(i)$ be the observed time of $\by_i$ after permutation, then $(\pi(1),\dots,\pi(n))$ is a permutation of $1,\dots,n$.  On the other hand, to obtain the permutation distribution, we can do it in two steps: 1) For each $i$, $\pi(i)$ is sampled uniformly from 1 to $n$; 2) only those that each value in $\{1,\dots,n\}$ is sampled once are retained.  It is easy to see that each permutation has the same occurrence probability after these two steps.  

 We call the distribution resulting from only performing the first step the bootstrap distribution, and we use $\bfPB$, $\bfEB$ and $\VB$ to denote the probability, expectation and variance, respectively. ($\bP$, $\bE$, $\bV$ without the subscript $\texttt{B}$ are used to denote the equivalences under the permutation distribution.) Let

 \begin{align*}
   Z_G^B(t) & = -\frac{R_G(t)-\bfEB(R_G(t))}{\sqrt{\VB(R_G(t))}}, \\
   X^B(t) & = \frac{n^B(t)-t}{\sqrt{t(1-t/n)}}, \text{ where }n^B(t) = \sum_{i=1}^n I_{\pi(i)\leq t}. 
 \end{align*}


Then following a similar argument in the proof for Lemma \ref{lemma:onechptstat} but replacing the permutation distribution with bootstrap distribution, we have
  \begin{align*}
    \bfEB(R_G(t)) & = p_1^B(t) |G|, \\
    \VB(R_G(t)) & = p_2^B(t) |G| + \left( \frac{1}{2}p_1^B(t) - p_2^B(t)\right) \sum_i|G_i|^2,
  \end{align*}
where
\begin{align*}
  p_1^B(t) & = \frac{2t(n-t)}{n^2}, &
  p_2^B(t) & = \frac{4t^2(n-t)^2}{n^4}. 
\end{align*}

  We prove the following two lemmas.
 
\begin{lemma}\label{lemma:bootstrap}
Under conditions \ref{cond:graph1} and \ref{cond:graph2}, for $0<u_1<u_2<\dots <u_k<1$, as $n\rightarrow\infty$, under the bootstrap distribution,
\begin{equation}\label{eq:star}
  (Z_G^B(t_1), Z_G^B(t_2),\dots,Z_G^B(t_K), X^B(t_1), X^B(t_2),\dots, X^B(t_K)) 
\end{equation}
 is multivariate normal and the covariance matrix of $$(X^B(t_1), X^B(t_2),\dots, X^B(t_K))$$ is positive definite.
\end{lemma}

\begin{lemma}\label{lemma:PBeq}
  When $|G|\sim o(n^2)$, for $t\sim \mathcal{O}(n)$, as $|G|\rightarrow \infty$, we have 
  \begin{enumerate}[1.]
  \item $$\frac{\VB(R_G(t))}{\bV(R_G(t))} \rightarrow 1.$$
  \item $$\frac{\bfEB(R_G(t))-\bE(R_G(t))}{\sqrt{\VB(R_G(t))}} \rightarrow 0. $$
  \end{enumerate}

\end{lemma}

From Lemma \ref{lemma:bootstrap}, $(Z_G^B(t_1), Z_G^B(t_2),\dots,Z_G^B(t_K)|X^B(t_1), X^B(t_2),\dots,X^B(t_K)$ is multivariate normal under the bootstrap distribution.  Since $(Z_G^B(t_1), Z_G^B(t_2),\allowbreak\dots,Z_G^B(t_K)|X^B(t_1)=0, X^B(t_2)=0,\dots,X^B(t_K)=0)$ under the bootstrap distribution has the same distribution as $(Z_G^B(t_1),Z^B(t_1),\dots,Z^B(t_K))$ under the permutation distribution, and notice that $$ Z_G(t) = \frac{\VB(R_G(t))}{\bV(R_G(t))} \left( Z_G^B(t) - \frac{\bfEB(R_G(t))-\bE(R_G(t))}{\sqrt{\VB(R_G(t))}} \right), $$ we conclude that
$(Z_G([nu_1]), Z_G([nu_2]), \dots, Z_G([nu_K]))$ is multivariate Gaussian under the permutation distribution.

We next prove the two lemmas.

\begin{proof}[Proof for Lemma \ref{lemma:bootstrap}]  To show that \eqref{eq:star} is multivariate normal, we only need to show that 
$\sum_{k=1}^K (a_k Z_G^B(t_k) + b_k X^B(t_k))$ is normal for any fixed $\{a_k\}$ and $\{b_k\}$.

If $\VB(\sum_{k=1}^K (a_k Z_G^B(t_k) + b_k X^B(t_k)))=0$, $\sum_{k=1}^K (a_k Z_G^B(t_k) + b_k X^B(t_k))$ is degenerating and we can claim any distribution for it.  For non-degenerating case, let $\sigma_0^2:=\VB(a_k Z_G^B(t_k) + b_k X^B(t_k)))$.  Then $\sigma_0\sim\mathcal{O}(1)$.  We prove the Gaussianity of $\sum_{k=1}^K (a_k Z_G^B(t_k) + b_k X^B(t_k))$  by the Stein's method.  

 Consider sums of the form
$W=\sum_{i\in{\cal J}} \xi_i,$
where $\mathcal{J}$ is an index set and $\xi$ are random variables with $E[\xi_i]=0$, and $E[W^2]=1$.  The following assumption restricts the dependence between $\{\xi_i:~i \in \mathcal{J}\}$.
\begin{assumption} \cite[p.\, 17]{chen2005stein}
  \label{assump:LD}
For each $i\in{\cal J}$ there exists $S_i \subset T_i \subset {\cal J}$ such that $\xi_i$ is independent of $\xi_{S_i^c}$ and $\xi_{S_i}$ is independent of $\xi_{T_i^c}$.
\end{assumption}
We will use the following specific form of Stein's method.
\begin{theorem}\label{thm:3.4} \cite[Theorem 3.4]{chen2005stein}
Under Assumption \ref{assump:LD}, we have
$$\sup_{h\in Lip(1)} |\bfE h(W) - \bfE h(Z)| \leq \delta$$
where $Lip(1) = \{h: \mathbb{R}\rightarrow \mathbb{R} \}$, $Z$ has ${\cal N}(0,1)$ distribution and
 $$\delta = 2 \sum_{i\in{\cal J}} (\bfE|\xi_i \eta_i\theta_i| + |\bfE(\xi_i\eta_i)|\bfE|\theta_i|) + \sum_{i\in{\cal J}} \bfE|\xi_i\eta_i^2|$$
with $\eta_i = \sum_{j\in S_i}\xi_j$ and $\theta_i = \sum_{j\in T_i} \xi_j$, where $S_i$ and $T_i$ are defined in Assumption \ref{assump:LD}.
\end{theorem}

We adopt the same notation with the index set $\mathcal{J} = \{G, 1,\dots,n\}$.

Let
$$\xi_{e,k} = \frac{I_{g_{\pi(e_-)}(t_k)\neq g_{\pi(e_+)}(t_k)} - p_1^B(t_k)}{\sigma^B(t_k)},$$
 Since $I_{g_{\pi(e_-)}(t_k)\neq g_{\pi(e_+)}(t_k)} \in \{0,1\}, p_1^B(t_k)\in (0,1)$, we have
$$|\xi_{e,k}| \leq \frac{1}{\sigma^B(t_k)}.$$
Let $$\xi_{i,k} = \frac{I_{\pi(i)\leq t_k} - u_k}{\sqrt{nu_k(1-u_k)}}.$$
Similarly, we have $$|\xi_{i,k}| \leq \frac{1}{\sqrt{nu_k(1-u_k)}}.$$
Let $\xi_e = \sum_k a_k \xi_{e,k}/\sigma_0$, $\xi_i = \sum_k b_k \xi_{i,k}/\sigma_0$, then $W = \sum_{j\in \mathcal{J}}\xi_j = \sum_k (a_k Z_G^B(t_k) + b_k X^B(t_k))/\sigma_0$, $\bfEB(W) = 0$, $\bfEB(W^2) = 1$. 
Let $a = \max(\max_k a_k, \max_k b_k)$, $\sigma = \min(\min_k\sigma^B(t_k), \min_k\sqrt{nu_k(1-u_k)})$.  Then
$$|\xi_j| \leq \frac{aK}{\sigma \sigma_0}, \quad \forall j\in \mathcal{J}.$$
For $e\in G$, let 
\begin{align*}
  S_e & = \{A_e, e^-, e^+\}, \\
  T_e & = B_e \cup \{\text{Nodes in } A_e\},
\end{align*}
where $A_e$, $B_e$ defined in \eqref{eq:Ae} and \eqref{eq:Be}.  Then $S_e$ and $T_e$ satisfy Assumption \ref{assump:LD}.

For $i=1,\dots,n$, let
\begin{align*}
  S_i & = G_i \\
  T_i & = G_{i,2} \cup \{\text{Nodes in } G_i\},
\end{align*}
where $G_{i,2}$ is the subgraph of $G$ including all edges connect to $G_i$.  Then $S_i$ and $T_i$ satisfy Assumption \ref{assump:LD}.

We have $|S_e| = |A_e|+2$, $|T_e| = |B_e| + |A_e|+1$, $|S_i| = |G_i|$, $|T_i| = |G_{i,2}| + |G_i| + 1$.

By Theorem \ref{thm:3.4}, we have $\sup_{h\in Lip(1)} |\bE h(W) - \bE h(Z)| \leq \delta$ for $Z \sim {\cal N}(0,1)$, where
\allowdisplaybreaks
\begin{align*}
  \delta & = 2 \sum_{j\in \mathcal{J}} (\bfE|\xi_j \eta_j\theta_j| + |\bfE(\xi_j\eta_j)|\bfE|\theta_j|) + \sum_{j\in \mathcal{J}} \bfE|\xi_j\eta_j^2| \\
& = 2 \sum_{e\in G} (\bfE|\xi_e \eta_e\theta_e| + |\bfE(\xi_e\eta_e)|\bfE|\theta_e|) + \sum_{e\in G} \bfE|\xi_e\eta_e^2| \\
& \quad \quad + 2 \sum_{i=1}^n (\bfE|\xi_i \eta_i\theta_i| + |\bfE(\xi_i\eta_i)|\bfE|\theta_i|) + \sum_{i=1}^n \bfE|\xi_i\eta_i^2| \\
& \leq \frac{a^3K^3}{\sigma^3\sigma_0^3} \left(\sum_{e\in G} 5(|A_e|+2)(|B_e|+|A_e|+1) + \sum_{i=1}^n 5 |G_i|(|G_{i,2}|+|G_i|+1) \right) \\
& \leq \frac{a^3K^3}{\sigma^3\sigma_0^3} \left(45\sum_{e\in G}|A_e||B_e| + 15\sum_{i=1}^n |G_i||G_{i,2}| \right)
\end{align*} 
Observe that if $e=(i,j)$, then $G_i, G_j \subseteq A_e$, $G_{i,2}, G_{j,2} \subseteq B_e$.  For each node $\by_i$, we can randomly pick an edge $e$ that connects node $\by_i$, and we have $|G_i||G_{i,2}|\leq |A_e||B_e|$.  Each node in the graph can be picked twice in maximum since an edge connects two nodes, therefore,
$$\sum_{i=1}^n |G_i||G_{i,2}| \leq 2\sum_{e\in G}|A_e||B_e|.$$  So
$$\delta\leq \frac{75 a^3K^3}{\sigma^3\sigma_0^3} \sum_{e\in G}|A_e||B_e|.$$
 
Notice that $\sigma\sim\mathcal{O}(\min(n^{1/2},|G|^{1/2}))$.  When $|G|\sim\mathcal{O}(n^\alpha)$, we have $\sigma\sim\mathcal{O}(n^{0.5(\alpha\wedge 1)})$.  When $\sum_{e\in G}|A_e||B_e| \sim o(n^{1.5(\alpha\wedge 1)})$,  we have $\delta\rightarrow 0$ as $n\rightarrow \infty$.

Let $\Sigma_X$ be the covariance matrix of $(X^B(t_1),X^B(t_2),\dots,X^B(t_K))$.  It's not hard to derive that for $i\leq j$,
$$\Sigma_X(i,j) = \Sigma_X(j,i) = \frac{t_i(1-t_j/n)}{\sqrt{t_i(1-t_i/n)t_j(1-t_j/n)}}, $$
where $t_{K+1} \overset{\Delta}{=} n$.

$\Sigma_X^{-1}$ admits the Cholesky decomposition
$$\Sigma_X^{-1} = L L^\prime,$$
where
\begin{align*}
  L(i,j) = \left\{ \begin{array}{ll} \frac{\sqrt{1-t_i/n}}{\sqrt{1-t_i/t_{i+1}}}, & j=i; \\ -\frac{t_i \sqrt{1-t_i/n}}{t_{i+1} \sqrt{1-t_i/t_{i+1}}}, & j=i-1; \\ 0, & \text{otherwise.}  \end{array}\right.
\end{align*}

Therefore, $$|\Sigma_X| = \frac{\prod_{k=1}^K (1-t_k/t_{k+1})}{\prod_{k=1}^K(1-t_k/n)},$$
is positive definite.

\end{proof} 

\begin{proof}[Proof for Lemma \ref{lemma:PBeq}]
Let $u=\lim_{n\rightarrow\infty} t/n$, then
\begin{align*}
  \lim_{n\rightarrow \infty} p_1(t) & = \lim_{n\rightarrow \infty} p_1^B(t) = 2u(1-u), \\
  \lim_{n\rightarrow \infty} p_2(t) & = \lim_{n\rightarrow \infty} p_2^B(t) = 4u^2(1-u)^2, \\
  \lim_{n\rightarrow \infty} \bV(R_G(t)) & = \lim_{n\rightarrow \infty} \VB(R_G(t)) = 4u^2(1-u)^2|G| + u(1-u)(1-2u)^2 \sum_i|G_i|^2.
\end{align*}
So $$\frac{\VB(R_G(t))}{\bV(R_G(t))} \rightarrow 1.$$
Since
\begin{align*}
  \bfEB(R_G(t))-\bE(R_G(t)) & = (p_1^B(t)-p_1(t))|G| = -\frac{2t(n-t)}{n^3}|G|,
\end{align*}
we have
\begin{align*}
  \lim_{n\rightarrow \infty} & \frac{\bfEB(R_G(t))-\bE(R_G(t))}{\sqrt{\VB(R_G(t))}}  = - \lim_{n\rightarrow \infty} \frac{2u(1-u)|G|/n}{\sqrt{4u^2(1-u)^2|G| + u(1-u)(1-2u)^2 \sum_i|G_i|^2}} \\
& = -\lim_{n\rightarrow \infty} \frac{2u(1-u)}{\sqrt{4u^2(1-u)^2n^2/|G| + u(1-u)(1-2u)^2 n^2\sum_i|G_i|^2/|G|^2}},
\end{align*}
which is 0 when $|G|\sim o(n^2)$.

\end{proof}


\subsection{Proof of Lemma \ref{lemma:rho}}
\label{sec:proof-lemma-rho}

First observe that $\rho^\star_G(u,u)=1$, which holds for \eqref{eq:rho}.  Because of the interchangeability of $u$ and $v$ in the definition of $\rho_G(u,v)$, it is enough to show that when $u<v$, 
\begin{equation}
  \label{eq:rhouv}
  \rho^\star_G(u,v) = \frac{2u^2(1-v)^2|G| + u(1-v))(1-2u)(1-2v)\sum_i|G_i|^2 }{\sigma^\star_G(u) \sigma^\star_G(v)}.
\end{equation}
  Let $\rho_{G,n}(u,v) \overset{\Delta}{=}  \cov(Z_G([nu]), Z_G([nv]))$, then $\rho_G(u,v) = \lim_{n\rightarrow\infty} \rho_{G,n}(u,v)$. 
 Let $s=[nu]$, $t=[nv]$, then $s<t$, and $\lim_{n\rightarrow\infty} s/n = u$, $\lim_{n\rightarrow\infty} t/n = v$.
Since 
\begin{align*}
  \cov(Z_G(s),Z_G(t)) & = \frac{\bE(R_G(s)R_G(t)) -\bE(R_G(s))\bE(R_G(t))}{\sqrt{\bV(R_G(s))\bV(R_G(t))}},
\end{align*}
where the expressions for $\bE(R_G(s))$, $\bE(R_G(t))$, $\bV(R_G(s))$, $\bV(R_G(t))$ can be found in Lemma \ref{lemma:onechptstat}, we only need to figure out 
\begin{align*}
  \bE(R_G(s)R_G(t))  = \sum_{(i,j),(k,l)\in G} \bP(g_i(s)\neq g_j(s), g_k(t)\neq g_l(t)).
\end{align*}
By examining different ways of placing $i$, $j$, $k$, $l$, we have
\begin{align*}
  \bP&[g_i(s) \neq g_j(s),  g_k(t)\neq g_l(t)]  \\
& = \left\{
    \begin{array}{ll}
      \frac{2s(n-t)}{n(n-1)} := q_1(s,t) & \text{ if } \left\{\footnotesize
          \begin{array}{l}
            i = k, j = l \\
            i = l, j = k
          \end{array}
          \right. \\
      \frac{s(n-t)(n+2t-2s-2)}{n(n-1)(n-2)} := q_2(s,t) & \text{ if } \left\{ \footnotesize
        \begin{array}{l}
          i = k, j\neq l \\
          i = l, j\neq k \\
          j = k, i\neq l \\
          j = l, i\neq k
        \end{array}
        \right. \\
      \frac{4 s(n-t)[(s-1)(n-s-1) + (t-s)(n-s-2)]}{n(n-1)(n-2)(n-3)} := q_3(s,t) & \text{ if } i,j,k,l \text{ all different. }
    \end{array}
\right.
\end{align*}
Then
\begin{align*}
  &\bE(R_G(s)R_G(t))  = \sum_{(i,j),(k,l)\in G} \bP(g_i(s)\neq g_j(s), g_k(t)\neq g_l(t))\\
& =  \sum_{(i,j)\in G} q_1(s,t) + \sum_{\footnotesize \begin{array}{c} (i,j),(i,k)\in G \\ j\neq k \end{array}} q_2(s,t) + \sum_{\footnotesize \begin{array}{c} (i,j),(k,l)\in G \\ i,j,k,l \text{ all different } \end{array}} q_3(s,t)   \\
& =  (q_1(s,t) - 2 q_2(s,t) +  q_3(s,t)) |G| + (q_2(s,t) - q_3(s,t)) \sum_{i=1}^n |G_i|^2  + q_3(s,t)|G|^2.
\end{align*}
So
\begin{align*}
  \lim_{n\rightarrow\infty}\bE(R_G(s)R_G(t)) & = 4u^2(1-v)^2|G|+u(1-v)(1-2u)(1-2v)\sum_{i=1}^n|G_i|^2 \\
  & \quad + 4uv(1-u)(1-v)|G|^2.
\end{align*}
Together with
\begin{align*}
  \lim_{n\rightarrow\infty}\bE(R_G(s)) & = 2u(1-u)|G|, \\
  \lim_{n\rightarrow\infty}\bV(R_G(s)) & = 4u^2(1-u)^2|G| + u(1-u)(1-2u)^2\sum_{i=1}^n|G_i|^2, 
\end{align*}
and similar for $R_G(t)$, we have \eqref{eq:rhouv}.

\subsection{Proof of Proposition \ref{thm:asym}}
\label{sec:proof-prop-refthm}

We first show the single change-point case.  We adopt Woodroofe's method \citep{woodroofe1976frequentist, woodroofe1978large} by condition on the first cross-over.
\begin{align}
&  \bP(\max_{n_0\leq t \leq n_1} Z_G^\star(t/n) > b) \nonumber \\
 & = \sum_{n_0\leq t\leq n_1} \int_0^\infty \bP(Z_G^\star(t/n) = b+dx) \bP(\max_{n_0\leq s<t } Z_G^\star(s/n)<b | Z_G^\star(t/n) = b+dx )  \label{proofline1} 
\end{align}
By change of measure and rearranging the terms, we have
\begin{align*}
  &  \bP(\max_{n_0\leq t \leq n_1} Z_G^\star(t/n) > b) \\
& = \frac{\phi(b)}{b} \sum_{n_0\leq t \leq n_1} \int_0^\infty e^{-x-\frac{x^2}{2b^2}} \bP(\max_{n_0\leq s<t }  b(Z_G^\star(s/n)-Z_G^\star(t/n))<-x | Z_G^\star(t/n) = b+\frac{x}{b}) dx.
\end{align*}
Since $b\rightarrow\infty$, if $x\sim o(b^2)$, then $\frac{x^2}{2b^2}$ is negligible to $x$ and $\frac{x}{b}$ is negligible to $b$; while if $x\sim \mathcal{O}(b)$, then $x+\frac{x^2}{2b^2} \rightarrow \infty$, and the integrand becomes 0, so  
\begin{align*}
  &  \bP(\max_{n_0\leq t \leq n_1} Z_G^\star(t/n) > b) \\
& \approx \frac{\phi(b)}{b} \sum_{n_0\leq t \leq n_1} \int_0^\infty e^{-x} \bP(\max_{n_0\leq s<t }  b(Z_G^\star(s/n)-Z_G^\star(t/n))<-x | Z_G^\star(t/n) = b) dx.
\end{align*}



Notice that for $u<v$,
$$b(Z_G^\star(u)-Z_G^\star(v))|(Z_G^\star(v)=b) \sim \mathcal{N}((\rho_G(u,v) -1)b^2, (1-\rho_G^2(u,v))b^2).$$


Let $\delta=v-u$, by Taylor expansion, we have
\begin{align*}
  \rho_G(u,v) & = 1 + f_{v,-}^\prime(0) \delta + f_{v,-}^{\prime\prime}(0) \delta^2/2 + \mathcal{O}(\delta^3), \\
  \rho_G^2(u,v) & = 1 + 2 f_{v,-}^\prime(0) \delta + ((f_{v,-}^\prime)^2 + f_{v,-}^{\prime\prime}(0))\delta^2 + \mathcal{O}(\delta^3).
\end{align*}

So for $\delta\sim \mathcal{O}(n^{-1})$,
$$b(Z_G^\star(u)-Z_G^\star(v))|(Z_G^\star(v)=b) \sim \mathcal{N} (-f_{v,-}^\prime(0)|\delta| b^2, 2f_{v,-}^\prime(0) |\delta|  b^2) .$$

One can show that, for $b = b_0\sqrt{n}$, and $n\rightarrow \infty$,
$$\lim_{k\rightarrow\infty} \limsup_{n\rightarrow\infty} \sum_{|i-t|>k} \bP(Z_G^\star(i/n)>b|Z_G^\star(t/n)=b+dx) = 0. $$



Let $W_m^{(t)}$ be a random walk with $W_1^{(t)}\sim \mathcal{N}(\mu^{(t)}, (\sigma^{(t)})^2)$, where $\mu^{(t)} = \frac{1}{n}f_{v,-}^\prime(0) b^2, (\sigma^{(t)})^2 = 2\mu^{(t)}$.  Then
\begin{align*}
  \bP(\max_{n_0\leq s<t }  b(Z_G^\star(s/n)-Z_G^\star(t/n))<-x | Z_G^\star(t/n) = b)& \sim \bP(\max_{n_0\leq s\leq t} -W_{t-s}^{(t)}<-x) \\
  &\sim \bP(\min_{m\geq 1} W_m^{(t)}>x).
\end{align*}

Together with the fact $$\int_0^\infty \exp{-2\mu x/\sigma} \bP(\min_{m\geq 1} W_m >x) dx =\mu \nu(2\mu/\sigma),$$
for a random walk $W_1\sim\mathcal{N}(\mu,\sigma)$ (see \citet{siegmund1992tail}), we have
\begin{align*}
  \lim_{n\rightarrow\infty}\bP(\max_{n_0\leq t \leq n_1} Z_G^\star(t/n) > b) & \approx \lim_{n\rightarrow\infty} \frac{\phi(b)}{b} \sum_{n_0\leq t \leq n_1} b_0^2 f_{t/n,-}^\prime(0) \nu (b_0\sqrt{2f_{t/n,-}^\prime(0)}) 
\end{align*}

For $f_{t/n,-}^\prime(0)$, we take the derivative of $\rho_G^\star(u,v)$, and after some tedious calculation, we have 
\begin{equation}
  \label{eq:rho0}
  f_{v,-}^\prime(0) = \frac{1}{2v(1-v)} + \frac{2}{4v(1-v)+(1-2v)^2(\sum_i|G_i|^2/|G|-4|G|)}.
\end{equation}

Putting everything together, we have
\begin{align*}
  \lim_{n\rightarrow\infty}\bP(\max_{n_0\leq t \leq n_1} Z_G^\star(t/n) > b) & \approx \lim_{n\rightarrow\infty} \frac{\phi(b)}{b} \sum_{n_0\leq t \leq n_1} b_0^2 h^*_{r_0,r_1}(t/n) \nu (b_0\sqrt{2h^*_{r_0,r_1}(t/n)}) \\
 & = \frac{\phi(b)}{b} \int_{x_0}^{x_1} b_0^2 h^*_{r_0,r_1}(x) \nu (b_0\sqrt{2h^*_{r_0,r_1}(x)}) n dx \\
 & =  b \phi(b) \int_{x_0}^{x_1} h^*_{r_0,r_1}(x)\nu\left(b_0\sqrt{2h^*_{r_0,r_1}(x)}\right)dx.
\end{align*}

Now, we show the changed interval case following the method of \citet{siegmund1988approximate,siegmund1992tail}.  We omit most of the technical details, which follow these two papers given that 
$$\rho^\star_{G, (u_1,u_2)}(\delta_1,\delta_2) \overset{\Delta}{=} \cov(Z_G^\star(u_1-\delta_1, u_2-\delta_2), Z_G^\star(u_1,u_2)). $$
 is differentiable with the derivative being continuous except at $\delta_1=0$ and at $\delta_2=0$.

A key intermediate form  is
\begin{align}
 & \bP \left(\max_{n_0\leq t_2-t_1 \leq n_1} Z_G^\star(t_1/n,t_2/n) > b \right)  \nonumber \\
  & \approx \frac{\phi(b)}{b} \sum_{ n_0 \leq t_2-t_1 \leq n_1 }  C_1(t_1,t_2)b^2 C_2(t_1,t_2) b^2 \times \nu\left(\sqrt{2C_1(t_1,t_2)b^2}\right) \nu\left(\sqrt{2C_2(t_1,t_2)b^2}\right), \label{intermediate}
\end{align}
 where $C_1$, $C_2$ are the partial derivatives $$C_1(nu_1,nu_2) \equiv \frac{1}{n} \left.\frac{\partial_- \rho^\star_{G,(u_1,u_2)}(\delta_1,0)}{\partial \delta_1}\right|_{\delta_1=0} = - \frac{1}{n} \left.\frac{\partial_+ \rho^\star_{G,(u_1,u_2)}(\delta_1,0)}{\partial \delta_1}\right|_{\delta_1=0} ,$$
$$C_2(nu_1,nu_2) \equiv -\frac{1}{n} \left.\frac{\partial_+ \rho^\star_{G,(u_1,u_2)}(0, \delta_2)}{\partial \delta_2}\right|_{\delta_2=0}.$$
Under the permutation null, the processes derived from perturbation of the left and right end points,
$$ Z_G^\star((t_1+k)/n,t_2/n), \quad k = \dots,-2,-1,0,1,2,\dots$$
and
$$ Z_G^\star(t_1/n,(t_2-k)/n), \quad k = \dots,-2,-1,0,1,2,\dots,$$
are identical in distribution to the process
$$ Z_G^\star((t_2-t_1-k)/n), \quad k = \dots,-2,-1,0,1,2,\dots,$$
Thus, the partial derivatives are equal to the derivative in the one change-point scenario, $$C_1(t_1,t_2)=C_2(t_1,t_2) = \frac{1}{n}f_{u_2-u_1,-}^\prime(0).$$  Substituting $\frac{1}{n}f_{u_2-u_1,-}^\prime(0)$  for $C_1(t_1,t_2)$ and $C_2(t_1,t_2)$ in \eqref{intermediate} and the double summation goes to an integral as $n\rightarrow\infty$ yields \eqref{eq:twochptp1star}.

\section{Skewness Correction}
\label{sec:skewness-correction-1}

\subsection{Derivation of \eqref{eq:onechptp2} and \eqref{eq:twochptp2}}
\label{sec:skewness}

We first show how to approximate the marginal probability $\bP(Z_G(t)\in b+dx/b)$ better by incorporating skewness.  In the derivation below we suppress the dependence on the graph $G$ and the time parameter $t$. Consider the probability measure $d \bQ_\theta = e^{\theta Z - \psi(\theta)} d\bP$, where $\psi(\theta) = \log\bE_{\bP}(e^{\theta Z}).$  Choose $\theta_b$ such that $\dot{\psi}(\theta_b) = \bE_{\bQ_{\theta_b}}(Z) = b$.  Then,
\begin{align}
 \bP(Z\in b+dx/b) & = \bE_\bP(\mathbf{1}_{Z\in b+dx/b})  \approx e^{-\theta_b (b+x/b) + \psi(\theta_b)} \bQ_{\theta_b} (Z\in b+dx/b). \label{eq:Qb}
\end{align}
Since under $\bQ_{\theta_b}$, $Z$ is centered at $b$ with variance $\overset{..}{\psi}(\theta_b)$, $\bQ_{\theta_b}(Z \in b+dx/b)$ can be approximated by the normal density,
\begin{equation}\label{Qthetab} \bQ_{\theta_b} (Z\in b+dx/b) \approx \frac{1}{\sqrt{2\pi \overset{..}{\psi}(\theta_b)}} \exp\left(-\frac{x^2}{2b^2 \overset{..}{\psi}(\theta_b)}\right) \approx \frac{1}{\sqrt{2\pi \overset{..}{\psi}(\theta_b)}}. \end{equation}
The second approximation above is accurate for $x/b \rightarrow 0$.

To obtain $\psi(\theta_b)$ and $\overset{..}{\psi}(\theta_b)$, we use Taylor expansions, noting that  $\psi(0)=\dot{\psi}(0) =0, \overset{..}{\psi}(0) = 1, \overset{\dots}{\psi}(0) = \bE_\bP(Z^3) \overset{\Delta}{=} \gamma$:
\begin{align} \label{psitheta}
  \psi(\theta) & \approx \psi(0) + \dot{\psi}(0) \theta + \overset{..}{\psi}(0) \frac{\theta^2}{2} + \overset{\dots}{\psi}(0) \frac{\theta^3}{6} = \frac{\theta^2}{2}\left( 1+ \frac{\gamma \theta}{3}\right), \\
\overset{..}{\psi}(\theta) & \approx \overset{..}{\psi}(0) + \overset{\dots}{\psi}(0) \theta = 1 + \gamma \theta. \label{psidotdottheta}
\end{align}
Combining \eqref{eq:Qb}, \eqref{Qthetab},\eqref{psitheta} and \eqref{psidotdottheta} gives
\begin{align} \label{skewcorrectedmarg}
\bP(Z\in b+dx/b) & \approx \frac{1}{\sqrt{2\pi (1 + \gamma \theta_b) }} \exp(-\theta_b b - x \theta_b /b + \theta_b^2 (1+ \gamma \theta_b/3)/2).
\end{align}
For an approximation of $\theta_b$, we solve $\dot{\psi}({\theta_b})$ by approximating $\psi$ up to the third order,
\begin{align}
b=\dot{\psi}(\theta_b) & \approx \dot{\psi}(0) + \overset{..}{\psi}(0) \theta_b + \overset{\dots}{\psi}(0) \frac{\theta_b^2}{2} = \theta_b + \frac{1}{2}\gamma \theta_b^2, \label{thetaeq}
\end{align}
yielding
\begin{equation}\label{thetab}
  \theta_b \approx (-1+\sqrt{1+2\gamma b})/\gamma.
\end{equation}
Note that when $\gamma=0$, $\theta_b=b$.  \eqref{eq:onechptp2} follows by using \eqref{skewcorrectedmarg} in (\ref{proofline1}) in the proof  of Theorem \ref{thm:asym} and approximating the $\theta_bx/b$ term in the exponent by $x$.  

The derivation for \eqref{eq:twochptp2} is similar but to give a better approximation to $\bP(Z_G(t_1,t_2)\in b+dx/b)$ by incorporating skewness.

\subsection{Effect of Skewness and Extrapolation at Boundaries}
\label{sec:extr-skewn-corr}

To gain a better understanding of the role of skewness, we explore the following quantities involved in the $p$-value approximations:
\begin{itemize}
\item $\gamma_G(t) \overset{\Delta}{=} \bE[Z_G^3(t)]$,
\item $\theta_{b,G}(t) \overset{\Delta}{=} (-1+\sqrt{1+2\gamma_G(t) b})/\gamma_G(t)$,
\item $S_G(t) \overset{\Delta}{=}  \frac{1}{\sqrt{1 + \gamma_G(t) \theta_{b,G}(t) }} \exp(\frac{1}{2}(b-\theta_{b,G}(t))^2 + \frac{\gamma_G(t) \theta_{b,G}(t)^3}{6})$.  
\end{itemize}
Figure \ref{fig:mdp} shows the three quantities versus $t$ for the single change-point scan statistic on a MDP graph when $n=1000, b=3$.  Since the structure of MDP is always the same and does not depend on the distribution of $\by_i$, Figure \ref{fig:mdp} is representative of all MDP graphs with $n=1000$ subjects and threshold $b=3$.  We can see from the figure that $\gamma$ is always larger than 0, indicating right skewness.  When $\gamma=0$, $\theta_b = b$; when $\gamma>0$, $\theta_b<b$.  When $Z_G(t)$ is right-skewed, the analytic approximation of the $p$-value assuming Gaussianity is smaller than the actual $p$-value, so the skewness correction should increase the $p$-value approximation.   This is indeed true as $S_G(t)$ is U-shaped with a minimum of 1.

Each node in the MDP has degree 1.  The shapes of $\gamma_G(t)$ and $\theta_{b,G}(t)$ for $Z_G(t)$ computed on graphs with very low number of hubs are similar to their shapes for $Z_G(t)$ computed on the MDP.
For example, for data in low dimensions ($<5$), scans based on MST and NNG constructed based on Euclidean distance have similar skewness properties as described above.  However, as the dimension of the data increases, MST and NNG constructed based on Euclidean distance tend to become dominated by hubs, and the distribution of $Z_G(t)$ becomes left-skewed.  For a left-skewed distribution, $\gamma\leq 0, \theta_b\geq b$, and $S \leq 1$.  One problem for left-skewed distributions is that if $\gamma$ is smaller than $-1/(2b)$, the current approximation does not yield real-valued solution for $\theta_b$.  This issue is discussed in Remark \ref{remark:extra} and here we provide a heuristic solution to this problem based on an extrapolation procedure.

We illustrate procedure through a MST constructed on a simulated 100-dimensional data based on Euclidean distance.  From Figure \ref{fig:mst_100d}, we see that $\theta_{b,G}(t)$ and $S_G(t)$ are not defined except in the middle region.  In this case, the integrand $$ S_G(nu)  h_G(n,x)\nu\sqrt{2b_0^2 h_G(n,x)}$$ is directly extrapolated to the edge regions using the boundary tangent at each side.  If extrapolation is negative, it is set to zero.  Figure \ref{fig:mst_integrand} illustrates the integrand before and after extrapolation.




\begin{figure}[!htp]
  \centering
  \includegraphics[width=.32\textwidth]{mdp_r.png}
  \includegraphics[width=.32\textwidth]{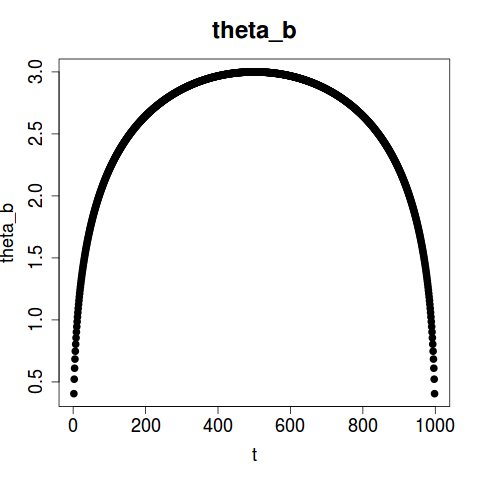}
  \includegraphics[width=.32\textwidth]{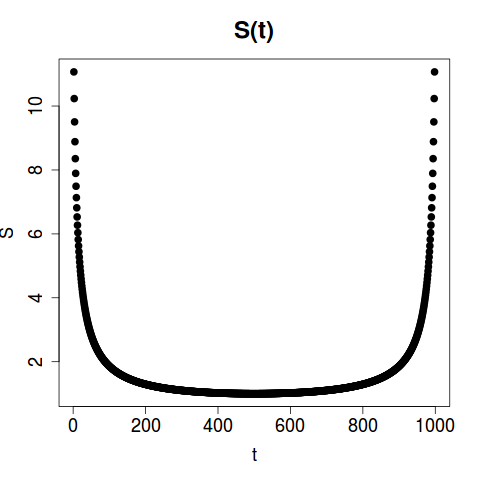} \\
  \caption{The three quantities, $\gamma_G(t), \theta_{b,G}(t)$ and $S_G(t)$ from left to right,  for a MDP graph. $n=1000, b=3$. }
  \label{fig:mdp}
\end{figure}


\begin{figure}[!htp]
  \centering
  \includegraphics[width=.32\textwidth]{mst_r_100.png}
  \includegraphics[width=.32\textwidth]{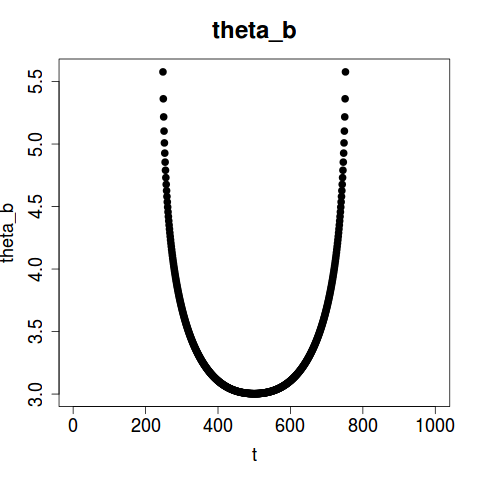}
  \includegraphics[width=.32\textwidth]{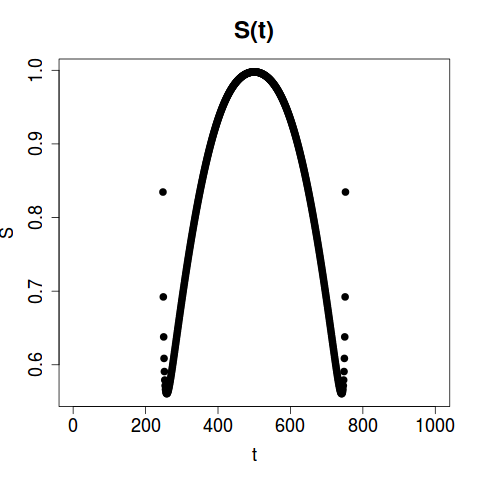}
  \caption{The three quantities, $\gamma_G(t), \theta_{b,G}(t)$ and $S_G(t)$ from left to right, for a MST graph constructed using Euclidean distance on a sequence of $n=1000$ observations \emph{iid} drawn from $N(\mathbf{0}, I_{100})$.  $b=3$. }
  \label{fig:mst_100d}
\end{figure}

\begin{figure}[!htp]
\centering
  \includegraphics[width=.4\textwidth]{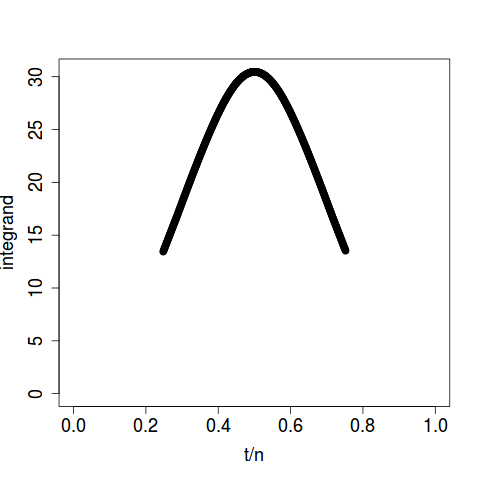}
  \includegraphics[width=.4\textwidth]{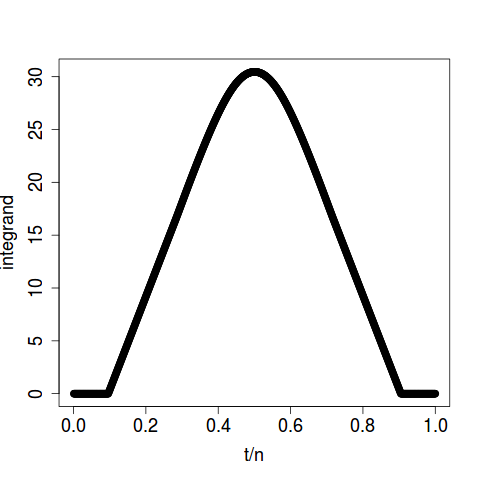}
  \caption{The integrand before (left) and after (right) extrapolation.  The integrand can only be directly calculated in the middle part ($t\in [248,752]$), and the outer part is obtained by extending using the boundary tangent.}
  \label{fig:mst_integrand}
\end{figure}


\section{Checking Analytic Approximations to $p$-values}
\label{sec:check-analyt-appr}

\subsection{Critical Value}
\label{sec:critical-value}

\subsubsection{A Single Change-Point Alternative (NNG)}
\label{sec:single-change-point-1}

Tables \ref{tab:nng_05} - \ref{tab:nng_01} show the results of $p$-value approximations for the single change-point alternative with the underlying graph being the NNG.  We see that the result is quite similar to that based on MST in Tables \ref{tab:mst_05} - \ref{tab:mst_01}.
 
\begin{table}[!htp]
  \caption{Critical values for the single change-point scan statistic based on NNG at 0.05 significance level.  $n=1000$.}
  \label{tab:nng_05}
  \centering
\begin{tabular}{c|ccc|ccc|ccc|cc}
\hline \hline
& \multicolumn{9}{|c|}{Critical Values} & \multicolumn{2}{|c}{Graph} \\ \cline{2-12}
& \multicolumn{3}{|c|}{$n_0=100$} & \multicolumn{3}{|c|}{$n_0=50$} & \multicolumn{3}{|c|}{$n_0=25$} & & \\ \cline{2-10}
  & A1 & A2 & Per & A1 & A2 & Per & A1 & A2 & Per & $\sum |G_i|^2$ & $d_{\text{max}}$ \\ \hline \hline
  & 2.96 & 2.98 & 2.95 & 3.04 & 3.07 & 3.03 & 3.10 & 3.13 & 3.08 &  2008 &  2 \\
N(0,1) & 2.96 & 2.98 & 2.97 & 3.05 & 3.07 & 3.05 & 3.10 & 3.13 & 3.09 &  1972 &  2 \\
$d=1$ & 2.96 & 2.98 & 3.01 & 3.04 & 3.07 & 3.10 & 3.10 & 3.12 & 3.13 &  2032 &  2 \\
  & 2.96 & 2.98 & 2.97 & 3.04 & 3.07 & 3.04 & 3.10 & 3.13 & 3.09 &  2008 &  2 \\
  & 2.96 & 2.99 & 3.01 & 3.05 & 3.08 & 3.10 & 3.10 & 3.13 & 3.13 &  1954 &  2 \\ \hline
  & 2.96 & 2.99 & 2.98 & 3.05 & 3.08 & 3.08 & 3.10 & 3.13 & 3.11 &  1948 &  2 \\
Exp(1) & 2.96 & 2.98 & 2.96 & 3.04 & 3.07 & 3.07 & 3.10 & 3.13 & 3.11 &  2038 &  2 \\
$d=1$ & 2.96 & 2.98 & 2.96 & 3.04 & 3.08 & 3.05 & 3.10 & 3.13 & 3.09 &  2014 &  2 \\
  & 2.96 & 2.98 & 2.99 & 3.04 & 3.07 & 3.08 & 3.10 & 3.12 & 3.13 &  2008 &  2 \\
  & 2.96 & 2.98 & 2.99 & 3.04 & 3.08 & 3.08 & 3.10 & 3.13 & 3.13 &  2038 &  2 \\ \hline
  & 2.94 & 2.92 & 2.89 & 3.02 & 2.97 & 2.93 & 3.07 & 3.00 & 2.96 &  3370 &  6 \\
N(0,1) & 2.94 & 2.91 & 2.90 & 3.02 & 2.97 & 2.95 & 3.07 & 2.99 & 2.96 &  3502 &  6 \\
$d=10$ & 2.94 & 2.91 & 2.89 & 3.01 & 2.96 & 2.95 & 3.06 & 2.98 & 2.96 &  3444 &  7 \\
  & 2.94 & 2.91 & 2.91 & 3.01 & 2.96 & 2.94 & 3.06 & 2.98 & 2.96 &  3436 &  6 \\
  & 2.94 & 2.91 & 2.88 & 3.02 & 2.97 & 2.93 & 3.07 & 2.99 & 2.94 &  3330 &  6 \\ \hline
  & 2.94 & 2.92 & 2.91 & 3.02 & 2.98 & 2.96 & 3.07 & 3.00 & 2.98 &  3144 &  5 \\
Exp(1) & 2.94 & 2.92 & 2.92 & 3.02 & 2.98 & 2.97 & 3.07 & 3.00 & 2.99 &  3096 &  6 \\
$d=10$ & 2.94 & 2.92 & 2.92 & 3.02 & 2.98 & 2.98 & 3.07 & 3.01 & 3.01 &  3118 &  6 \\
  & 2.94 & 2.93 & 2.92 & 3.02 & 2.98 & 2.97 & 3.07 & 3.01 & 2.99 &  3114 &  5 \\
  & 2.94 & 2.92 & 2.91 & 3.02 & 2.98 & 2.98 & 3.07 & 3.01 & 3.00 &  3152 &  6 \\ \hline
  & 2.87 & 2.65 & 2.62 & 2.95 & 2.65 & 2.62 & 3.00 & 2.65 & 2.62 &  9382 & 52 \\
N(0,1) & 2.87 & 2.73 & 2.70 & 2.95 & 2.75 & 2.71 & 3.01 & 2.76 & 2.71 &  8466 & 24 \\
$d=100$ & 2.88 & 2.76 & 2.72 & 2.96 & 2.78 & 2.72 & 3.01 & 2.79 & 2.72 &  7756 & 20 \\
  & 2.86 & 2.59 & 2.56 & 2.94 & 2.59 & 2.56 & 3.00 & 2.59 & 2.56 & 11092 & 68 \\
  & 2.87 & 2.68 & 2.64 & 2.95 & 2.69 & 2.64 & 3.00 & 2.69 & 2.64 &  9538 & 38 \\ \hline
  & 2.86 & 2.71 & 2.70 & 2.95 & 2.72 & 2.70 & 3.00 & 2.73 & 2.70 & 10222 & 34 \\
Exp(1) & 2.86 & 2.72 & 2.68 & 2.95 & 2.74 & 2.69 & 3.00 & 2.74 & 2.69 & 10390 & 37 \\
$d=100$ & 2.86 & 2.70 & 2.64 & 2.94 & 2.71 & 2.64 & 3.00 & 2.71 & 2.64 & 11574 & 35 \\
  & 2.87 & 2.74 & 2.72 & 2.95 & 2.76 & 2.73 & 3.01 & 2.77 & 2.73 &  8782 & 22 \\
  & 2.87 & 2.73 & 2.68 & 2.95 & 2.74 & 2.68 & 3.01 & 2.74 & 2.68 &  8622 & 41 \\
\hline \hline
\end{tabular}
\end{table}

\begin{table}[!htp]
  \caption{Critical values for the single change-point scan statistic based on NNG at 0.01 significance level.  $n=1000$.}
  \label{tab:nng_01}
  \centering
\begin{tabular}{c|ccc|ccc|ccc|cc}
\hline \hline
& \multicolumn{9}{|c|}{Critical Values} & \multicolumn{2}{|c}{Graph} \\ \cline{2-12}
& \multicolumn{3}{|c|}{$n_0=100$} & \multicolumn{3}{|c|}{$n_0=50$} & \multicolumn{3}{|c|}{$n_0=25$} & & \\ \cline{2-10}
  & A1 & A2 & Per & A1 & A2 & Per & A1 & A2 & Per & $\sum |G_i|^2$ & $d_{\text{max}}$ \\ \hline \hline
  & 3.50 & 3.53 & 3.53 & 3.57 & 3.61 & 3.59 & 3.61 & 3.65 & 3.63 &  2008 &  2 \\
N(0,1) & 3.50 & 3.54 & 3.52 & 3.57 & 3.61 & 3.63 & 3.61 & 3.65 & 3.66 &  1972 &  2 \\
$d=1$ & 3.50 & 3.53 & 3.58 & 3.57 & 3.61 & 3.66 & 3.61 & 3.65 & 3.71 &  2032 &  2 \\
  & 3.50 & 3.53 & 3.56 & 3.57 & 3.61 & 3.63 & 3.61 & 3.65 & 3.68 &  2008 &  2 \\
  & 3.50 & 3.54 & 3.53 & 3.57 & 3.62 & 3.64 & 3.61 & 3.66 & 3.65 &  1954 &  2 \\ \hline
  & 3.50 & 3.54 & 3.50 & 3.57 & 3.62 & 3.61 & 3.61 & 3.66 & 3.64 &  1948 &  2 \\
Exp(1) & 3.50 & 3.53 & 3.57 & 3.57 & 3.61 & 3.63 & 3.61 & 3.65 & 3.65 &  2038 &  2 \\
$d=1$ & 3.50 & 3.54 & 3.52 & 3.57 & 3.61 & 3.63 & 3.61 & 3.66 & 3.66 &  2014 &  2 \\
  & 3.50 & 3.53 & 3.60 & 3.57 & 3.61 & 3.66 & 3.61 & 3.65 & 3.71 &  2008 &  2 \\
  & 3.50 & 3.54 & 3.54 & 3.57 & 3.62 & 3.58 & 3.61 & 3.66 & 3.66 &  2038 &  2 \\ \hline
  & 3.48 & 3.45 & 3.46 & 3.55 & 3.49 & 3.48 & 3.59 & 3.50 & 3.49 &  3370 &  6 \\
N(0,1) & 3.48 & 3.44 & 3.47 & 3.54 & 3.48 & 3.48 & 3.59 & 3.49 & 3.48 &  3502 &  6 \\
$d=10$ & 3.48 & 3.44 & 3.42 & 3.54 & 3.47 & 3.45 & 3.58 & 3.48 & 3.46 &  3444 &  7 \\
  & 3.48 & 3.44 & 3.43 & 3.54 & 3.47 & 3.46 & 3.59 & 3.48 & 3.47 &  3436 &  6 \\
  & 3.48 & 3.44 & 3.44 & 3.55 & 3.48 & 3.48 & 3.59 & 3.49 & 3.48 &  3330 &  6 \\ \hline
  & 3.49 & 3.45 & 3.46 & 3.55 & 3.49 & 3.51 & 3.59 & 3.50 & 3.51 &  3144 &  5 \\
Exp(1) & 3.49 & 3.45 & 3.48 & 3.55 & 3.49 & 3.52 & 3.59 & 3.50 & 3.52 &  3096 &  6 \\
$d=10$ & 3.49 & 3.46 & 3.48 & 3.55 & 3.49 & 3.54 & 3.59 & 3.51 & 3.57 &  3118 &  6 \\
  & 3.49 & 3.46 & 3.41 & 3.55 & 3.50 & 3.46 & 3.59 & 3.51 & 3.46 &  3114 &  5 \\
  & 3.49 & 3.46 & 3.49 & 3.55 & 3.49 & 3.52 & 3.59 & 3.51 & 3.53 &  3152 &  6 \\ \hline
  & 3.42 & 3.13 & 3.07 & 3.49 & 3.13 & 3.07 & 3.54 & 3.13 & 3.07 &  9382 & 52 \\
N(0,1) & 3.43 & 3.21 & 3.19 & 3.50 & 3.21 & 3.19 & 3.54 & 3.21 & 3.19 &  8466 & 24 \\
$d=100$ & 3.44 & 3.25 & 3.23 & 3.50 & 3.25 & 3.23 & 3.54 & 3.25 & 3.23 &  7756 & 20 \\
  & 3.42 & 3.09 & 3.08 & 3.48 & 3.09 & 3.08 & 3.53 & 3.09 & 3.08 & 11092 & 68 \\
  & 3.42 & 3.16 & 3.16 & 3.49 & 3.16 & 3.16 & 3.54 & 3.16 & 3.16 &  9538 & 38 \\ \hline
  & 3.42 & 3.20 & 3.19 & 3.49 & 3.20 & 3.19 & 3.53 & 3.20 & 3.19 & 10222 & 34 \\
Exp(1) & 3.42 & 3.22 & 3.21 & 3.49 & 3.22 & 3.21 & 3.53 & 3.22 & 3.21 & 10390 & 37 \\
$d=100$ & 3.42 & 3.18 & 3.17 & 3.48 & 3.18 & 3.17 & 3.53 & 3.18 & 3.17 & 11574 & 35 \\
  & 3.43 & 3.23 & 3.23 & 3.49 & 3.23 & 3.23 & 3.54 & 3.23 & 3.23 &  8782 & 22 \\
  & 3.43 & 3.22 & 3.24 & 3.50 & 3.22 & 3.24 & 3.54 & 3.22 & 3.24 &  8622 & 41 \\
\hline \hline
\end{tabular}

\end{table}

\subsubsection{A Changed-Interval Alternative}
\label{sec:chang-interv-altern}

Tables \ref{tab:mst_05_2d} - \ref{tab:nng_01_2d} show the results of $p$-value approximations for the changed interval alternative.  The notation and simulation settings are identical to those for the single change-point alternative in Section \ref{sec:numerical-studies}, except that $n_0$ is replaced by $l_0$ for the smallest window size.  ($l_1$ is set to $n-l_0$.)

From the tables, conclusions similar to those for the single change-point alternative can be drawn.  The analytic approximation after skewness correction performs much better than the analytic approximation under Gaussian assumption, especially when dimension increases.  The accuracy of skew-corrected approximation does not degrade significantly with dimension.  It does well for MST- and NNG- based tests when the smallest window size to be considered is as small as 25 for both 0.05 and 0.01 significance levels, and for MDP-based test when the smallest window size is 50.

\begin{table}[!htp]
  \caption{Critical values for the changed interval scan statistic based on MST at 0.05 significance level.  $n=1000$.}
  \label{tab:mst_05_2d}
  \centering


\begin{tabular}{c|ccc|ccc|ccc|cc}
\hline \hline
& \multicolumn{9}{|c|}{Critical Values} & \multicolumn{2}{|c}{Graph} \\ \cline{2-12}
& \multicolumn{3}{|c|}{$l_0=100$} & \multicolumn{3}{|c|}{$l_0=50$} & \multicolumn{3}{|c|}{$l_0=25$} & & \\ \cline{2-10}
  & A1 & A2 & Per & A1 & A2 & Per & A1 & A2 & Per & $\sum |G_i|^2$ & $d_{\text{max}}$ \\ \hline \hline
$d=1$ & 4.08 & 4.29 & 4.24 & 4.22 & 4.76 & 4.73 & 4.33 & 5.44 & 5.77 & 4994 & 2 \\ \hline
& 3.97 & 3.89 & 3.84 & 4.07 & 3.92 & 3.89 & 4.16 & 3.93 & 3.89 & 5454 & 8 \\
N(0,1) & 3.97 & 3.91 & 3.81 & 4.07 & 3.95 & 3.85 & 4.16 & 3.97 & 3.87 & 5400 & 7 \\
$d=10$ & 3.97 & 3.90 & 3.81 & 4.07 & 3.93 & 3.90 & 4.16 & 3.94 & 3.91 & 5448 & 8 \\
& 3.97 & 3.90 & 3.91 & 4.07 & 3.94 & 3.93 & 4.16 & 3.95 & 3.94 & 5440 & 7 \\
& 3.97 & 3.89 & 3.82 & 4.07 & 3.91 & 3.85 & 4.15 & 3.93 & 3.85 & 5524 & 8 \\ \hline
& 3.99 & 3.93 & 3.86 & 4.09 & 3.97 & 3.92 & 4.17 & 3.99 & 3.95 & 5042 & 8 \\
Exp(1) & 3.99 & 3.93 & 3.84 & 4.09 & 3.96 & 3.90 & 4.17 & 4.00 & 3.92 & 5040 & 6 \\
$d=10$ & 3.99 & 3.93 & 3.85 & 4.09 & 3.97 & 3.91 & 4.17 & 4.00 & 3.93 & 5106 & 6 \\
& 3.99 & 3.93 & 3.82 & 4.09 & 3.97 & 3.87 & 4.17 & 3.99 & 3.91 & 5042 & 6 \\
& 3.99 & 3.91 & 3.94 & 4.08 & 3.95 & 3.98 & 4.17 & 3.97 & 3.98 & 5126 & 8 \\ \hline
& 3.87 & 3.51 & 3.52 & 3.98 & 3.51 & 3.52 & 4.09 & 3.51 & 3.52 & 11600 & 40 \\
N(0,1) & 3.86 & 3.49 & 3.55 & 3.98 & 3.49 & 3.55 & 4.08 & 3.49 & 3.55 & 13346 & 64 \\
$d=100$ & 3.88 & 3.57 & 3.66 & 3.99 & 3.57 & 3.66 & 4.09 & 3.57 & 3.66 & 10422 & 34 \\
& 3.88 & 3.57 & 3.58 & 3.99 & 3.57 & 3.58 & 4.09 & 3.57 & 3.58 & 10804 & 43 \\
& 3.88 & 3.56 & 3.58 & 3.99 & 3.56 & 3.58 & 4.09 & 3.56 & 3.58 & 10862 & 36 \\ \hline
& 3.88 & 3.63 & 3.59 & 3.99 & 3.63 & 3.59 & 4.09 & 3.63 & 3.59 & 10384 & 24 \\
Exp(1) & 3.87 & 3.58 & 3.49 & 3.98 & 3.58 & 3.49 & 4.09 & 3.58 & 3.49 & 11922 & 33 \\
$d=100$ & 3.88 & 3.60 & 3.63 & 3.99 & 3.60 & 3.63 & 4.09 & 3.60 & 3.63 & 11194 & 34 \\
& 3.89 & 3.63 & 3.55 & 4.00 & 3.63 & 3.55 & 4.10 & 3.63 & 3.55 & 9680 & 27 \\
& 3.88 & 3.62 & 3.60 & 3.99 & 3.62 & 3.60 & 4.09 & 3.62 & 3.60 & 10468 & 29 \\ \hline \hline
\end{tabular}

\end{table}

\begin{table}[!htp]
  \caption{Critical values for the changed interval scan statistic based on MST at 0.01 significance level.  $n=1000$.}
  \label{tab:mst_01_2d}
  \centering


\begin{tabular}{c|ccc|ccc|ccc|cc}
\hline \hline
& \multicolumn{9}{|c|}{Critical Values} & \multicolumn{2}{|c}{Graph} \\ \cline{2-12}
& \multicolumn{3}{|c|}{$l_0=100$} & \multicolumn{3}{|c|}{$l_0=50$} & \multicolumn{3}{|c|}{$l_0=25$} & & \\ \cline{2-10}
  & A1 & A2 & Per & A1 & A2 & Per & A1 & A2 & Per & $\sum |G_i|^2$ & $d_{\text{max}}$ \\ \hline \hline
$d=1$ & 4.51 & 4.78 & 4.73 & 4.63 & 5.31 & 5.30 & 4.72 & 6.08 & 6.65 & 4994 & 2 \\ \hline
& 4.42 & 4.32 & 4.31 & 4.50 & 4.33 & 4.33 & 4.58 & 4.33 & 4.33 & 5454 & 8 \\
N(0,1) & 4.42 & 4.34 & 4.22 & 4.51 & 4.36 & 4.25 & 4.58 & 4.37 & 4.25 & 5400 & 7 \\
$d=10$ & 4.42 & 4.33 & 4.20 & 4.50 & 4.51 & 4.25 & 4.58 & 4.35 & 4.29 & 5448 & 8 \\
& 4.42 & 4.34 & 4.36 & 4.50 & 4.32 & 4.36 & 4.58 & 4.36 & 4.36 & 5440 & 7 \\
& 4.42 & 4.32 & 4.31 & 4.50 & 4.33 & 4.32 & 4.57 & 4.33 & 4.32 & 5524 & 8 \\ \hline
& 4.43 & 4.36 & 4.36 & 4.52 & 4.39 & 4.36 & 4.59 & 4.39 & 4.36 & 5042 & 8 \\
Exp(1) & 4.43 & 4.36 & 4.30 & 4.52 & 4.39 & 4.36 & 4.59 & 4.40 & 4.36 & 5040 & 6 \\
$d=10$ & 4.43 & 4.36 & 4.32 & 4.52 & 4.39 & 4.38 & 4.59 & 4.40 & 4.44 & 5106 & 6 \\
& 4.43 & 4.36 & 4.27 & 4.52 & 4.39 & 4.33 & 4.59 & 4.39 & 4.33 & 5042 & 6 \\
& 4.43 & 4.35 & 4.35 & 4.52 & 4.37 & 4.35 & 4.59 & 4.37 & 4.35 & 5126 & 8 \\ \hline
& 4.34 & 3.99 & 4.28 & 4.43 & 3.99 & 4.28 & 4.52 & 3.99 & 4.28 & 11600 & 40 \\
N(0,1) & 4.33 & 3.98 & 3.95 & 4.42 & 3.98 & 3.95 & 4.51 & 3.98 & 3.95 & 13346 & 64 \\
$d=100$ & 4.34 & 4.04 & 4.12 & 4.44 & 4.04 & 4.12 & 4.52 & 4.04 & 4.12 & 10422 & 34 \\
& 4.34 & 4.05 & 4.22 & 4.43 & 4.05 & 4.22 & 4.52 & 4.05 & 4.22 & 10804 & 43 \\
& 4.34 & 4.03 & 4.00 & 4.43 & 4.03 & 4.00 & 4.52 & 4.03 & 4.00 & 10862 & 36 \\ \hline
& 4.34 & 4.10 & 3.95 & 4.44 & 4.10 & 3.95 & 4.52 & 4.10 & 3.95 & 10384 & 24 \\
Exp(1) & 4.33 & 4.05 & 3.87 & 4.43 & 4.05 & 3.87 & 4.52 & 4.05 & 3.87 & 11922 & 33 \\
$d=100$ & 4.34 & 4.08 & 4.14 & 4.43 & 4.08 & 4.14 & 4.52 & 4.08 & 4.14 & 11194 & 34 \\
& 4.35 & 4.10 & 3.86 & 4.44 & 4.10 & 3.86 & 4.53 & 4.10 & 3.86 & 9680 & 27 \\
& 4.34 & 4.08 & 4.10 & 4.44 & 4.08 & 4.10 & 4.52 & 4.08 & 4.10 & 10468 & 29 \\ \hline \hline
\end{tabular}

\end{table}

\begin{table}[!htp]
  \caption{Critical values for the changed interval scan statistic based on MDP.  $n=1000$.}
  \label{tab:mdp_2d}
  \centering
significance level = 0.05
  \begin{tabular}{c|cc|cc|cc|cc}
\hline \hline
 & & & \multicolumn{2}{|c|}{$d=1$} & \multicolumn{2}{|c|}{$d=10$} & \multicolumn{2}{|c}{$d=100$}   \\ \hline
$l_0$ & A1 & A2 & N(0,1) & Exp(1) & N(0,1) & Exp(1) & N(0,1) & Exp(1) \\ \hline \hline
100 & 4.08 & 4.38 & 4.39 & 4.46 & 4.30 & 4.29 & 4.32 & 4.32 \\ \hline
50  & 4.22 & 4.97 & 5.03 & 5.12 & 5.10 & 4.87 & 5.19 & 4.99 \\ \hline
25  & 4.33 & 5.81 & 6.31 & 6.32 & 6.14 & 6.12 & 6.60 & 6.35 \\ \hline \hline
  \end{tabular}

\bigskip

significance level = 0.01
  \begin{tabular}{c|cc|cc|cc|cc}
\hline \hline
 & & & \multicolumn{2}{|c|}{$d=1$} & \multicolumn{2}{|c|}{$d=10$} & \multicolumn{2}{|c}{$d=100$}   \\ \hline
$l_0$ & A1 & A2 & N(0,1) & Exp(1) & N(0,1) & Exp(1) & N(0,1) & Exp(1) \\ \hline \hline
100 & 4.51 & 4.90 & 4.91 & 5.13 & 4.93 & 4.92 & 5.01 & 4.91 \\ \hline
50  & 4.63 & 5.58 & 5.63 & 5.94 & 5.64 & 5.48 & 6.13 & 5.63 \\ \hline
25  & 4.72 & 6.52 & 6.91 & 6.91 & 6.91 & 6.91 & 7.12 & 6.91 \\ \hline \hline
  \end{tabular}

\end{table}

\begin{table}[!htp]
  \caption{Critical values for the changed interval scan statistic based on NNG at 0.05 significance level.  $n=1000$.}
  \label{tab:nng_05_2d}
  \centering
\begin{tabular}{c|ccc|ccc|ccc|cc}
\hline \hline
& \multicolumn{9}{|c|}{Critical Values} & \multicolumn{2}{|c}{Graph} \\ \cline{2-12}
& \multicolumn{3}{|c|}{$l_0=100$} & \multicolumn{3}{|c|}{$l_0=50$} & \multicolumn{3}{|c|}{$l_0=25$} & & \\ \cline{2-10}
  & A1 & A2 & Per & A1 & A2 & Per & A1 & A2 & Per & $\sum |G_i|^2$ & $d_{\text{max}}$ \\ \hline \hline
& 4.04 & 4.10 & 4.07 & 4.15 & 4.23 & 4.20 & 4.23 & 4.31 & 4.30 & 2026 & 2 \\
N(0,1) & 4.04 & 4.10 & 4.09 & 4.15 & 4.24 & 4.18 & 4.23 & 4.31 & 4.24 & 1942 & 2 \\
$d=1$ & 4.04 & 4.10 & 4.11 & 4.15 & 4.24 & 4.23 & 4.23 & 4.31 & 4.35 & 1948 & 2 \\
& 4.04 & 4.10 & 3.96 & 4.15 & 4.23 & 4.11 & 4.23 & 4.31 & 4.25 & 2038 & 2 \\
& 4.04 & 4.10 & 4.04 & 4.15 & 4.24 & 4.17 & 4.23 & 4.31 & 4.31 & 1960 & 2 \\ \hline
& 4.04 & 4.10 & 4.00 & 4.15 & 4.23 & 4.14 & 4.23 & 4.31 & 4.24 & 2086 & 2 \\
Exp(1) & 4.04 & 4.10 & 4.08 & 4.15 & 4.23 & 4.20 & 4.23 & 4.31 & 4.24 & 1990 & 2 \\
$d=1$ & 4.04 & 4.10 & 4.00 & 4.15 & 4.24 & 4.15 & 4.23 & 4.32 & 4.27 & 2014 & 2 \\
& 4.04 & 4.10 & 4.01 & 4.15 & 4.23 & 4.20 & 4.23 & 4.31 & 4.34 & 2080 & 2 \\
& 4.04 & 4.10 & 4.04 & 4.15 & 4.23 & 4.18 & 4.23 & 4.31 & 4.27 & 2008 & 2 \\ \hline
& 3.99 & 3.92 & 3.82 & 4.09 & 3.96 & 3.88 & 4.18 & 3.97 & 3.90 & 3558 & 6 \\
N(0,1) & 3.99 & 3.91 & 3.86 & 4.09 & 3.94 & 3.86 & 4.18 & 3.95 & 3.88 & 3508 & 6 \\
$d=10$ & 4.00 & 3.92 & 3.86 & 4.10 & 3.96 & 3.93 & 4.18 & 3.97 & 3.93 & 3394 & 6 \\
& 3.99 & 3.91 & 3.81 & 4.09 & 3.94 & 3.86 & 4.18 & 3.95 & 3.90 & 3418 & 6 \\
& 3.99 & 3.91 & 3.88 & 4.09 & 3.94 & 3.88 & 4.18 & 3.95 & 3.88 & 3450 & 6 \\ \hline
& 4.00 & 3.94 & 3.85 & 4.10 & 3.98 & 3.91 & 4.18 & 3.99 & 3.91 & 3306 & 6 \\
Exp(1) & 4.01 & 3.95 & 3.91 & 4.11 & 4.00 & 3.98 & 4.19 & 4.02 & 3.99 & 3118 & 5 \\
$d=10$ & 4.00 & 3.94 & 3.89 & 4.10 & 3.98 & 3.93 & 4.19 & 4.00 & 3.94 & 3018 & 5 \\
& 4.00 & 3.95 & 3.90 & 4.11 & 3.99 & 3.93 & 4.19 & 4.01 & 3.93 & 3014 & 5 \\
& 4.01 & 3.96 & 3.95 & 4.11 & 4.01 & 3.97 & 4.19 & 4.03 & 3.99 & 3092 & 5 \\ \hline
& 3.89 & 3.55 & 3.48 & 4.00 & 3.55 & 3.48 & 4.10 & 3.55 & 3.48 & 8240 & 30 \\
N(0,1) & 3.88 & 3.50 & 3.49 & 3.99 & 3.50 & 3.49 & 4.09 & 3.50 & 3.49 & 9360 & 33 \\
$d=100$ & 3.90 & 3.61 & 3.60 & 4.00 & 3.61 & 3.60 & 4.10 & 3.61 & 3.60 & 8482 & 18 \\
& 3.88 & 3.51 & 3.48 & 3.99 & 3.51 & 3.48 & 4.09 & 3.51 & 3.48 & 9154 & 40 \\
& 3.88 & 3.50 & 3.44 & 3.99 & 3.50 & 3.44 & 4.09 & 3.50 & 3.44 & 9392 & 39 \\ \hline
& 3.88 & 3.54 & 3.47 & 3.99 & 3.54 & 3.47 & 4.09 & 3.54 & 3.47 & 10406 & 45 \\
Exp(1) & 3.88 & 3.55 & 3.55 & 3.99 & 3.55 & 3.55 & 4.09 & 3.55 & 3.55 & 10504 & 44 \\
$d=100$ & 3.88 & 3.54 & 3.61 & 3.99 & 3.54 & 3.61 & 4.09 & 3.54 & 3.61 & 10106 & 32 \\
& 3.90 & 3.64 & 3.53 & 4.00 & 3.63 & 3.53 & 4.10 & 3.63 & 3.53 & 8666 & 22 \\
& 3.90 & 3.58 & 3.57 & 4.00 & 3.58 & 3.57 & 4.10 & 3.58 & 3.57 & 8274 & 28 \\ \hline \hline
\end{tabular}
\end{table}

\begin{table}[!htp]
  \caption{Critical values for the changed interval scan statistic based on NNG at 0.01 significance level.  $n=1000$.}
  \label{tab:nng_01_2d}
  \centering
\begin{tabular}{c|ccc|ccc|ccc|cc}
\hline \hline
& \multicolumn{9}{|c|}{Critical Values} & \multicolumn{2}{|c}{Graph} \\ \cline{2-12}
& \multicolumn{3}{|c|}{$l_0=100$} & \multicolumn{3}{|c|}{$l_0=50$} & \multicolumn{3}{|c|}{$l_0=25$} & & \\ \cline{2-10}
  & A1 & A2 & Per & A1 & A2 & Per & A1 & A2 & Per & $\sum |G_i|^2$ & $d_{\text{max}}$ \\ \hline \hline
& 4.48 & 4.55 & 4.58 & 4.57 & 4.67 & 4.65 & 4.64 & 4.73 & 4.65 & 2026 & 2 \\
N(0,1) & 4.48 & 4.56 & 4.53 & 4.57 & 4.68 & 4.71 & 4.64 & 4.74 & 4.79 & 1942 & 2 \\
$d=1$ & 4.48 & 4.56 & 4.56 & 4.57 & 4.68 & 4.72 & 4.64 & 4.74 & 4.83 & 1948 & 2 \\
& 4.48 & 4.55 & 4.45 & 4.57 & 4.67 & 4.68 & 4.64 & 4.74 & 4.69 & 2038 & 2 \\
& 4.48 & 4.56 & 4.56 & 4.57 & 4.68 & 4.66 & 4.64 & 4.74 & 4.82 & 1960 & 2 \\ \hline
& 4.48 & 4.55 & 4.49 & 4.57 & 4.67 & 4.62 & 4.64 & 4.74 & 4.68 & 2086 & 2 \\
Exp(1) & 4.48 & 4.55 & 4.49 & 4.57 & 4.67 & 4.57 & 4.64 & 4.73 & 4.57 & 1990 & 2 \\
$d=1$ & 4.48 & 4.56 & 4.49 & 4.57 & 4.68 & 4.59 & 4.64 & 4.75 & 4.60 & 2014 & 2 \\
& 4.48 & 4.55 & 4.61 & 4.57 & 4.67 & 4.65 & 4.64 & 4.74 & 4.76 & 2080 & 2 \\
& 4.48 & 4.55 & 4.60 & 4.57 & 4.67 & 4.65 & 4.64 & 4.73 & 4.78 & 2008 & 2 \\ \hline
& 4.44 & 4.35 & 4.20 & 4.52 & 4.39 & 4.25 & 4.60 & 4.37 & 4.25 & 3558 & 6 \\
N(0,1) & 4.44 & 4.34 & 4.34 & 4.52 & 4.35 & 4.38 & 4.59 & 4.35 & 4.38 & 3508 & 6 \\
$d=10$ & 4.44 & 4.35 & 4.28 & 4.52 & 4.36 & 4.33 & 4.60 & 4.37 & 4.33 & 3394 & 6 \\
& 4.44 & 4.34 & 4.30 & 4.52 & 4.36 & 4.30 & 4.59 & 4.35 & 4.30 & 3418 & 6 \\
& 4.44 & 4.34 & 4.22 & 4.52 & 4.35 & 4.22 & 4.59 & 4.35 & 4.22 & 3450 & 6 \\ \hline
& 4.44 & 4.37 & 4.31 & 4.53 & 4.43 & 4.39 & 4.60 & 4.39 & 4.39 & 3306 & 6 \\
Exp(1) & 4.45 & 4.38 & 4.39 & 4.53 & 4.42 & 4.50 & 4.60 & 4.42 & 4.50 & 3118 & 5 \\
$d=10$ & 4.45 & 4.37 & 4.31 & 4.53 & 4.72 & 4.33 & 4.60 & 4.39 & 4.38 & 3018 & 5 \\
& 4.45 & 4.38 & 4.42 & 4.53 & 4.35 & 4.45 & 4.60 & 4.41 & 4.45 & 3014 & 5 \\
& 4.45 & 4.39 & 4.46 & 4.53 & 4.43 & 4.47 & 4.61 & 4.43 & 4.47 & 3092 & 5 \\ \hline
& 4.35 & 4.02 & 3.91 & 4.44 & 4.02 & 3.91 & 4.53 & 4.02 & 3.91 & 8240 & 30 \\
N(0,1) & 4.34 & 3.97 & 3.82 & 4.44 & 3.97 & 3.82 & 4.52 & 3.97 & 3.82 & 9360 & 33 \\
$d=100$ & 4.36 & 4.07 & 3.94 & 4.45 & 4.07 & 3.94 & 4.53 & 4.07 & 3.94 & 8482 & 18 \\
& 4.35 & 3.99 & 4.06 & 4.44 & 3.99 & 4.06 & 4.52 & 3.99 & 4.06 & 9154 & 40 \\
& 4.34 & 3.98 & 3.83 & 4.44 & 3.98 & 3.83 & 4.52 & 3.98 & 3.83 & 9392 & 39 \\ \hline
& 4.34 & 4.02 & 3.87 & 4.43 & 4.02 & 3.87 & 4.52 & 4.02 & 3.87 & 10406 & 45 \\
Exp(1) & 4.34 & 4.03 & 3.99 & 4.43 & 4.03 & 3.99 & 4.52 & 4.03 & 3.99 & 10504 & 44 \\
$d=100$ & 4.34 & 4.02 & 4.22 & 4.43 & 4.02 & 4.22 & 4.52 & 4.02 & 4.22 & 10106 & 32 \\
& 4.35 & 4.10 & 3.95 & 4.44 & 4.10 & 3.95 & 4.53 & 4.10 & 3.95 & 8666 & 22 \\
& 4.36 & 4.05 & 4.02 & 4.45 & 4.05 & 4.02 & 4.53 & 4.05 & 4.02 & 8274 & 28 \\ \hline \hline
\end{tabular}

\end{table}

\subsection{Coverage Probability}
\label{sec:coverage-probability}
Here, we show extra boxplots of coverage probability for the single change-point alternative.  We see that the skewness correction helps in general.

\begin{figure}[!htp]
  \centering
  \includegraphics[width=\textwidth]{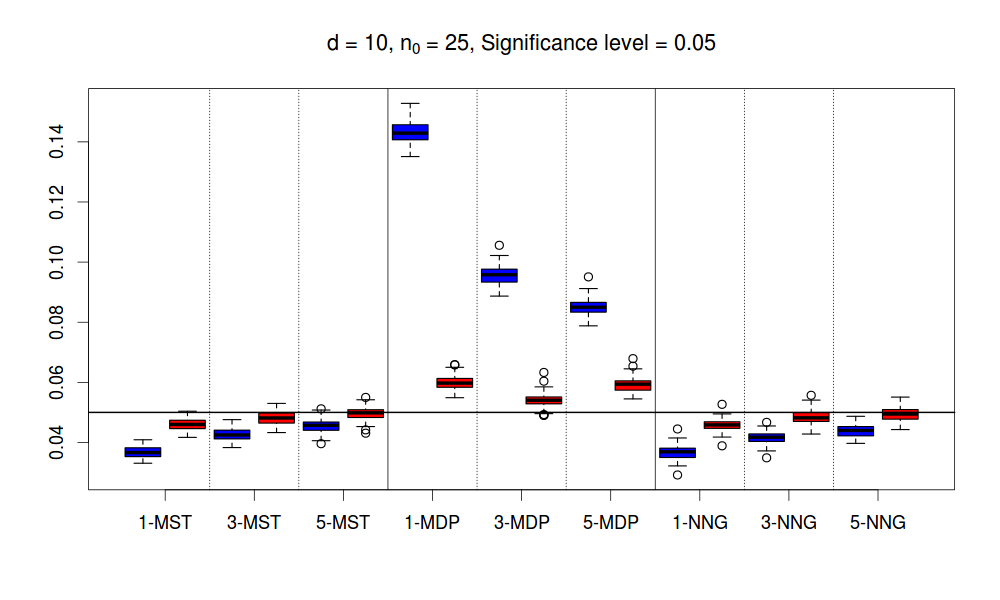}
  \includegraphics[width=\textwidth]{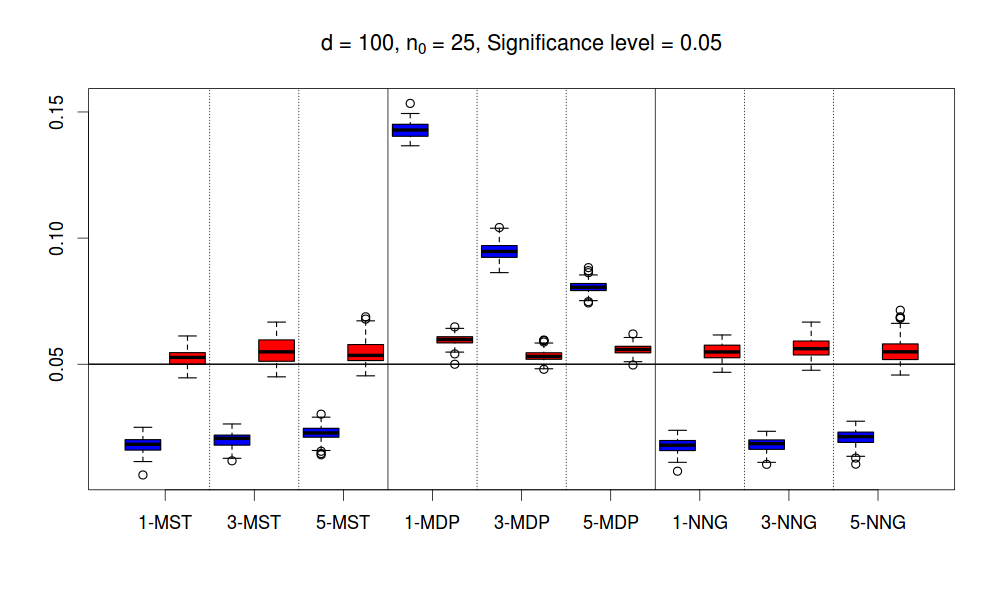}
  \caption{The counterpart boxplots of Figure \ref{fig:cov_50_05} with the smallest window size being 25 and significance level 0.05.}
  \label{fig:cov_25_05}
\end{figure}

\begin{figure}[!htp]
  \centering
  \includegraphics[width=\textwidth]{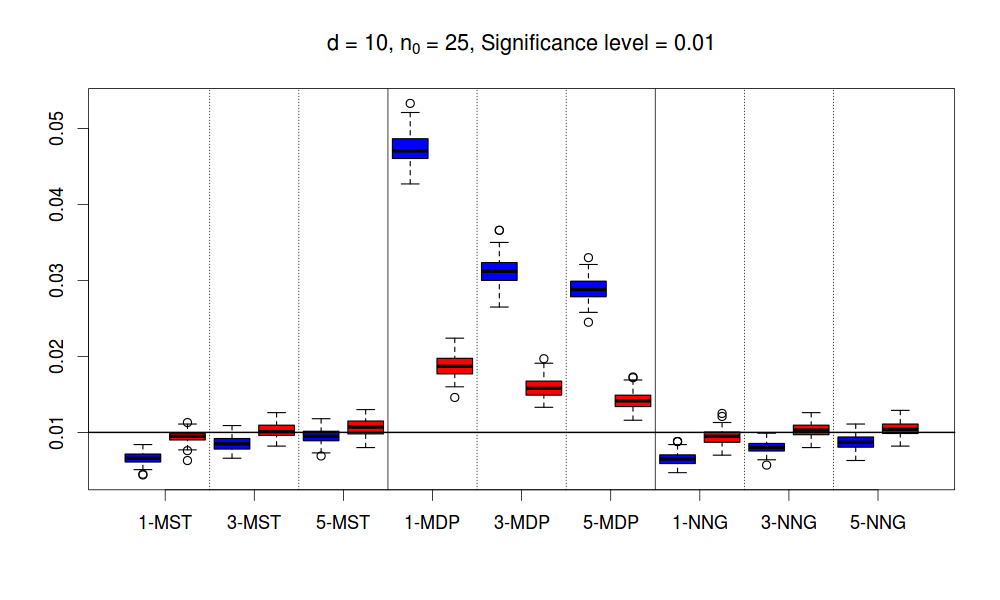}
  \includegraphics[width=\textwidth]{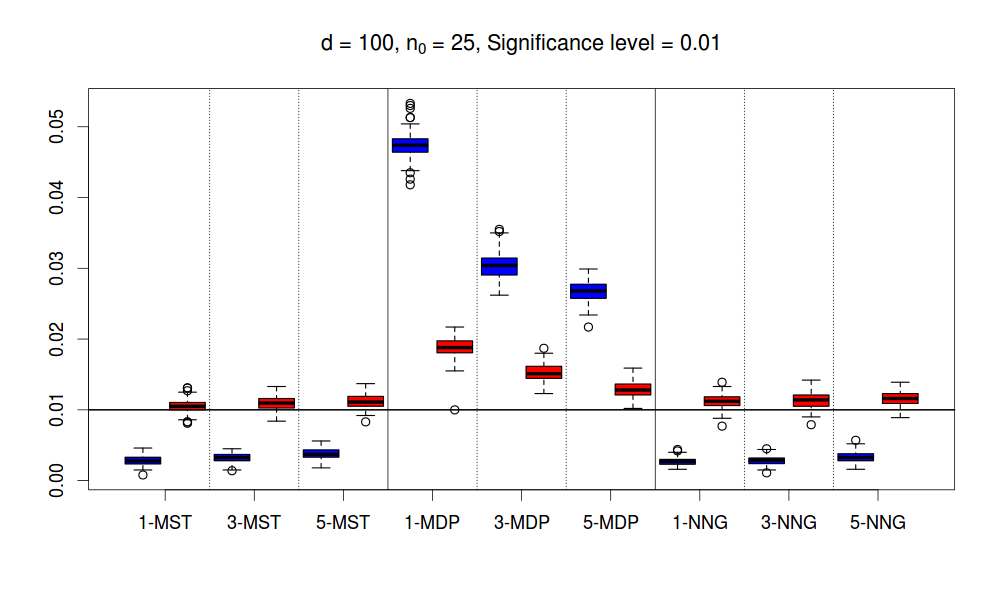}
  \caption{The counterpart boxplots of Figure \ref{fig:cov_50_05} with the smallest window size being 25 and significance level 0.01.}
  \label{fig:cov_25_01}
\end{figure}

\begin{figure}[!htp]
  \centering
  \includegraphics[width=\textwidth]{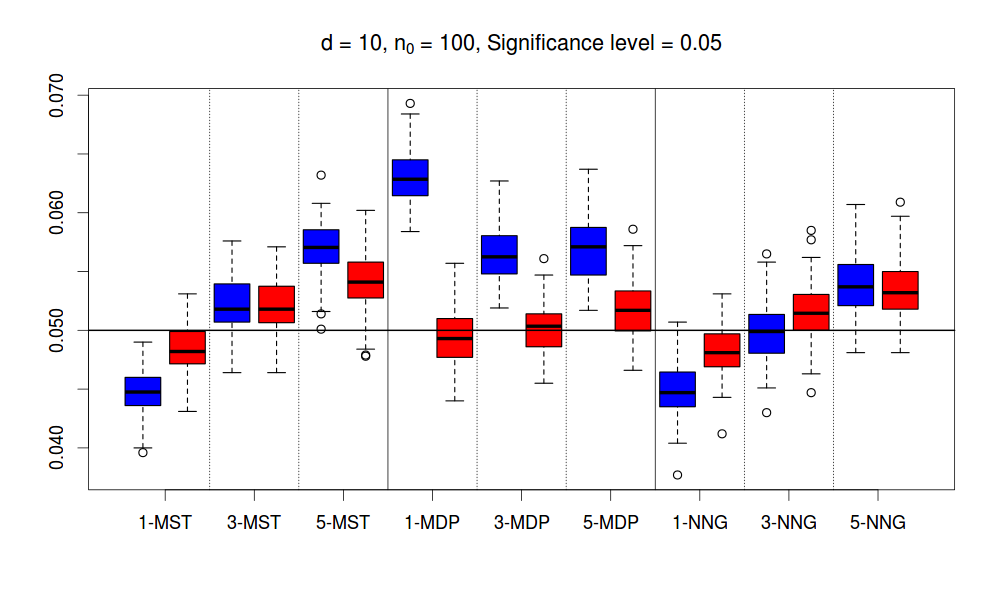}
  \includegraphics[width=\textwidth]{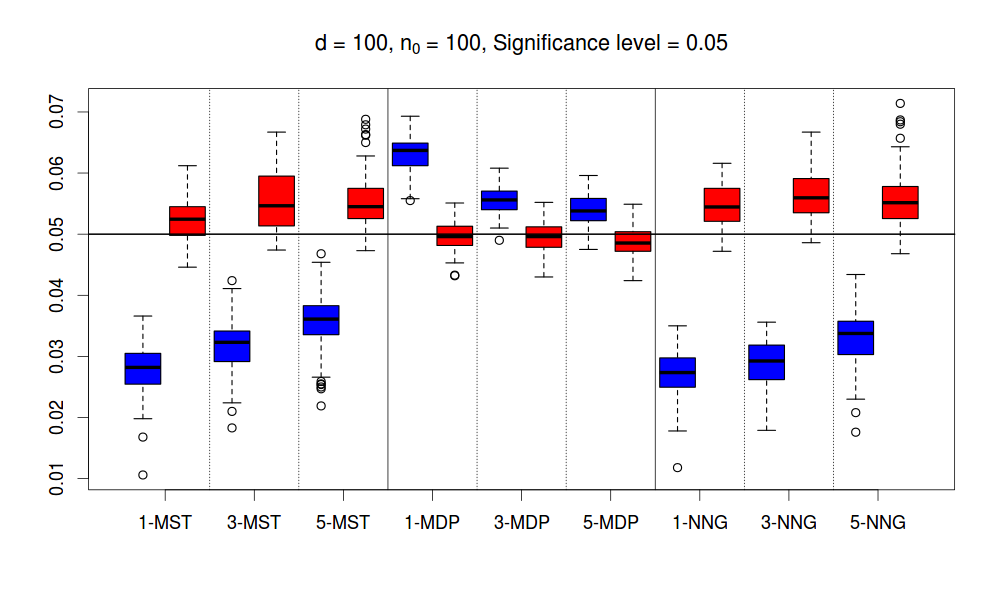}
  \caption{The counterpart boxplots of Figure \ref{fig:cov_50_05} with the smallest window size being 100 and significance level 0.05.}
  \label{fig:cov_100_05}
\end{figure}

\begin{figure}[!htp]
  \centering
  \includegraphics[width=\textwidth]{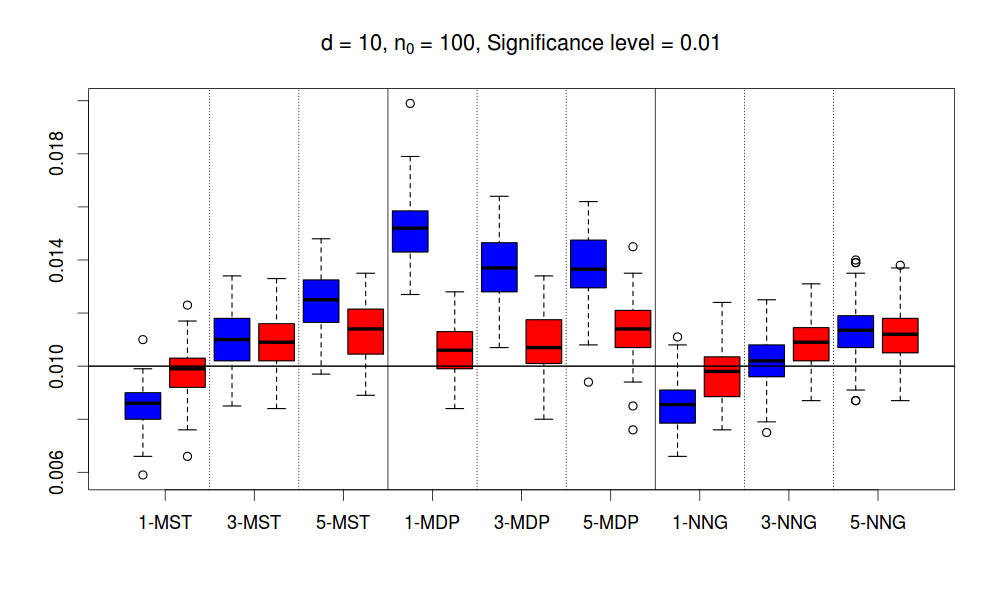}
  \includegraphics[width=\textwidth]{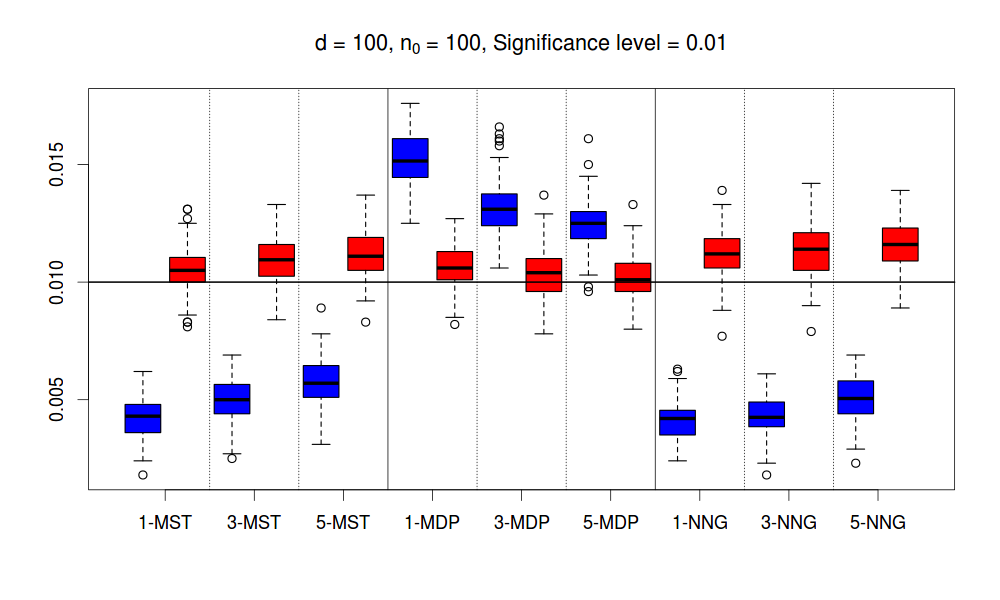}
  \caption{The counterpart boxplots of Figure \ref{fig:cov_50_05} with the smallest window size being 100 and significance level 0.01.}
  \label{fig:cov_100_01}
\end{figure}


\newpage
\section{Block Permutation Results}
\label{sec:block-perm-results}

\subsection{Authorship Data}
\label{sec:authorship-data}

\subsubsection{Scan over the Entire Book}
\label{sec:scan-over-entire}
Here, we show plots for $Z_{G,bp}$ for the authorship data under block size 2 and 10.

\begin{figure}[!htp]
  \centering
  \includegraphics[width=\textwidth]{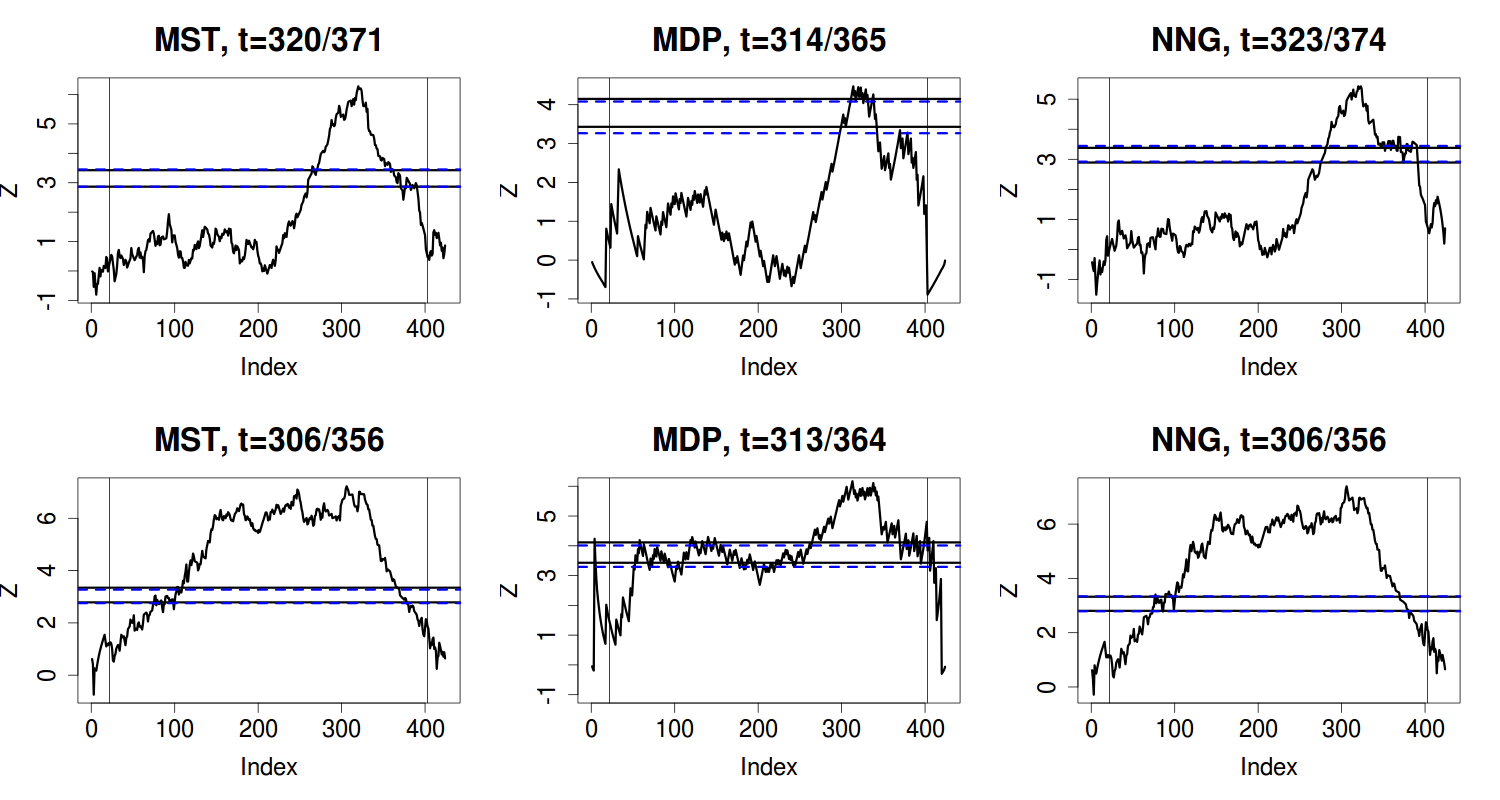}
   \caption{Block permutation results for the authorship data with block size 2.}
  \label{fig:authorbp2}
\end{figure}

\begin{figure}[!htp] 
  \centering
  \includegraphics[width=\textwidth]{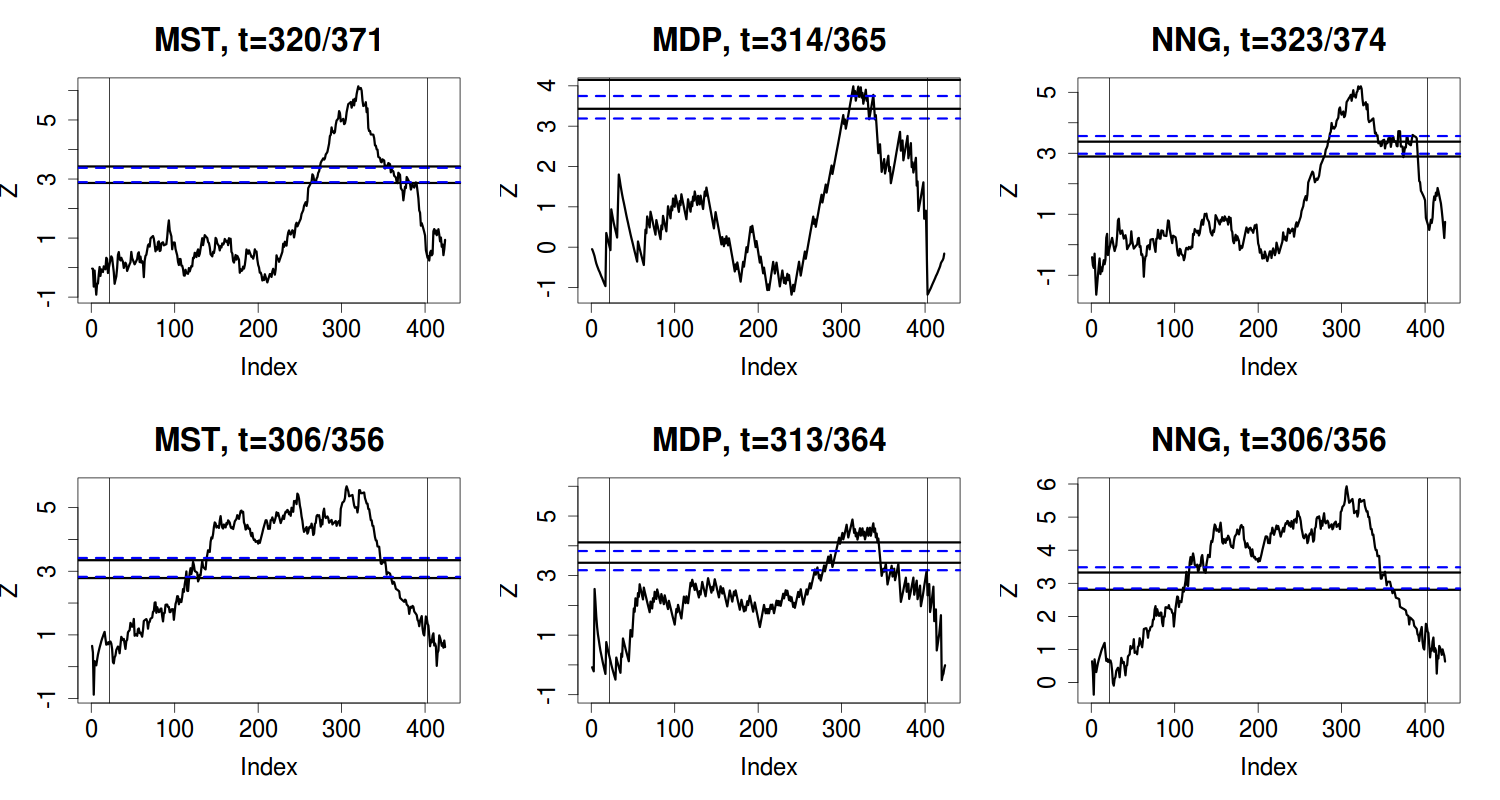}
   \caption{Block permutation results for the authorship data with block size 10.}
  \label{fig:authorbp10}
\end{figure}

\subsubsection{Scan over the First 350 Chapters}
\label{sec:scan-over-first}
Here, we show results under block permutations for the authorship data but only using data from the first 350 chapters.

\begin{table}[!htp]
  \centering
   \caption{p-values from 10,000 block permutations for the authorship data only using data from the firt 350 chapters.}
 \begin{tabular}{|c||ccc|ccc|}
    \hline
    & \multicolumn{3}{|c|}{Word length} & \multicolumn{3}{c|}{Context-free word frequency} \\ \cline{2-7}
block size   & MST & MDP & NNG & MST & MDP & NNG \\ \hline \hline
1 & 0.0485 & 0.1079 & 0.3053 & 0 & 0.0019 & 0 \\ \hline
2 & 0.0918 & 0.1345 & 0.4287 & 0 & 0.0033 & 0 \\ \hline
5 & 0.1838 & 0.1788 & 0.5490 & 0 & 0.0029 & 0 \\ \hline
10 & 0.2330 & 0.2335 & 0.6360 & 0 & 0.0127 & 0 \\ \hline
  \end{tabular}
 \label{tab:authorship2bp}
\end{table}

\begin{figure}[!htp] 
  \centering
  \includegraphics[width=\textwidth]{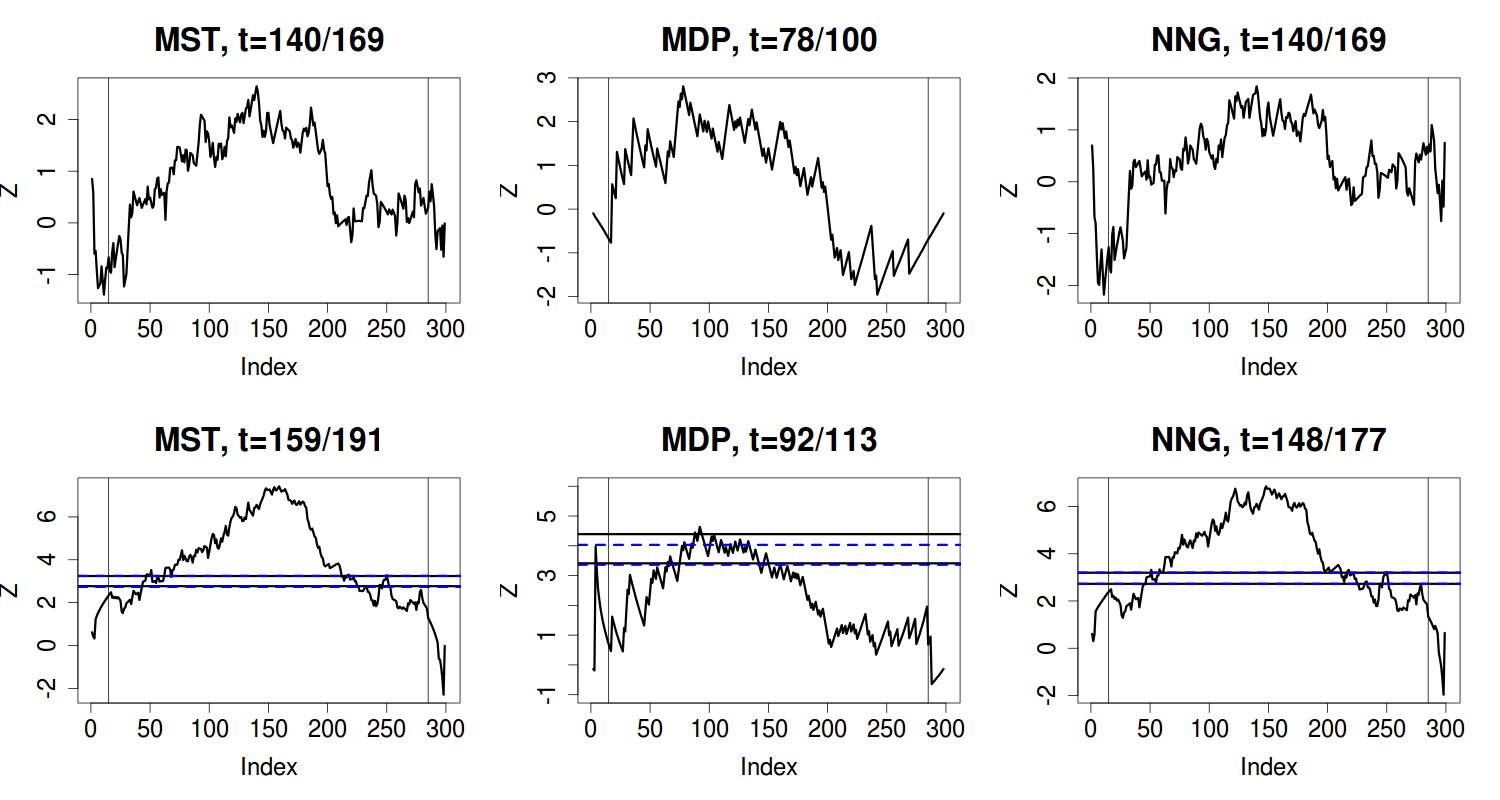}
   \caption{Block permutation results for the authorship data only using the first 350 chapters with block size 2.}
  \label{fig:author2bp2}
\end{figure}

\begin{figure}[!htp] 
  \centering
  \includegraphics[width=\textwidth]{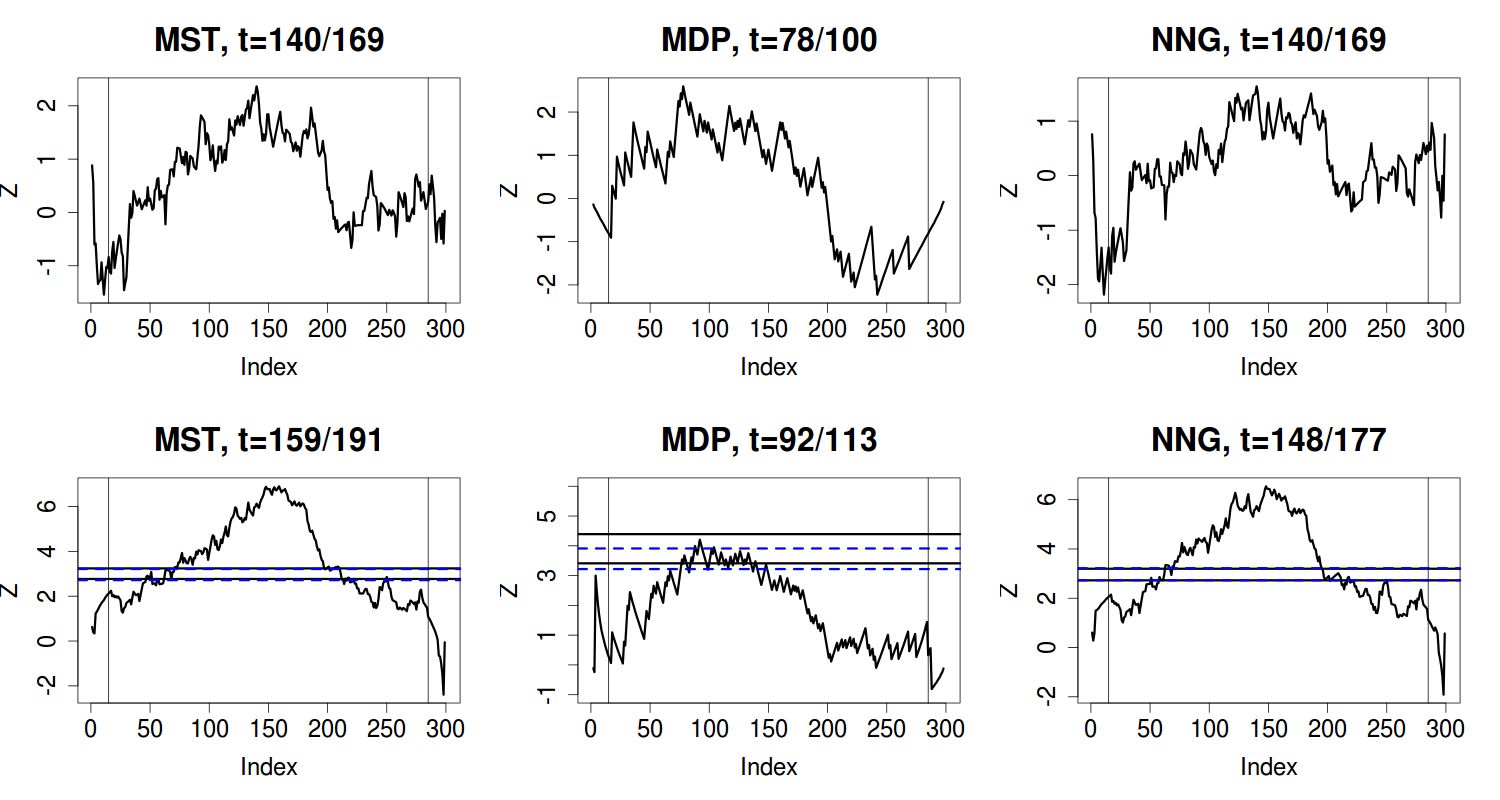}
   \caption{Block permutation results for the authorship data only using the first 350 chapters with block size 5.}
  \label{fig:author2bp5}
\end{figure}

\begin{figure}[!htp] 
  \centering
  \includegraphics[width=\textwidth]{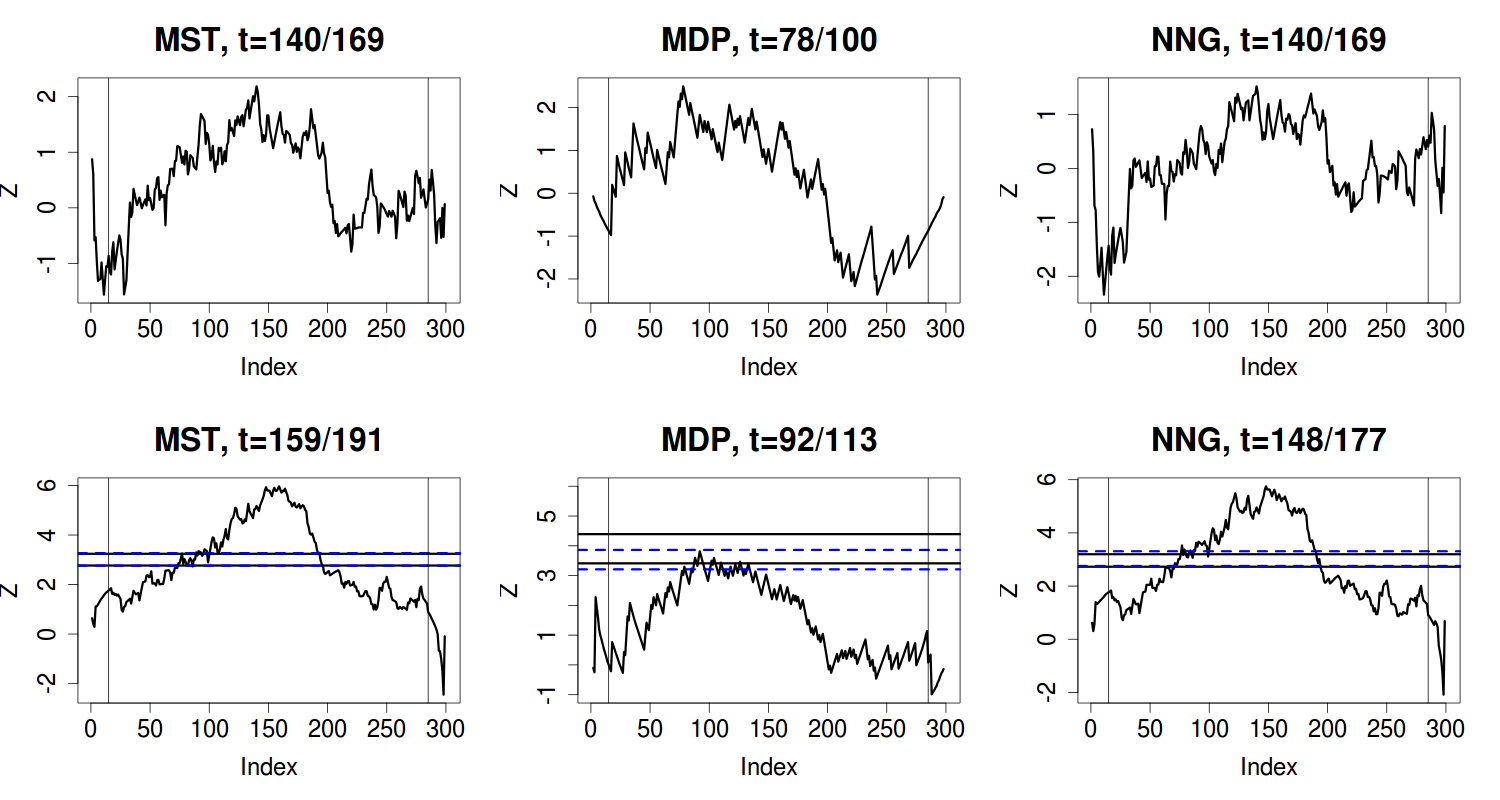}
   \caption{Block permutation results for the authorship data only using the first 350 chapters with block size 10.}
  \label{fig:author2bp10}
\end{figure}

\newpage
\subsection{Friendship Network}
\label{sec:friendship-network-1}
Here, we show results under block permutations for the phone call network data.

\begin{table}[!htp]
  \centering
   \caption{p-values from 10,000 block permutations for the phone call network data.}
 \begin{tabular}{|c||ccc|ccc|}
    \hline
    & \multicolumn{3}{|c|}{Word length} & \multicolumn{3}{c|}{Context-free word frequency} \\ \cline{2-7}
block size   & MST & MDP & NNG & MST & MDP & NNG \\ \hline \hline
1 & 0 & 0 & 0 & 0 & 0 & 0 \\ \hline
2 & 0 & 0 & 0 & 0 & 0 & 0 \\ \hline
5 & 0 & 0 & 0 & 0 & 0 & 0 \\ \hline
10 & 0 & 0 & 0.0017 & 0 & 0 & 0 \\ \hline
  \end{tabular}
 \label{tab:networkbp}
\end{table}

\begin{figure}[!htp] 
  \centering
  \includegraphics[width=\textwidth]{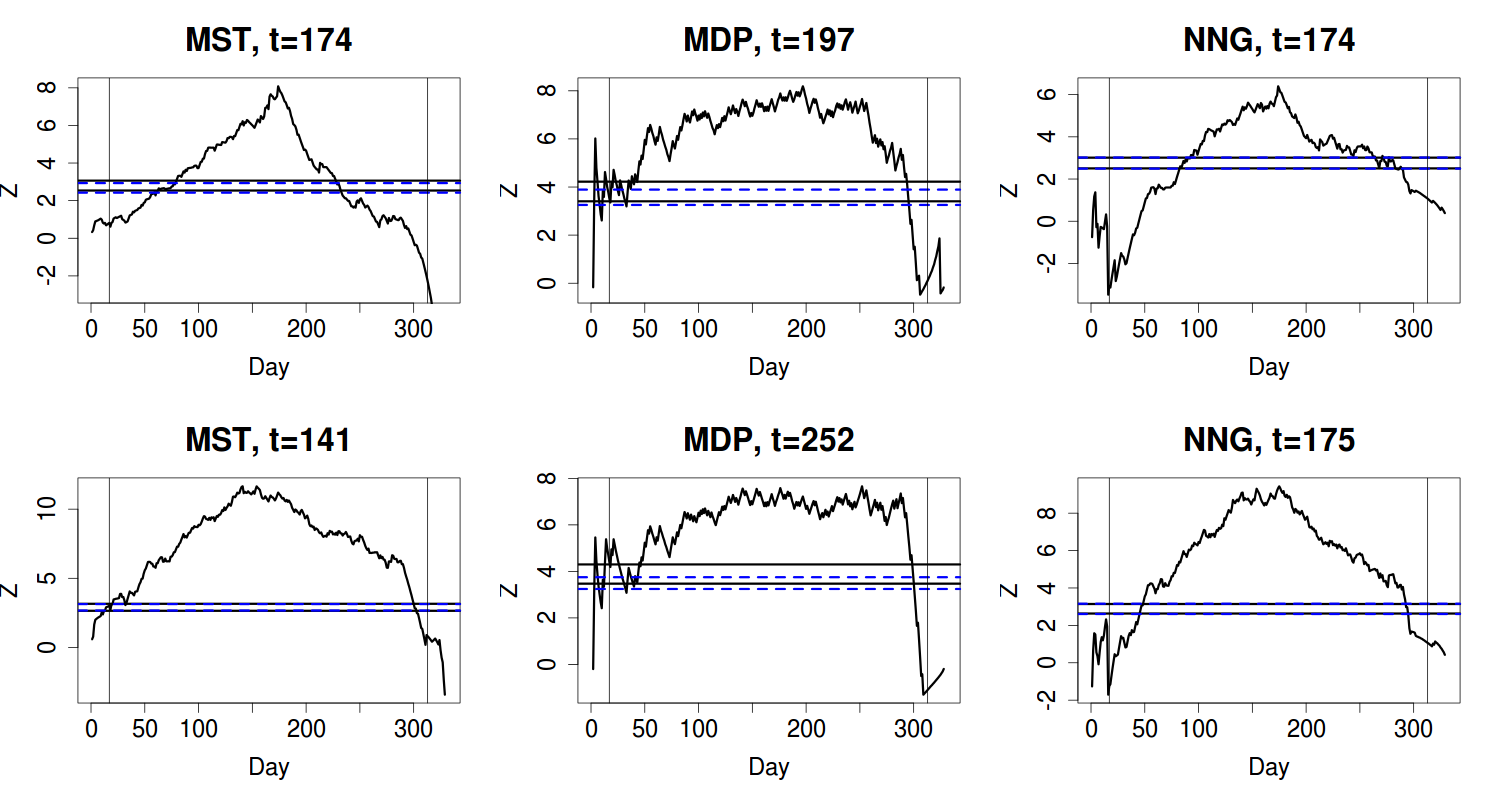}
   \caption{Block permutation results for the phone call network data with block size 2.}
  \label{fig:networkbp2}
\end{figure}

\begin{figure}[!htp] 
  \centering
  \includegraphics[width=\textwidth]{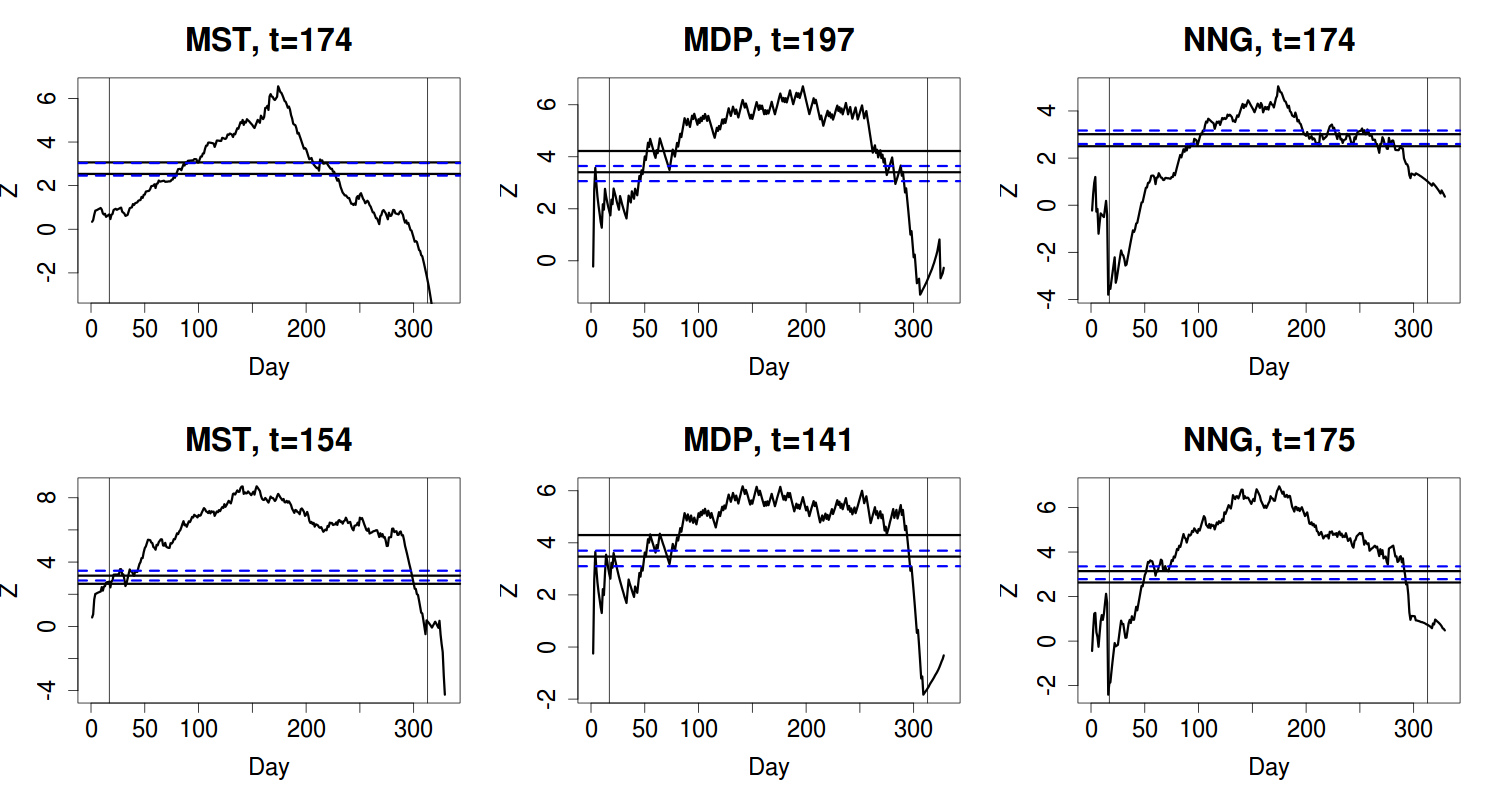}
   \caption{Block permutation results for the phone call network data with block size 5.}
  \label{fig:networkbp5}
\end{figure}

\begin{figure}[!htp] 
  \centering
  \includegraphics[width=\textwidth]{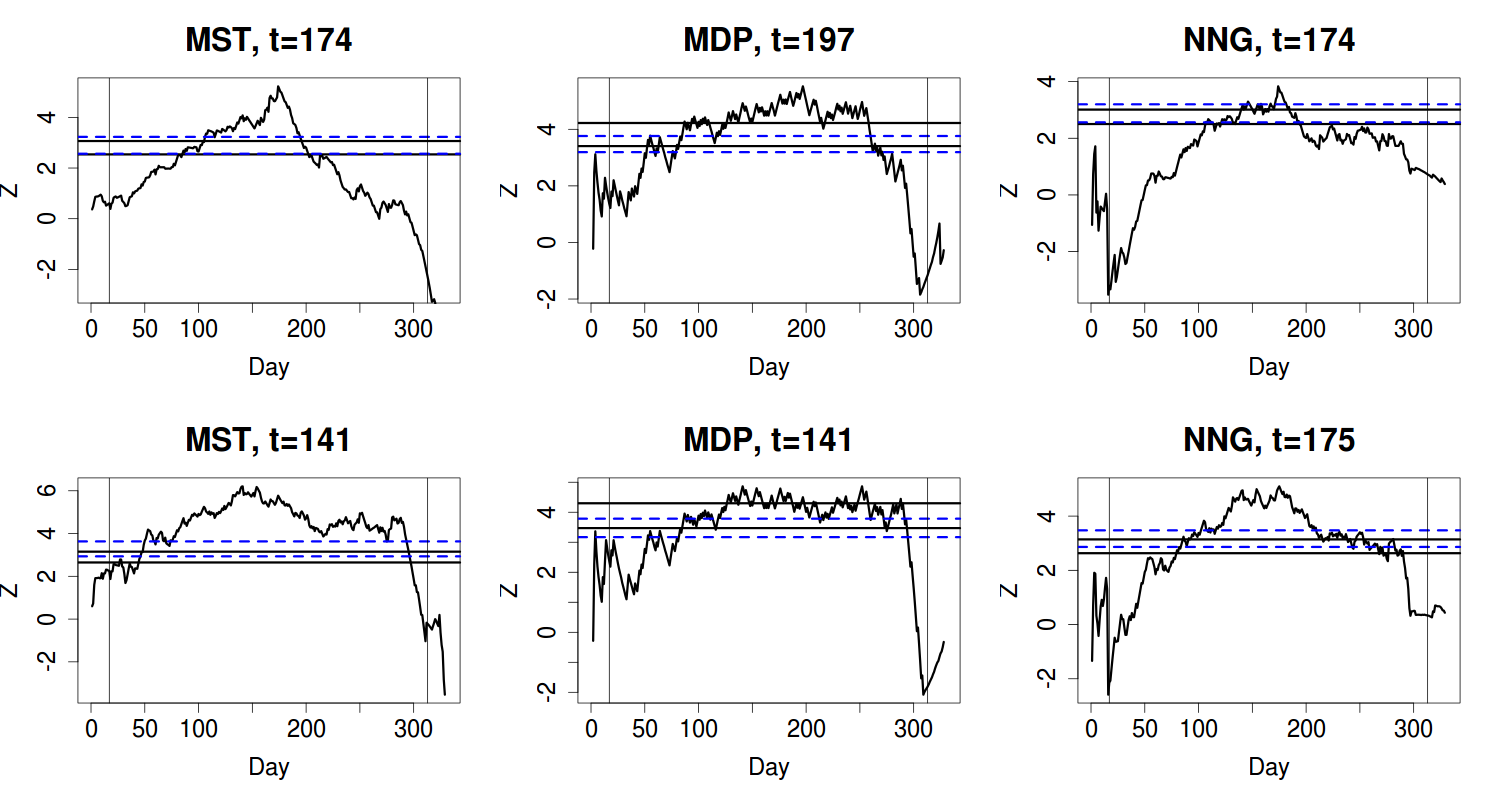}
   \caption{Block permutation results for the phone call network data with block size 10.}
  \label{fig:networkbp10}
\end{figure}


\newpage
\section*{Acknowledgements}
We thank David Siegmund, Jerome Friedman, and Susan Holmes for helpful discussions.  We also thank J. Gir\'{o}n for kindly providing the data for the analysis of Tirant lo Blanc.


\end{document}